\documentclass[twoside]{IEEEtran} 

\usepackage{booktabs}       
\usepackage{amsfonts}       
\usepackage{microtype}      

\usepackage[bookmarks=false, colorlinks=true, allcolors=blue]{hyperref}

\usepackage{pgfplots}
\usepackage{amsmath, amssymb, amsthm, algorithm}
\usepackage{algpseudocode}
\usepackage{mathtools}
\usepackage{bm}
\usepackage{epsfig,color,epstopdf,graphicx,multirow,array,bbm}

\renewcommand{\subsubsection}[1]{\noindent {\em #1. }}
\usepackage{cite}

\usepackage{tikz}
\usepackage{pgfplots}
\usepgfplotslibrary{groupplots}

\pgfplotsset{
        my stylecompare/.style={
			width=.5\linewidth,
            height=4.5cm,
            label style={font=\Large},
            title style={font=\Large},
            x tick label style={font =\small, /pgf/number format/1000 sep=},
	    axis lines=left,
	    major x tick style = transparent,
	    major y tick style = transparent,
	    every y tick label/.style={
 		   xshift=-.7cm, yshift=-2pt,anchor=south west,inner sep=0pt,font=\small
	    },
        },
        my legend style compare/.style={
            legend entries={
            		ReProCS ($37$),
            		GRASTA ($1.1$),
            		ORPCA ($3.0$),
            		Offline ReProCS ($85.0$),
            		PCP ($89.0$),
            		Alt Proj ($130.0$),
            		RPCA-GD ($470.0$),            		
            },
            legend style={
                at={(-.05,1.65)},
                anchor=north west,
            },
            legend columns=7,
	    legend style={font=\tiny},
        },
        cycle multi list={
        {red, line width=0.6pt, mark=o,mark size=2pt}, 
        {black, line width=0.6pt, mark=square,mark size=1.5pt}, 
        {blue, line width=0.6pt, mark=triangle,mark size=1.5pt},
        {red, solid, line width=0.5pt, mark=*,mark size=1.8pt}, 
        {olive, line width=0.6pt, mark=10-pointed star,mark size=1.5pt},  
        {cyan, line width=0.6pt, mark=Mercedes star,mark size=2pt}, 
        {teal, line width=0.6pt, mark=oplus,mark size=1.5pt}, 
		},
}

\pgfplotstableread[col sep = comma]{figures/final_files/error_L_train.dat}\ltdatatrain
\pgfplotstableread[col sep = comma]{figures/final_files/error_SE_train.dat}\sedatatrain

\pgfplotstableread[col sep = comma]{figures/final_files/error_L_5000_new.dat}\ltnfivethousand
\pgfplotstableread[col sep = comma]{figures/final_files/error_SE_5000_new.dat}\senfivethousand

\pgfplotstableread[col sep = comma]{figures/final_files/error_L_train_10000.dat}\ltntenthousand
\pgfplotstableread[col sep = comma]{figures/final_files/error_SE_train_10000.dat}\sententhousand

\pgfplotstableread[col sep = comma]{figures/final_files/error_L_train_bern.dat}\ltbern
\pgfplotstableread[col sep = comma]{figures/final_files/error_SE_train_bern.dat}\sebern

\pgfplotstableread[col sep = comma]{figures/final_files/error_Prot_Erot.dat}\ProtErotdata
\pgfplotstableread[col sep = comma]{figures/final_files/error_Phat_Prot.dat}\PhatProtdata
\pgfplotstableread[col sep = comma]{figures/final_files/error_Phat_P.dat}\PhatPdata

\usetikzlibrary{decorations.markings}
\usetikzlibrary{decorations.pathreplacing}
\def\MarkLt{4pt}
\def\MarkSep{2pt}

\tikzset{
  TwoMarks/.style={
    postaction={decorate,
      decoration={
        markings,
        mark=at position #1 with
          {
              \begin{scope}[xslant=0.2]
              \draw[line width=\MarkSep,white,-] (0pt,-\MarkLt) -- (0pt,\MarkLt) ;
              \draw[-] (-0.5*\MarkSep,-\MarkLt) -- (-0.5*\MarkSep,\MarkLt) ;
              \draw[-] (0.5*\MarkSep,-\MarkLt) -- (0.5*\MarkSep,\MarkLt) ;
              \end{scope}
          }
       }
    }
  },
  TwoMarks/.default={0.5},
}

\usetikzlibrary{positioning}

\usetikzlibrary{calc}

\tikzstyle{block}  = [rectangle, draw, rounded corners, text width=5cm, text centered, minimum height=1em]
\tikzstyle{smallblock}  = [rectangle, draw, rounded corners,text width=1.8cm, text centered, minimum height=1em]
\tikzstyle{input}  = [rectangle, draw, text width=1.2cm, text centered, minimum height=1em]
\tikzstyle{output}  = [rectangle, draw, text width=1.2cm, text centered, minimum height=1em]
\tikzstyle{block1}  = [rectangle, draw, rounded corners,text width=5cm, text centered, minimum height=1em]
\tikzstyle{blockl1}  = [rectangle, draw, rounded corners,text width=4cm, text centered, minimum height=1em]

\newcommand{\norm}[1]{\left\|#1\right\|}

\newtheorem{theorem}{Theorem}
\newtheorem{lem}[theorem]{Lemma}
\newtheorem{claim}[theorem]{Claim}
\newtheorem{fact}[theorem]{Fact}
\newtheorem{sigmodel}[theorem]{Model}
\newtheorem{corollary}[theorem]{Corollary}
\newtheorem{definition}[theorem]{Definition}
\newtheorem{remark}[theorem]{Remark}

\newcommand{\bi}{\begin{itemize}}
\newcommand{\ei}{\end{itemize}}
\newcommand{\ben}{\begin{enumerate}}
\newcommand{\een}{\end{enumerate}}

\newcommand{\bean}{\begin{eqnarray*} }
\newcommand{\eean}{\end{eqnarray*} }
\newcommand{\bea}{\begin{eqnarray} }
\newcommand{\eea}{\end{eqnarray} }
\newcommand{\ba}{\begin{align*} }
\newcommand{\ea}{\end{align*} }

\newcommand{\nn}{\nonumber}

\newcommand{\rest}{\mathrm{rest}}

\newcommand{\xhat}{\bm{\hat{x}}}

\newcommand{\bl}{\begin{frame}}
\newcommand{\el} {\end{frame}}

\newcommand{\cred}{\color{red}} 

\newcommand{\svdeq}{\overset{\mathrm{SVD}}=}

\newcommand{\qreq}{\overset{\mathrm{QR}}=}

\renewcommand\thetheorem{\arabic{section}.\arabic{theorem}}

\newcommand{\tmax}{d} 

\newcommand{\bbeta}{\beta}

\newcommand{\wt}{\bm{w}_t}

\newcommand{\xt}{\bm{x}_t}
\newcommand{\x}{\bm{x}}
\newcommand{\xhatt}{\hat{\bm{x}}_t}
\renewcommand{\l}{\bm{\ell}}

\newcommand{\lt}{\bm{\ell}_t}
\newcommand{\lhat}{\hat{\bm{\ell}}}

\newcommand{\lhatt}{\hat{\l}_t}   
\newcommand{\yt}{\bm{y}_t}
\newcommand{\y}{\bm{y}}
\newcommand{\tty}{\tilde{\bm{y}}}
\newcommand{\w}{\bm{w}}

\newcommand{\vt}{\bm{v}_t}

\renewcommand{\a}{\bm{a}}
\newcommand{\e}{\bm{e}} 

\newcommand{\et}{\bm{e}_t}

\newcommand{\at}{\bm{a}_t}

\newcommand{\I}{\bm{I}}

\newcommand{\Lam}{\bm{\Lambda}}

\newcommand{\T}{\mathcal{T}}
\newcommand{\J}{\mathcal{J}}

\newcommand{\D}{\bm{D}}
\newcommand{\A}{\bm{A}}

\newcommand{\Lhat}{\hat{\bm{L}}}

\renewcommand{\P}{\bm{P}}
\newcommand{\U}{\bm{U}}
\newcommand{\V}{\bm{V}}
\newcommand{\R}{\bm{R}}

\newcommand{\fx}{{\mathrm{fix}}}
\newcommand{\ch}{{\mathrm{ch}}}
\newcommand{\new}{\mathrm{new}}
\newcommand{\rot}{{\mathrm{rot}}}
\newcommand{\add}{\rot}

\newcommand{\tP}{\bm{P}}
\newcommand{\tPhat}{\hat{\tP}}

\renewcommand{\b}{\bm{b}}
\newcommand{\Phat}{\hat{\P}}

\newcommand{\Span}{\operatorname{span}} 
\newcommand{\del}{\mathrm{del}}

\newcommand{\basis}{\operatorname{basis}}
\newcommand{\rank}{\operatorname{rank}}
\newcommand{\E}{\mathbb{E}}

\newcommand{\train}{\mathrm{train}}
\newcommand{\That}{\hat{\mathcal{T}}}

\newcommand{\SE}{\mathrm{SE}}

\newcommand{\that}{{\hat{t}}}

\newcommand{\M}{\bm{M}}

\renewcommand{\L}{\bm{L}}

\newcommand{\X}{\bm{X}}
\newcommand{\Y}{\bm{Y}}

\newcommand{\bE}{\bm{E}}

\newcommand{\outfracrow}{\text{\small{max-outlier-frac-row}}}
\newcommand{\outfraccol}{\text{\small{max-outlier-frac-col}}}

\newcommand{\one}{\mathbbm{1}}

\newcommand{\matr}[2]{ \left[\begin{array}{cc}
     #1 \\
     #2
   \end{array}
  \right]
  }

\newcommand{\ed}{\mathrm{end}}
\newcommand{\termoneone}{{\mathrm{term11}}}
\newcommand{\SVD}{{SVD}}
\newcommand{\B}{\bm{B}}
\newcommand{\alphadel}{\alpha}
\renewcommand{\Re}{\mathbb{R}}

\renewcommand{\forall}{\text{ for all }}

\newcommand{\rmat}{r_{\scriptscriptstyle{L}}}
\newcommand{\xmin}{x_{\min}}
\newcommand{\xmint}{x_{\min,t}}

\newcommand{\rrow}{\rho_{\mathrm{row}}}
\newcommand{\rcol}{\rho_{\mathrm{col}}}
\newcolumntype{C}[1]{>{\centering\let\newline\\\arraybackslash\hspace{0pt}}m{#1}}

\newcommand{\st}{\bm{x}_t}
\newcommand{\Mt}{\bm{M}_t}
\newcommand{\mot}{\bm{M}_{1,t}}
\newcommand{\mtt}{\bm{M}_{2,t}}
\newcommand{\ep}{\mathbb{E}}
\newcommand{\atf}{\bm{a}_{t,\fx}}
\newcommand{\atr}{\bm{a}_{t,\ch}}

\newcommand{\enew}{\bm{E}_{\rot}}
\newcommand{\enperp}{\bm{E}_{\rot,\perp}}

\newcommand{\anew}{\bm{A}}
\newcommand{\anperp}{\A_\rest}
\newcommand{\Rnew}{\bm{R}_{\rot}}
\newcommand{\bt}{\bm{b}_t}
\newcommand{\zt}{\bm{Z}_t}

\newcommand{\pa}{\pt_{\rot}}
\newcommand{\qfix}{q_{0}}
\newcommand{\qch}{q_{0}}

\newcommand{\qa}{q_{\rot}}

\newcommand{\rch}{r_{\ch}}
\newcommand{\rfix}{r}

\newcommand{\lfix}{\bm{\Lambda}_{\fx}}
\newcommand{\lch}{\bm{\Lambda}_{\ch}}

\newcommand{\lfp}{\lambda^{+}}
\newcommand{\lcp}{\lambda_{\ch}}

\newcommand{\lcm}{\lambda_{\ch}}

\newcommand{\zz} {\tilde{\varepsilon}} 

\newcommand{\Tt}{\mathcal{T}_t}

\newcommand{\Thatt}{\hat{\mathcal{T}}_t}

\newcommand{\znkop}{\zeta_{\rot,k - 1}^+}

\newcommand{\phata}{\Phat_\rot}
\newcommand{\phatz}{\Phat_{*}}

\newcommand{\shatcs}{\xhat_{t,cs}}

\newcommand{\ezero}{\mathcal{E}_0}
\newcommand{\estar}{\mathcal{E}_*}

\newcommand{\pfix}{\bm{P}_{\fx}}
\newcommand{\pch}{\bm{P}_{\ch}}
\newcommand{\pnew}{\bm{P}_{\new}}

\newcommand{\pt}{\P}
\newcommand{\shatt}{\hat{\bm{x}}_t}
\newcommand{\phatt}{\hat{\pt}}

\newcommand{\itt}{\bm{I}_{\Tt}}
\newcommand{\bphi}{\bm{\Phi}}
\newcommand{\bpsi}{\bm{\Psi}}

\newcommand{\lthres}{\omega_{evals}} 

\newcommand{\tildej}{j}

\begin{document}
\title{Provable Dynamic Robust PCA  \\ or Robust Subspace Tracking}
\author{Praneeth Narayanamurthy,~\IEEEmembership{Student Member,~IEEE,} and Namrata Vaswani,~\IEEEmembership{Senior Member,~IEEE,}%
\thanks{A short version of this paper was presented at the IEEE International Symposium on Information Theory, 2018 \cite{rrpcp_isit18}}%
\thanks{The authors are with Department of Electrical and Computer Engineering, Iowa State University, Ames,
IA, 50010 USA (e-mail: \texttt{\{pkurpadn, namrata\} @iastate.edu}).}%
}
\maketitle
\noindent

%
%
%
%

\begin{abstract}
Dynamic robust PCA refers to the dynamic (time-varying) extension of robust PCA (RPCA). It assumes that the true (uncorrupted) data lies in a low-dimensional subspace that can change with time, albeit slowly. The goal is to track this changing subspace over time in the presence of sparse outliers. We develop and study a novel algorithm, that we call simple-ReProCS, based on the recently introduced Recursive Projected Compressive Sensing (ReProCS) framework. Our work provides the first guarantee for dynamic RPCA that holds under weakened versions of standard RPCA assumptions, slow subspace change and a lower bound assumption on most outlier magnitudes. Our result is significant because (i) it removes the strong assumptions needed by the two previous complete guarantees for ReProCS-based algorithms; (ii) it shows that it is possible to achieve significantly improved outlier tolerance, compared with all existing RPCA or dynamic RPCA solutions by exploiting the above two simple extra assumptions; and (iii) it proves that simple-ReProCS is online (after initialization), fast, and, has near-optimal memory complexity.
\end{abstract}


\section{Introduction}
Principal Components Analysis (PCA) is a widely used dimension reduction technique in a variety of scientific applications.  Given a set of  data vectors, PCA tries to finds a smaller dimensional subspace that best approximates a given dataset. According to its modern definition \cite{rpca}, robust PCA (RPCA) is the problem of decomposing a given data matrix into the sum of a low-rank matrix (true data) and a sparse matrix (outliers). The column space of the low-rank matrix then gives the desired principal subspace (PCA solution).
In recent years, the RPCA problem has been extensively studied, e.g., \cite{rpca,rpca2,rpca_zhang,rrpcp_perf,robpca_nonconvex,rrpcp_aistats,rpca_gd}. A common application of RPCA is in video analytics in separating video into a slow-changing background image sequence (modeled as a low-rank matrix) and a foreground image sequence consisting of moving objects or people (sparse) \cite{rpca}. 
Dynamic RPCA refers to the dynamic (time-varying) extension of RPCA \cite{rrpcp_perf,rrpcp_tsp,rrpcp_aistats}. It assumes that the true (uncorrupted) data lies in a low-dimensional subspace that can change with time, albeit slowly. This is a more appropriate model for long data sequences, e.g., surveillance videos. The goal is to track this changing subspace over time in the presence of sparse outliers. Hence this problem can also be referred to as {\em robust subspace tracking.}
%

\subsection{Notation and Problem Setting}
\subsubsection{Notation}
We use bold lower case letters to denote vectors, bold upper case letters to denote matrices, and calligraphic letters to denote sets or events. We use the interval notation $[a, b]$ to mean all of the integers between $a$ and $b$, inclusive, and $[a,b):= [a,b-1]$.
We will often use $\J$ to denote a time interval and $\J^\alpha$ to denote a time interval of length $\alpha$.
We use $\one_{S}$ to denote the indicator function for statement $S$, i.e. $\one_{S}=1$ if $S$ holds and $\one_{S}=0$ otherwise.
We use $\|\cdot \|$ without a subscript to denote the $l_2$ norm of a vector or the induced $l_2$ norm of a matrix.
For other $l_p$ norms, we use $\|\cdot \|_p$.
For a set $\T$, we use $\I_{\T}$ to refer to an $n \times |\T|$ matrix of columns of the identity matrix indexed by entries in $\T$. 
For a matrix $\bm{A}$, $\bm{A}'$ denotes its transpose and $\bm{A}_{\T} := \bm{AI}_{\T}$ is the sub-matrix of $\bm{A}$ that contains the columns of $\A$ indexed by entries in $\T$.
Also, we use $\bm{A}^i$ to denote its $i$-th row.
We use $\lambda_{\min}(.)$ ($\sigma_{\min}(.)$) to denote the minimum eigen (singular) value of a matrix. Similarly for $\lambda_{\max}(.)$ and $\sigma_{\max}(.)$.
We use $\delta_s(\bm{A})$ to denote the $s$-restricted isometry constant (RIC) \cite{candes_rip} of $\bm{A}$.

A matrix with mutually orthonormal columns is referred to as a {\em basis} matrix and is used to represent the subspace spanned by its columns. For basis matrices $\Phat$, $\P$, we use
\[
\SE(\Phat,\P) := \|(\I - \Phat \Phat')\P\|
\]
to quantify the {\em subspace error (SE)} between their respective column spans. This measures the sine of the maximum principal angle between the subspaces.  When $\Phat$ and $\P$ are of the same size, then $\SE(.)$ is symmetric, i.e., $\SE(\Phat,\P) = \SE(\P,\Phat)$. We use $\P_\perp$ to denote a basis matrix for the orthogonal complement of $\Span(\P)$.

For a matrix $\M$, we use $\basis(\M)$ to denote a basis matrix whose columns span the same subspace as the columns of $\M$.

The letters $c$ and $C$ denote different numerical constants in each use; $c$ is used for constants less than one and $C$ for those equal to or greater than one.

\subsubsection{Dynamic RPCA or Robust Subspace Tracking Problem Statement}
At each time $t$, we observe $\yt \in \Re^n$ that satisfies
\bea
\yt := \lt + \x_t + \vt, \text{ for } t = 1, 2, \dots, \tmax
\label{orpca_eq}
\eea
where $\xt$ is the sparse outlier vector, $\lt$ is the true data vector that lies in a fixed or slowly changing low-dimensional subspace of $\Re^n$, and $\vt$ is small unstructured noise or modeling error. To be precise, $\lt = \tP_{(t)} \a_t$ where $\tP_{(t)}$ is an $n \times r$ {\em basis matrix} with $r \ll n$ and with $\|(\I - \tP_{(t-1)}\tP_{(t-1)}{}')\tP_{(t)}\|$ small compared to $\|\tP_{(t)}\|=1$ (slow subspace change). 
We use $\T_t$ to denote the support set of $\xt$ and we let $s:=\max_t|\T_t|$.
Given an initial subspace estimate, $\Phat_{0}$, the goal is to track $\Span(\tP_{(t)})$ within a short delay of each subspace change. The initial estimate can be obtained by applying any static (batch) RPCA technique, e.g., PCP \cite{rpca} or AltProj \cite{robpca_nonconvex}, to the first $t_\train$ data frames, $\Y_{[1,t_\train]}$. 
A by-product of our solution approach is that the true data vectors $\lt$, the sparse outliers $\xt$, and their support sets $\T_t$ can also be tracked on-the-fly. In many practical applications, in fact, $\xt$ or $\T_t$ is often the quantity of interest.

We also assume that (i) $|\T_t|/n$ is upper bounded, (ii) $\T_t$ changes enough over time so that any one index is not part of the outlier support for too long, (iii) the columns of $\tP_{(t)}$ are dense (non-sparse), and (iv) the subspace coefficients $\at$ are element-wise bounded, mutually independent, zero mean, have identical and diagonal covariance matrices, and are independent of the outlier supports $\T_t$.
We quantify everything in Sec. \ref{main_res}.

\subsubsection{Subspace Change Assumption}
To ensure that the number of unknowns is not too many (see the discussion in Sec. \ref{why_pwconst}), we will further assume that the subspace $\Span(\tP_{(t)})$ is {\em piecewise constant} with time, i.e., 
\bea
\tP_{(t)} = \tP_{(t_j)} \ \forall t \in [t_j, t_{j+1}), \ j=0,1,\dots, J,
\label{pw_cons}
\eea
with $t_0=1$ and $t_{J+1}=\tmax$. Let $\P_j:= \tP_{(t_j)}$.
At each change time, $t_j$, the change is ``slow''.  This means two things:
\ben
\item First, at each $t_j$, only one direction can change with the rest of the subspace remaining fixed, i.e.,
\bea
\SE(\P_{j-1},\P_j) = \SE(\P_{j-1,\ch},\P_{j,\rot})
\label{one_change}
\eea
where  $\P_{j-1,\ch}$ is {\em a} direction from $\Span(\P_{j-1})$ that ``changes'' at $t_j$ and  $\P_{j,\rot}$ is its ``rotated'' version. Thus $\Span(\P_{j-1}) = \Span([\P_{j-1,\fx},\P_{j-1,\ch}])$ and $\Span(\P_{j}) = \Span([\P_{j-1,\fx},\P_{j,\rot}])$ where $\P_{j-1,\fx}$ is an $n \times (r-1)$ matrix that denotes the part of the subspace that remains ``fixed'' at $t_j$.

Of course  at different $t_j$'s, the changing directions could be different.

\item Second, the angle of change is small, i.e., for a $\Delta \ll 1$,
\bea
\SE(\P_{j-1},\P_j) = \SE(\P_{j-1,\ch},\P_{j,\rot}) \le \Delta.
\label{ss_ch_0}
\eea
\een

\subsubsection{Equivalent generative model}
With the above model,
\[
\P_{j,\new}:=  \frac{(\I - \P_{j-1,\ch} \P_{j-1,\ch}{}') \P_{j,\rot}}{\SE(\P_{j-1,\ch},\P_{j,\rot})}
\]
is the newly added direction at $t_j$, $\theta_j: = \cos^{-1}|\P_{j-1,\ch}{}' \P_{j,\rot}|$ is the angle by which $\P_{j-1,\ch}$ gets rotated out-of-plane (towards $\P_{j,\new}$ which lies in $\Span(\P_{j-1})^\perp$) to get $\P_{j,\rot}$. Without loss of generality, assume $0 \le \theta_j \le \pi/2$. Thus,
\bi
\item $|\sin \theta_j|= \sin \theta_j = \SE(\P_{j-1,\ch},\P_{j,\rot}) = \SE(\P_{j-1},\P_j) \le \Delta$, and
\item $\P_{j,\del}:=\P_{j-1,\ch}\sin\theta_j - \P_{j,\new} \cos \theta_j$ is the direction that got deleted at $t_j$.
\ei
We have the following equivalent generative model for getting $\P_j$ from $\P_{j-1}$: let $\U_j$ be an $r \times r$ rotation matrix,
\begin{align}
& \P_j = [\underbrace{(\P_{j-1} \U_j)_{[1,r-1]}}_{\P_{j-1,\fx}}, \P_{j,\rot}], \text{ where } \nn \\ 
& \P_{j,\rot}: = \underbrace{ (\P_{j-1} \U_j)_{r}}_{\P_{j-1,\ch}} \cos \theta_j +  \P_{j,\new} \sin \theta_j
\label{ss_ch}
\end{align}
For a simple example of this in 3D ($n=3$), see Fig. \ref{subs_ch_fig}. 

To make our notation easy to remember, we try to explain its meaning better. Consider the change at $t_j$. The direction from $\Span(\P_{j-1})$ that changes is denoted by $\P_{j-1,\ch}$. This changes by getting rotated (out-of-plane) by a small angle $\theta_j$ towards a new out-of-plane direction $\P_{j,\new}$ to get the changed/rotated direction $\P_{j,\rot}$.
Here ``plane'' refers to the hyperplane $\Span(\P_{j-1})$. The basis for the $r-1$-dimensional subspace of $\Span(\P_{j-1})$ that does not change at $t_j$ is $\P_{j-1,\fx}$. So $\P_j = [\P_{j-1,\fx}, \P_{j,\rot}]$.

The span of left singular vectors of $\L$ is contained in, or equal to, $\Span([\P_0,\P_{1,\new},\P_{2,\new},\dots,\P_{J,\new}])$. Equality holds if $\P_{j,\new}$ is orthogonal to $\Span([\P_0, \P_{1,\new}, \dots, \P_{j-1,\new}])$ for each $j$.

In this work we have assumed the simplest possible model on subspace change where, at a change time, only one direction can change. Observe though that, at different change times, the changing direction could be different and hence, over a long period of time, the entire subspace could change. This simple model can be generalized to $r_\ch>1$ directions changing; see the last appendix in the ArXiv posting of this work. It is also possible to study the most general case where $r_\ch=r$ and hence no model is assumed for subspace change (only a bound on the maximum principal angle of the change). This requires significant changes to both the algorithm and the guarantee; it is studied in follow-up work \cite{rrpcp_icml_trans_it}.



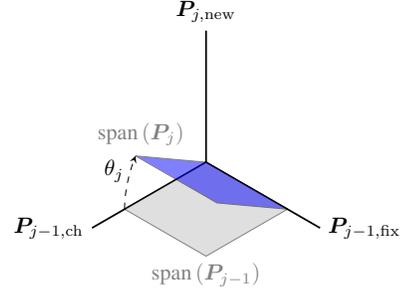
\begin{figure}[t!]
\begin{center}
\resizebox{0.3\textwidth}{!}{%
\begin{tikzpicture}[x={(0.866cm,-0.5cm)}, y={(-0.866cm,-0.5cm)}, z={(0cm,1cm)}, scale = 0.15,
                    color = {black}, axis/.style={thick},]

                    \coordinate (O) at (0, 0, 0);
    \draw[axis] (O) -- +(14, 0,   0) node [right] {{$\P_{j-1,\fx}$}};
    \draw[axis] (O) -- +(0,  14, 0) node [left] {{$\P_{j-1,\ch}$}};
    \draw[axis] (O) -- +(0,  0,   14) node [above] {{$\P_{j,\new}$}};

    \draw[fill=lightgray,draw=black,opacity=.5,very thin,line join=round]
 		(0,0,0) --
 		(0,10,0) --
 		(10,10,0) --
 		(10,0,0) --cycle node[above] at (0, 0,-14) {$\text{span}\left(\bm{P}_{j-1}\right)$};
 		
	\draw[fill=blue,draw=black,opacity=.5,very thin,line join=round]
 		(0,0,0) --
 		(0,{10 * cos(30)},{10* sin(30)}) --
 		(10,{10 * cos(30)},{10* sin(30)}) --
 		(10,0,0) --cycle node[above] at (0, 8, 5) {$\text{span}\left(\bm{P}_{j}\right)$}; 		
 		

%
\draw[->, >=stealth, domain=0:30,variable=\x,black, dashed] plot (0, {10 * cos(\x)}, {10 * sin(\x)});
\draw [] node[above] at (0,{12 * cos(20)},{8* sin(20)}) {$\theta_j$};
\end{tikzpicture}
}
\end{center}
\vspace{-0.2in}
\caption{\footnotesize{Subspace change example in 3D with $r=2$.}}
\vspace{-0.2in}
\label{subs_ch_fig}
\end{figure}

%

\subsubsection{Relation to original RPCA}
To connect with the original RPCA problem \cite{rpca,rpca_zhang,robpca_nonconvex}, define the $n \times \tmax$ data matrix $\Y:=[\y_1, \y_2, \dots \y_{\tmax}] := \L + \X + \V$ where $\L$, $\X$, $\V$ are similarly defined. Let $\rmat$ denote the rank of $\L$ and use $\outfraccol$ and $\outfracrow$ to denote the maximum fraction of outliers per column and per row of $\Y$. RPCA results bound $\max(\outfracrow,\outfraccol)$.
For dynamic RPCA, we will define $\outfracrow$ slightly differently. It will be the maximum fraction per row of any $n \times \alpha$ sub-matrix of $\Y$ with $\alpha$ consecutive columns. Here $\alpha$ denotes the number of frames used in each subspace update. We will denote this by $\outfracrow^\alpha$ to indicate the difference. Since $\alpha$ is large enough (see \eqref{def_alphas}), the two definitions are only a little different.
The dynamic RPCA assumption of a bound on $\max_t |\T_t|/n$ is equivalent to bounding $\outfraccol$ since $\outfraccol = \max_t |\T_t|/n$. The requirement of $\T_t$'s changing enough is equivalent to a bound on $\outfracrow^\alpha$. As we explain later, the denseness assumption on the $\tP_{(t)}$'s is similar to the denseness (incoherence) of left singular vectors of $\L$ assumed by all standard RPCA solutions, while the assumptions on $\at$'s replace the right singular vectors' incoherence assumption of standard RPCA.


\subsection{Related Work and our Contributions}

\subsubsection{Related Work} We briefly mention all related work here, but provide a detailed discussion later in Sec. \ref{rel_work}.
There is very little work on other solutions for provably correct dynamic RPCA. This includes our early work on a {\em partial guarantee} (guarantee required assumptions on intermediate algorithm estimates) \cite{rrpcp_perf} and later complete correctness results \cite{rrpcp_isit15,rrpcp_aistats} for more complicated ReProCS-based algorithms. We refer to all of these as ``original-ReProCS''.
It also includes older work on  modified-PCP, which is a batch solution for RPCA with partial subspace knowledge, and which can be shown to also provably solve dynamic RPCA in a piecewise batch fashion \cite{zhan_pcp_jp}.
The original-ReProCS guarantees require strong assumptions on how the outlier support changes (need a very specific model inspired by a video moving object); their subspace change assumptions are unrealistic; and their subspace tracking delay (equal to the required delay between subspace change times) is very large. On the other hand, the Modified-PCP guarantee \cite{zhan_pcp_jp} requires the outlier support to be uniformly randomly generated (strong assumption; for video, it means that the moving objects need to be single pixel wide and should be jumping around randomly from frame to frame); requires a different stronger assumption on subspace change; and cannot detect subspace change automatically.
Other than the above, there is some work on online algorithms for RPCA. The only work that comes with some guarantee, although it is a {\em partial guarantee}, is an online stochastic optimization based solver for the PCP convex program (ORPCA) \cite{xu_nips2013_1}. Its guarantee assumed that the basis matrix for the subspace estimate at each $t$ was full rank. To our best knowledge, there is no follow-up work on a complete correctness result for it.
There is also much work on empirical online solutions for RPCA, e.g., \cite{grass_undersampled}, and older work, e.g., \cite{Li03anintegrated,ipca_weightedand}. From a practical standpoint, any online algorithm will implicitly also provide a tracking solution. However, as shown in Sec. \ref{sims_detail}, the solution is not as good as that of ReProCS which explicitly exploits slow subspace change.

The standard RPCA problem has been extensively studied \cite{rpca,rpca2,rpca_zhang,robpca_nonconvex,rpca_gd,rmc_gd}.  We discuss these works in detail in Sec. \ref{rel_work}. A summary is provided in Table \ref{compare_assu}. Briefly,  these either need outlier fractions in each row and each column of the observed data matrix to be $O(1/\rmat)$ (AltProj \cite{robpca_nonconvex}, GD \cite{rpca_gd}, NO-RMC \cite{rmc_gd}, PCP result of \cite{rpca2,rpca_zhang}, denoted PCP(H)) or need the outlier support to be uniformly randomly generated (PCP result of \cite{rpca}, denoted PCP(C)). Moreover, all these are batch solutions with large memory complexity $O(n\tmax)$.

\begin{table*}[t!]
\caption{\small{Comparing s-ReProCS with other RPCA solutions with complete guarantees. For simplicity, we ignore all dependence on condition numbers. In this table $\rmat$ is the rank of the entire matrix $\L$, while $r$ is the maximum  rank of any sub-matrices of consecutive columns of $\L$ of the form $\L_{ [t_j, t_{j+1}) }$ and thus $r \le \rmat$. 
We show the unrealistic assumptions in {\color{red} red}. 
}}
\vspace{-0.12in}
\begin{center}
\renewcommand*{\arraystretch}{1.15}
\resizebox{\linewidth}{!}{
\begin{tabular}{lllll}
\toprule
Algorithm & Outlier tolerance & Assumptions  & Memory, Time, &  \# params.\\
\midrule
PCP (C) \cite{rpca}  & $\outfracrow \in \mathcal{O}(1)$&  {\cred outlier support: uniform random} & Memory: $\mathcal{O}(n \tmax)$   & zero \\
(offline)        & $\outfraccol \in \mathcal{O}(1) $ &  {$\rmat \leq  {c\min(n,\tmax)}/{\log^2 n}$ }     &             Time: $\mathcal{O}(n \tmax^2 \frac{1}{\epsilon})$   & \ \\
\midrule
PCP (H) \cite{rpca_zhang}  & $\outfracrow \in \mathcal{O}(1/\rmat)$&  {} & Memory: $\mathcal{O}(n \tmax)$   & $2$ \\
(offline)        & $\outfraccol \in \mathcal{O}(1/\rmat) $ &       &             Time: $\mathcal{O}(n \tmax^2 \frac{1}{\epsilon})$   & \ \\
\midrule

AltProj \cite{robpca_nonconvex},  & $\outfracrow = O\left(1/\rmat\right)$ & & Memory: $\mathcal{O}(n \tmax)$   & 2 \\
(offline)   & $\outfraccol \in O\left( 1/\rmat\right)$      &    & Time: $\mathcal{O}(n \tmax \rmat^2 \log \frac{1}{\epsilon})$   &   \\
\midrule
RPCA-GD \cite{rpca_gd} & $\outfracrow \in \mathcal{O}(1/\rmat^{1.5})$       &  & Memory: $\mathcal{O}(n \tmax)$   & 5 \\
  (offline)       & $\outfraccol \in \mathcal{O}(1/\rmat^{1.5})$ & & Time: $\mathcal{O}(n \tmax \rmat \log \frac{1}{\epsilon})$   &  \\
\midrule
NO-RMC   \cite{rmc_gd} & $\outfracrow \in O\left(1/\rmat\right)$ & {\cred $C n \ge \tmax  \ge c n$}  & Memory: $\mathcal{O}(n \tmax)$   & 3 \\
 (offline)      & $\outfraccol  \in \mathcal{O}(1/\rmat)$            &  & Time: $\mathcal{O}(n \rmat^3 \log^2 n \log^2 \frac{1}{\epsilon})$   &  \\
\midrule
{\bf s-ReProCS} &        {$\bm{\outfracrow^\alpha  \in \mathcal{O}(1)}$} & {\bf most outlier magnitudes lower bounded} & {\bf Memory: $\mathcal{O}(nr \log n )$} & {\bf 4}   \\
({\bf online}) &       {$\bm{\outfraccol \in \mathcal{O}(1/r)}$} &    {\bf slow subspace change}     & {\bf Time: $\mathcal{O}(n \tmax r \log \frac{1}{\epsilon})$}    \\
(this work) &  & {\bf first $Cr$ samples: AltProj assumptions} & {\bf Detect delay: $2 \alpha = C r \log n$ }  \\
 & & & {\bf Tracking Delay: $K\alpha = C r \log n \log(1/\epsilon)$} & \\
\bottomrule
\end{tabular}
}
\label{compare_assu}
\end{center}
\end{table*}

\subsubsection{Contributions}
We develop a simple algorithm, termed simple-ReProCS or s-ReProCS, for provably solving the robust subspace tracking or dynamic RPCA problem described earlier. We also develop its offline extension that can be directly compared with the standard RPCA results. Simple-ReProCS is based on the ReProCS framework \cite{rrpcp_perf}. Our main contribution is the {\em first correctness guarantee} for dynamic RPCA that {\em holds under weakened versions of standard RPCA assumptions, slow subspace change, and a lower bound on most outlier magnitudes (this lower bound is proportional to the rate of subspace change)}.
We say ``weakened'' because our guarantee implies that, after initialization, s-ReProCS can tolerate an order-wise larger fraction of outliers per row than all existing approaches, without requiring the outlier support to be uniformly randomly generated or without needing any other model on support change. It allows $\outfracrow^\alpha \in O(1)$ (instead of $O(1/\rmat)$).
For the video application, this implies that it tolerates slow moving and occasionally static foreground objects  much better than other approaches. This fact is also backed up by comparisons on real videos, see Sec. \ref{sims_detail} and also see \cite{rrpcp_review}.

A second key contribution is the algorithm itself. Unlike original-ReProCS \cite{rrpcp_isit15,rrpcp_aistats}, s-ReProCS ensures that the estimated subspace dimension is bounded by $(r+1)$ at all times without needing the complicated cluster-EVD step. More importantly, s-ReProCS is provably fast and memory-efficient: its time complexity is comparable to that of SVD for vanilla PCA, and its memory complexity is near-optimal and equal to $O(nr \log n \log(1/\zz))$ where $\zz$ is the desired subspace recovery accuracy.
This is near-optimal because $nr$ is the memory needed to output an $r$-dimensional subspace estimate in $\Re^n$, and the complexity is within log factors of the optimal.
To our best knowledge, s-ReProCS is the first provably correct RPCA or dynamic RPCA solution that is {\em as fast as the best RPCA solution in terms of computational complexity without requiring the data matrix to be nearly square} and {\em has near-optimal memory complexity}.
We provide a tabular comparison of guarantees of offline s-ReProCS with other provable RPCA solutions in Table \ref{compare_assu}. We compare s-ReProCS with other online or tracking solutions for RPCA or dynamic RPCA in Table \ref{compare_assu_reprocs} (original-ReProCS, modified-PCP, follow-up work on RePrOCS-NORST \cite{rrpcp_icml,rrpcp_icml_trans_it}, ORPCA and GRASTA). 

We give a significantly shorter and simpler proof than that for the earlier guarantees for ReProCS-based methods. We do this by first separately {\em proving a result for the problem of ``correlated-PCA" or ``PCA in data-dependent noise" \cite{corpca_nips,pca_dd} with partial subspace knowledge.} This result given in Theorem \ref{thm_corpca} of Sec. \ref{corpca_sec} may also be of independent interest.

\subsection{The need for a piecewise constant model on subspace change}\label{why_pwconst}
We explain  why the piecewise-constant subspace change model is needed.
Even if the observed data were perfect (no noise/outlier/missing-data, i.e., we observed $\lt$, and all measurements were linearly independent) and the previous subspace were exactly known, in order to obtain a correct $r$-dimensional estimate\footnote{requires finding both the newly added direction, $\P_{j,\new}$, {\em and} the deleted direction, $\P_{j,\del}$} for each $\tP_{(t)}$, one would need at least $r$ samples. Of course, to just find the newly added direction $\P_{j,\new}$ and use an $(r+1)$-dimensional estimate, one sample would suffice in this ideal setting (doing this will be especially problematic if the subspace changes at each time because it will mean the estimated subspace dimension will keep growing as $r+t$ at time $t$).
Our actual setting is not this ideal one: we know the previous subspace only up to $\epsilon$ error and we observe $\yt$ which is a noisy and outlier-corrupted version of $\lt$. This is why, in our setting, more than one data samples are needed even to accurately estimate the newly added direction. Since we get only one observed data vector $\yt$ at each time, the only way to have enough data samples for estimating each subspace is to assume that $\tP_{(t)}$ is piecewise constant with time, i.e., it satisfies \eqref{pw_cons}. In fact, our required lower bound on $t_{j+1}-t_j$ is only a little more than $r$ (see Theorem \ref{thm1}), thus making our model a good approximation to slow continuous subspace change.



Furthermore, the following point should be mentioned. In the entire literature on subspace tracking (both with and without outliers, and with and without even missing data), there is no model for subspace change for which there are any provable guarantees. There is no work on provable subspace tracking with outliers (robust subspace tracking) except our own previous work which also used the piecewise constant subspace change model. The subspace tracking (ST) problem (without outliers), and with or without missing data, has been extensively studied \cite{past,past_conv,adaptivesigproc_book,grouse,petrels,local_conv_grouse}; however, all existing guarantees are asymptotic results for the statistically stationary setting of data being generated from a {\em single unknown} subspace. Moreover, most of these also make assumptions on intermediate algorithm estimates. For a longer discussion of this, please see \cite{rrpcp_proc}.%

\subsection{Paper Organization}
The proposed algorithm, simple-ReProCS, and its performance guarantees, Theorem \ref{thm1}, are given in Sec. \ref{main_res}. We discuss the related work in detail in Sec. \ref{rel_work} and explain how our guarantee compares with other provable results on RPCA or dynamic RPCA from the literature.
Sec. \ref{proof_idea} provides the main ideas that lead to the proof of Theorem \ref{thm1}.  We prove Theorem \ref{thm1} under the assumption that the subspace change times are known in Sec. \ref{proof_thm1_section}. This proof helps illustrate all the ideas of the actual proof but with minimal notation. The general proof of Theorem \ref{thm1} is given in Appendix \ref{proof_auto_thm1}. Theorem \ref{thm1} relies on a guarantee for PCA in data-dependent noise \cite{corpca_nips,pca_dd} when partial subspace knowledge is available. This result is proved in Appendix \ref{proof_thm_corpca}. We provide detailed empirical evaluation evaluation of simple-ReProCS in Sec. \ref{sims_detail}. We conclude and discuss future directions in Sec. \ref{conclude}.


\begin{table*}[t!]
\caption{\small{Comparing  s-ReProCS with  online or tracking approaches for RPCA. We show the unrealistic assumptions in {\color{red} red}. 
Here, $f$ denotes the condition number of $\Lam$, $r$ is the maximum dimension of the subspace at any time, and $\rmat$ refers to the rank of matrix $\L$. Thus $r \le \rmat$. Here, s-ReProCS-no-delete refers to Algorithm \ref{simp_reprocs_auto} without the subpace deletion step.
}}
\vspace{-0.12in}
\begin{center}
\renewcommand*{\arraystretch}{1.15}
\resizebox{\linewidth}{!}{
\begin{tabular}{llll}
\toprule
Algorithm & Outlier tolerance & Assumptions  & Memory, Time \\ 
\midrule
{orig-ReProCS} \cite{rrpcp_aistats,rrpcp_isit15} &         {$\outfracrow^\alpha  \in \mathcal{O}(1/f^2) $} &  {\cred outlier support: moving object model,}      & {Memory: $O(nr^2/{\epsilon^2} )$}  \\
(online)  &         {$\outfraccol \in O(1/\rmat) $} & {\cred unrealistic subspace change model,}    &  {Time: $\mathcal{O}(n \tmax r \log \frac{1}{\epsilon})$ }   \\
 &         \ &  {\cred changed eigenvalues small for some time,}    & Detect Delay: $2\alpha = \frac{C r^2 \log n}{\epsilon^2}$  \\
  &         \ & outlier mag. lower bounded,      &  Tracking Delay: $ K \alpha = \frac{C r^2 \log n \log(1/\epsilon)}{\epsilon^2} $  \\
  & & $x_{\min} \geq 14 [c \gamma_{\new} + \sqrt{\zz}(\sqrt{r} + \sqrt{c})]$ \\
& & where, $\gamma_{\new}$ quantifies slow subspace change \\
        &  & init data: AltProj assumptions, &  \   \\
\                           &   & {\cred $\tmax \ge C r^2/{\epsilon^2}$}    &    \\
\midrule
Modified-PCP \cite{zhan_pcp_jp} & $\outfracrow^\alpha \in \mathcal{O}(1)$ & {\cred outlier support: uniform random} & Memory: $\mathcal{O}(nr \log^2 n)$ \\
(piecewise batch) & $\outfraccol \in \mathcal{O}(1)$ & {\cred unrealistic subspace change model} & Time: $\mathcal{O}( \frac{n \tmax r \log^2 n }{\epsilon} )$      \\
&  & {$\rmat \leq {c\min(n,\tmax)}/{\log^2 n}$} & Detect delay: $\infty$  \\
\midrule
ORPCA \cite{xu_nips2013_1} & \multicolumn{2}{c}{Has a partial guarantee -- assumes algorithm estimates at each time  $t$ are full rank} & \  \\ \midrule 

GRASTA \cite{grass_undersampled} & \multicolumn{2}{c}{Has no theoretical guarantees} & \  \\ \midrule
{\bf s-ReProCS} &        {$\bm{\outfracrow^\alpha  \in \mathcal{O}(1/f^2)}$} & {\bf most outlier magnitudes lower bounded} & {\bf Memory: $\mathcal{O}(nr \log n )$}    \\
({\bf online}) &       {$\bm{\outfraccol \in \mathcal{O}(1/r)}$} &   $\xmin \geq 15C (2\zz \sqrt{r \lambda^+} + \Delta \sqrt{\lambda_\ch} )$ & {\bf Time: $\mathcal{O}(n \tmax r \log \frac{1}{\epsilon})$}    \\
(this work) &  &  {\bf slow subspace change}     & {\bf Detect delay: $2 \alpha = C r \log n$ }  \\
 & & {\bf first $Cr$ samples: AltProj assumptions}  & {\bf Tracking Delay: $K\alpha = C r \log n \log(1/\epsilon)$}  \\  \midrule
 {\bf s-ReProCS-no-delete} &        {$\bm{\outfracrow^\alpha  \in \mathcal{O}(1)}$} & {\bf most outlier magnitudes lower bounded} & {\bf Memory: $\mathcal{O}(nr \log n )$}    \\
({\bf online}) &       {$\bm{\outfraccol \in \mathcal{O}(1/\rmat)}$} &   $\xmin \geq 15C (2\zz \sqrt{r \lambda^+} + \Delta \sqrt{\lambda_\ch} )$& {\bf Time: $\mathcal{O}(n \tmax r \log \frac{1}{\epsilon})$}    \\
(this work) &  &  {\bf slow subspace change}       & {\bf Detect delay: $2 \alpha = C r \log n$ }  \\
 & & {\bf first $Cr$ samples: AltProj assumptions} & {\bf Tracking Delay: $K\alpha = C r \log n \log(1/\epsilon)$}  \\ \midrule
 {ReProCS-NORST} & {$\outfracrow  = \mathcal{O}(1/f^2) $} & outlier mag. lower bounded       & {Memory: $\mathcal{O}(n r \log n \log \frac{1}{\epsilon} )$ } \\
 \cite{rrpcp_merop,rrpcp_icml_trans_it} ({\em online})    & $\outfraccol = \mathcal{O}(1/r) $ & $\xmin \geq C_1 \sqrt{r \lambda^+} (\Delta + 2\zz)$ &  Time: $\mathcal{O}(n \tmax r \log \frac{1}{\epsilon})$    \\
 {(follow-up}      &   & slow subspace change or fixed subspace  &  Detect delay: $ C r \log n$  \\
{to this work)}   & & first $Cr$ samples: AltProj assumptions     & Tracking Delay: $ C r \log n \log (1/\epsilon)$ \\
\bottomrule
\end{tabular}
}
\label{compare_assu_reprocs}
\end{center}
\end{table*}

\section{The simple-ReProCS Algorithm and its Guarantee}\label{main_res}

\subsection{Simple-ReProCS (s-ReProCS)}
S-ReProCS proceeds as follows. The initial subspace is assumed to be accurately known (obtained using AltProj or PCP).
At time $t$, if the previous subspace estimate, $\tPhat_{(t-1)}$, is accurate enough, because of slow subspace change, projecting $\yt = \xt + \lt + \vt$ onto its orthogonal complement will nullify most of $\lt$. Moreover, $\|\vt\|$ is small (by assumption). We compute $\tty_t:= \bm\Psi  \yt$ where $\bm{\Psi} := \I - \tPhat_{(t-1)}\tPhat_{(t-1)}{}'$. Thus, $\tty_t = \bm{\Psi} \xt + \b_t$ where $\b_t:= \bm\Psi (\lt+ \vt)$ and $\| \bm{b}_t \|$ is small. Recovering $\xt$ from $\tty_t$ is thus a traditional compressive sensing (CS) / sparse recovery problem in small noise \cite{candes_rip}. This is solvable because incoherence (denseness) of $\tP_{(t)}$'s and slow subspace change implies \cite{rrpcp_perf} that $\bpsi$ satisfies the restricted isometry property \cite{candes_rip}.
We compute $\xhat_{t,cs}$ using $l_1$ minimization followed by thresholding based support estimation to get $\That_t$.
A Least Squares (LS) based debiasing step on $\That_t$ returns the final $\xhat_t$. We then estimate $\lt$ as $\lhat_t = \yt - \xhatt$. We refer to the above step as {\em Projected Compressive Sensing (CS)}. As explained in \cite{rrpcp_review,rrpcp_proc}, this can also be understood as solving a {\em Robust Regression}  problem\footnote{The above step equivalently solves for $\tilde\a,\tilde\x$ that satisfy $\yt = \Phat_{t-1}\tilde\a + \tilde\x + \b_t$ with $\tilde\x$ being sparse and $\|\b_t\|$ being small. This is the approximate robust regression problem where columns of $\Phat_{t-1}$ are the regressors/predictors, $\tilde\x$ is the sparse outliers and $\b_t$ is the small ``noise'' or model inaccuracy.}.

The $\lhatt$'s are used for the {\em Subspace Update} step which involves (i) detecting subspace change; (ii) obtaining improved estimates of the changed direction(s) by $K$ steps of projection-SVD \cite{rrpcp_perf}, each done with a new set of $\alpha$ frames of $\lhatt$; and (iii) a simple SVD based subspace re-estimation step, done with another new set of $\alpha$ frames. This is done to remove the deleted direction and get an $r$-dimensional estimate of the new subspace.
We explain the subspace change detection strategy in Sec. \ref{det_works}. Suppose the change is detected at $\that_j$. The $k$-th projection-SVD step involves computing $\Phat_{j,\rot,k}$ as the top singular vector of $(\I - \Phat_{j-1}\Phat_{j-1}{}')[\lhat_{\that_j+(k-1)\alpha},\lhat_{\that_j+(k-1)\alpha+1},\dots,\lhat_{\that_j+k\alpha-1}]$ and setting $\Phat_{(t)} = \Phat_{j,k}:=[\Phat_{j-1},\Phat_{j,\rot,k}]$.
For ease of understanding, we summarize a basic version of s-ReProCS in Algorithm \ref{simp_reprocs_tj}. This assumes that the change times $t_j$ are known, i.e., that $\that_j=t_j$. The actual algorithm that detects changes automatically is longer and is given as Algorithm \ref{simp_reprocs_auto} in Sec. \ref{det_works}. {\em We both analyze and implement this one.}%


The above approach works because, every time the subspace changes, with high probability (whp), the change can be detected within a short delay, and after that, the $K$ projection-SVD steps help get progressively improved estimates of the changed/rotated direction $\P_{j,\rot}$. 
The final simple SVD step re-estimates the entire subspace in order to delete $\P_{j,\del}$, from the estimate.

The estimates of the subspace or of $\lt$'s are improved in offline mode as follows. At $t = \that_j + K \alpha$, the $K$ projection-SVD steps are complete and hence the subspace estimate at this time is accurate enough whp. At this time, {\em offline s-ReProCS} (last line of Algorithm \ref{auto_reprocs}) goes back and sets $\tPhat_{(t)} \leftarrow [\Phat_{j-1}, \Phat_{j,\rot,K}]$ for all $t \in [\that_{j-1}+K \alpha, \that_j + K \alpha)$. It also uses this to get improved estimates of $\xhat_t$ and $\lhat_t$ for all these times $t$.




\subsection{Assumptions and Main Result} \label{assu_main}

\subsubsection{Incoherence (denseness) of columns of $\P_j$'s}
In order to separate the $\lt$'s from the sparse outliers $\xt$, we need an assumption that ensures that the $\lt$'s are themselves not sparse. One way to ensure this is to assume $\mu$-incoherence \cite{rpca} of the basis matrix for the subspace spanned by the columns of $\P_{j-1}$ and $\P_j$, i.e., assume that
\bea
\max_{j=1,2,\dots, J} \max_{i=1,2,\dots,n} \| \basis([\P_{j-1},\P_{j}])^i\| \le \sqrt{\frac{\mu (r+1)}{n}}
\label{incoh_2}
\eea
for a $\mu \ge 1$ but not too large (assumed to be a numerical constant henceforth). Because of our subspace change model, the subspace spanned by the columns of $[\P_{j-1},\P_j]$ has dimension $r+1$. In fact, $ \basis([\P_{j-1},\P_{j}]) = [\P_{j-1}, \P_{j,\new}]$. 

It is easy to see that \eqref{incoh_2}, along with the bound on $\outfraccol$ assumed in Theorem \ref{thm1} given below ($\outfraccol \le {0.01}/{(2 \mu (r+1))}$), implies that \eqref{dense_bnd} given later holds\footnote{This is true because (i) for any basis matrix $\P$, $\max_{\T: |\T| \le 2s} \|\I_\T {}' \P\|^2 \le 2s \max_i \|\I_i{}' \P\|^2$ \cite{rrpcp_perf}, here $s = \outfraccol \cdot n$;  (ii) if $\tilde{\P}$ is such that $\Span(\tilde{\P}) \subseteq \Span(\P)$, then $\|\I_\T {}' \tilde\P\|^2 \le \|\I_\T {}' \P\|^2$; and (iii) both $\Span(\P_j)$  and $\Span(\P_{j,\new})$ are contained in the span of $\basis([\P_{j-1},\P_{j}])$. In fact $\basis([\P_{j-1},\P_{j}]) = [\P_{j-1}, \P_{j,\new}]$. Thus, using the $\outfraccol$ bound, $\max_{\T: |\T| \le 2s} \|\I_\T {}' \P_j\|^2 \le 2s \mu (r+1)/n = 2 \outfraccol  \mu (r+1) \le 0.01$ and the same also holds for $\max_{\T: |\T| \le 2s} \|\I_\T {}' \P_{j,\new}\|^2$.}.
Our result actually only needs \eqref{dense_bnd}, but that is complicated to state and explain. Hence we use the above stronger but well-understood assumption.



\subsubsection{Assumption on principal subspace coefficients $\at$}
We assume that the $\a_t$'s are zero mean, mutually independent, {\em element-wise bounded} random variables (r.v.), have identical and diagonal covariance matrix denoted $\bm\Lambda$, and are independent of the outlier supports $\T_t$.
Here element-wise bounded means that there exists a numerical constant $\eta$, such that
\[
\max_{j=1,2, \dots r}  \max_t \frac{( \at)_j^2}{\lambda_j(\Lam)} \le \eta. \nn
\]
For most bounded distributions, $\eta$ is a little more than one, e.g., if the entries of $\at$ are zero mean uniform, then $\eta=3$. 
As we explain later the above assumptions of $\at$ replace the right singular vectors' incoherence assumption used by all standard RPCA solutions.

\subsubsection{Outlier fractions bounded}
Similar to earlier RPCA works, we also need outlier fractions to be bounded. However, we need different bounds on this fraction per column and per row. The row bound can be much larger\footnote{One practical application where this is useful is for slow moving or occasionally static video foreground moving objects. For a stylized example of this, see Model \ref{mod:moving_object} given in Sec. \ref{sims_detail}.}.
Since the ReProCS subspace update step operates on mini-batches of data of size $\alpha$ (i.e. on $n \times \alpha$ sub-matrices of consecutive columns), we need to bound $\outfracrow$ for each such sub-matrix. We denote this by $\outfracrow^\alpha$. 


\begin{definition}
\ 
\ben
\item For a time interval, $\J$, define
\bea
\gamma(\J): = \max_{i=1,2,\dots,n} \frac{1}{|\J|} \sum_{t \in \J} \one_{ \{i \in \T_t \} }.
\label{def_gamma_J}
\eea
Thus $\gamma(\J)$ is the maximum outlier fraction in any row of the sub-matrix $\Y_\J$ of $\Y$. Let $\J^\alpha$ denote a time interval of duration $\alpha$. Define
\bea
\outfracrow^\alpha:= \max_{\J^\alpha \subseteq [t_1, \tmax]} \gamma(\J^\alpha).
\label{def_outfracrow}
\eea

\item Define $\outfraccol:=\max_t|\T_t|/n$. 

\item Let $\xmin:=\min_t \min_{i \in \T_t} |(\xt)_i|$ denote the minimum outlier magnitude. 

\item  Use $\lambda^-$ and $\lambda^+$ to denote the minimum and maximum eigenvalues of $\Lam$ and
$
f:=\frac{\lambda^+}{\lambda^-}
$
its condition number.%

\item Split $\at$ as $\at = \begin{bmatrix}\a_{t,\fx}{} \\ \a_{t,\ch}\end{bmatrix}$ where $\a_{t,\ch}$ is the scalar coefficient corresponding to the changed direction. Similarly split its diagonal covariance matrix as $\Lam = \begin{bmatrix} \Lam_\fx & \bm{0} \\ \bm{0} & \lambda_\ch \end{bmatrix}$.


\item Let $\zz$ denote the bound on initial subspace error, i.e., let $\SE(\Phat_0, \P_0) \le \zz.$

\item For numerical constants $C$ that are re-used to denote different numerical values, define
\begin{align}\label{def_alphas}
& K := \lceil C \log(\Delta / \zz) \rceil, \text{ and } \alpha \ge \alpha_* :=C  f^2 (r \log n).  
\end{align}

\een
\end{definition}

\subsubsection{Main Result} We can now state our main result. For ease of understanding, we provide a table explaining various symbols, and assumptions required for Theorem \ref{thm1} in Table \ref{tab:not1}.

\begin{theorem}
\label{thm1}
Consider simple-ReProCS given in Algorithm \ref{auto_reprocs}. 
Assume that $\SE(\Phat_0, \P_0) \le \zz$ with $\zz f \le  0.01 \SE(\P_{j-1},\P_j)$.
\ben


\item (statistical assumptions) assumptions on $\at$'s hold;%

\item (subspace change)
\ben
\item \eqref{pw_cons}, \eqref{one_change}, and \eqref{ss_ch_0} hold with  $t_{j+1}-t_j > (K+3) \alpha$ where $K$ and $\alpha$ are defined above in \eqref{def_alphas},
\item $\Delta$ satisfies $C (2\zz \sqrt{r \lambda^+} + \Delta \sqrt{\lambda_\ch} ) < \xmin/15$ with $C = \sqrt{\eta}$;  
\een
\item (outlier fractions and left incoherence)
\ben
\item  \eqref{incoh_2} holds and $\outfraccol \le \rcol := \frac{0.01}{2\mu (r+1)}$, 
\item $\outfracrow^\alpha \le \rrow := \frac{0.01}{f^2}$;
\een

\item (noise $\vt$) $\vt$'s are zero mean, mutually independent, independent of the $\xt$'s and $\lt$'s, and satisfy
$\|\vt\|^2 \le 0.1 \zz^2 r \lambda^+$ and $\|\E[\vt \vt{}']\| \le 0.1 \zz^2 \lambda^+$;

\item (algorithm parameters) set
 $K$ and $\alpha$ as in \eqref{def_alphas}, $\xi = \xmin/15$, $\omega_{supp} =\xmin/2$, $\omega_{evals} = 5 \zz^2 f \lambda^+$; 

\een
then, with probability at least $1 - 12 \tmax n^{-12} $,  at all times, $t$,
\ben
\item  $\That_t = \T_t$,
\item $t_j \le \that_j \le t_j+2\alpha$,
\item $\SE(\tPhat_{(t)}, \tP_{(t)}) \le$ \\ 
\[
\left\{
\begin{array}{ll}
2\zz + \Delta & \text{ if }  t \in [t_j, \that_j+\alpha) \\
1.2 \zz + (0.5)^{k-2} 0.06 \Delta & \text{ if }  t \in [\that_j+(k-1)\alpha, \that_j+ k\alpha) \\
2\zz   & \text{ if }  t \in [\that_j+ K\alpha, \that_j + K \alpha + \alphadel) \\
\zz   & \text{ if }  t \in [\that_j+ K\alpha+\alpha, t_{j+1})
\end{array}
\right.
\]
\item and $\|\xhat_t-\xt\| = \|\lhat_t-\lt\| \le C (\zz \sqrt{r \lambda^+} + \SE(\tPhat_{(t)}, \tP_{(t)}) \sqrt{\lambda_\ch}$ with $\SE(\tPhat_{(t)}, \tP_{(t)})$ bounded as above.
\een
Consider offline s-ReProCS (last line of Algorithm \ref{auto_reprocs}). At all $t$,
\[
 \SE(\tPhat_{(t)}^{\mathrm{offline}} , \tP_{(t)}) \le 2 \zz, \text{ and }
 \|\lhat_t^{\mathrm{offline}} - \lt \| \le  2.4 \zz   \|\lt\|.
\]
\end{theorem}

The upper bound on $\vt$ and the lower bound on $\xmin$ can be relaxed significantly to get a more complicated result which we state in the corollary below.
\begin{corollary}
Let $\xmint:= \min_{i \in \T_t}|(\xt)_i|$ denote the minimum outlier magnitude at time $t$ and define the time intervals
\bi
\item $\J_0 = [t_j, \that_j)$ (interval before the change gets detected),
\item $\J_k := [\that_j + (k-1)\alpha, \that_j + k\alpha)$ ($k$-th subspace update uses data from this interval) for $k=1,2,3,\dots,K$,
\item and $\J_{K+1}:=[\that_j + K\alpha, \that_j + K\alpha+\alpha)$ (final SVD-based re-estimation step uses data from this interval).
\ei
All conclusions of Theorem \ref{thm1} hold if the following hold instead of assumptions 1b, 4, and 5 of Theorem \ref{thm1}:

$\vt$'s are zero mean, mutually independent, independent of the $\xt$'s and $\lt$'s, $\|\vt\| \le b_{v,t}$, $\|\E[\vt \vt{}']\| \le b_{v,t}^2/r$, $\xmint$ and $b_{v,t}$ satisfy the following:
\ben
\item for $t \in \J_0 \cup \J_1$, $b_{v,t}= C (2\zz\sqrt{r \lambda^+} +  0.11 \Delta \sqrt{\lambda_\ch})$, and $\xmint \ge 30 b_{v,t}$,
\item for $t \in \J_k$, $b_{v,t}= C (2\zz\sqrt{r \lambda^+} + 0.5^{k-2} 0.06 \Delta \sqrt{\lambda_\ch})$, and $\xmint \ge 30 b_{v,t}$, for $k=2,\dots,K$,
\item for $t \in \J_{K+1}$, $b_{v,t}= C ( \zz \sqrt{r \lambda^+})$ and $\xmint \ge 30 b_{v,t}$
\een
with $C=\sqrt{\eta}$; and we set $\omega_{supp,t} = \xmint/2$, and $\xi_t = \xmint/15$ (alternatively, one can also set $\omega_{supp,t}$ and $\xi_t$ to be proportional to $b_{v,t}$ which itself is proportional to the bound on $\|\xt-\xhat_t\|$ is each interval).
\label{gen_xmin}
\end{corollary}

\begin{proof}We explain the ideas leading to the proof in Sec. \ref{proof_idea}. Instead of first proving Theorem \ref{thm1} and then Corollary \ref{gen_xmin}, we directly only prove the latter. The proof of the former is almost the same and is immediate once the latter proof can be understood.  For notational simplicity, we first prove the results under the assumption $\that_j=t_j$ in Sec. \ref{proof_thm1_section}. The proof without assuming $\that_j=t_j$ is given in Appendix \ref{proof_auto_thm1}.
\end{proof}

With the above corollary, the following remark is immediate.
\begin{remark}[Bi-level outliers]
The lower bound on outlier magnitudes can be relaxed to the following which only requires that {most} outlier magnitudes are lower bounded, while the others have small enough magnitudes so that their squared sum is upper bounded:
Assume that the outlier magnitudes are such that the following holds: $\xt$ can be split as $\xt = (\xt)_{small} + (\xt)_{large}$ with the two components having disjoint supports and being such that, $\|(\xt)_{small}\|  \le b_{v,t}$ and the smallest nonzero entry of $(\xt)_{large}$ is greater than $30 b_{v,t}$ with $b_{v,t}$ as defined in Corollary \ref{gen_xmin}. If the above is true, and if the vectors $(\xt)_{small}$ are zero mean, mutually independent, and independent of $\lt$'s and of the support of $(\xt)_{large}$, then all conclusions of Theorem \ref{thm1} hold except the exact support recovery conclusion (this gets replaced by exact recovery of the support of $(\xt)_{large}$).

This remark follows by replacing $\vt$ by $\vt + (\xt)_{small}$ and $\xt$ by $(\xt)_{large}$ in Corollary \ref{gen_xmin}.
\label{rem:r_change}
\end{remark}

\begin{remark}
The first condition (accurate initial estimate) can be satisfied by applying any standard RPCA solution, e.g., PCP, AltProj, or GD, on the first $t_\train = C r$ data frames. This requires assuming that $t_1 \ge C r$, and that $\Y_{[1,t_\train]}$ has outlier fractions in any row or column bounded by $c/r$.
Moreover, it is possible to significantly relax the initial estimate requirement to only requiring that $\SE(\Phat_0,\P_0) \le c/\sqrt{r}$ if we use $K$ iterations of the approach of follow-up work \cite{rrpcp_icml_trans_it} to improve the estimate of $\P_0$ until a $\zz$ accurate estimate is obtained, and then run s-ReProCS. For this to work, we will need a larger lower bound on $\xmin$ for the initial period.
\end{remark}

\begin{remark}[Connecting to the left incoherence of standard RPCA solutions]
With minor changes, our left incoherence assumption, \eqref{incoh_2}, can be replaced by something that is very close to the one used by all standard RPCA solutions. Instead of \eqref{incoh_2}, we can assume $\mu$-incoherence of $\basis([\P_0,\P_{1,\new},\P_{2,\new},\dots,\P_{J,\new}])$. This implies that\footnote{follows because the union of the spans of $\P_{j-1}$ and $\P_j$ is contained in the span of $[\P_0,\P_{1,\new},\P_{2,\new},\dots,\P_{J,\new}]$.} the RHS of \eqref{incoh_2} is bounded by $\sqrt{ \mu (r+J)/n}$.  With this, the only change to Theorem \ref{thm1} will be that we will need $\outfraccol \le 0.01/(2\mu (r+J))$.

If  $\P_{j,\new}$ is orthogonal to $\Span([\P_0, \P_{1,\new}, \dots, \P_{j-1,\new}])$ for each $j$, then the matrix $[\P_0,\P_{1,\new},\P_{2,\new},\dots,\P_{J,\new}]$ is itself a basis matrix, its span is {\em equal} to that of the left singular vectors of $\L$, and its rank $r+J = \rmat$. In this case, the above assumption and the corresponding required bound on $\outfraccol$ are exactly the same as those used by the standard RPCA solutions.
\end{remark}

\begin{algorithm}[t!]
\caption{\small{Simple-ReProCS (with $t_j$ known). We state this first for simplicity. The actual automatic version is given later in Algorithm \ref{simp_reprocs_auto}.
Let $\Lhat_{t;\alpha}:= [\lhat_{t-\alpha+1}, \lhat_{t-\alpha+2}, \dots, \lhat_t]$.
}}
\label{simp_reprocs_tj} \label{reprocs}
\begin{algorithmic}[1]
\algrenewcommand\algorithmicindent{1em}
\State \textbf{Input}:  $\Phat_0$, $\yt$,  \textbf{Output}:  $\shatt$, $\lhatt$,  $\Phat_{(t)}$,  \textbf{Parameters:} $\omega_{supp}$, $K$, $\alpha$, $\xi$, $r$, $t_j$'s  
\State $\tPhat_{(t_\train)} \leftarrow \hat{\pt}_{0}$;  $\tildej~\leftarrow~1$, $k~\leftarrow~1$
\For {$t > t_\train$}
\State ($\xhat_t$, $\That_t$) $\leftarrow$ \Call{ProjCS}{$\tPhat_{(t-1)}$, $\yt$} \Comment{Algorithm \ref{func_projCS} }
\State $\hat{\bm{\ell}}_t \leftarrow \yt - \hat{\bm{x}}_t$.
\State ($\tPhat_{(t)}$, $\Phat_j$, $j$, $k$) $\leftarrow$ \Call{SubUp}{$\Lhat_{t; \alpha}$, $\Phat_{j-1}$, $t$, $t_j$, $j$, $k$, $\tPhat_{(t-1)}$} \Comment{Algorithm \ref{func_sub_update}} 
\EndFor
\end{algorithmic}
\end{algorithm}

\begin{algorithm}[t!]
\caption{\small{Projected CS (ProjCS)}}
\label{func_projCS}
\begin{algorithmic}[0]
\algrenewcommand\algorithmicindent{1em}
\Function{ProjCS}{$\tPhat_{(t-1)}$, $\yt$}
\State $\bpsi \leftarrow \bm{I} - \tPhat_{(t-1)}\tPhat_{(t-1)}{}'$
\State $\tty_t \leftarrow \bpsi \yt$
\State $\xhat_{t,cs} \leftarrow \arg\min_{\tilde{\bm{x}}} \norm{\tilde{\bm{x}}}_1 \ \text{s.t}\ \norm{\tilde{\bm{y}}_t - \bpsi \tilde{\bm{x}}} \leq \xi$
\State $\That_t \leftarrow \{i:\ |\xhat_{t,cs}| > \omega_{supp} \}$
\State $\xhat_t \leftarrow \I_{\That_t} ( \bpsi_{\That_t}{}' \bpsi_{\That_t} )^{-1} \bpsi_{\That_t}{}'\tty_t$
\State \Return $\xhat_t$, $\That_t$
\EndFunction
\end{algorithmic}
\end{algorithm}

\begin{algorithm}[t!]
\caption{\small{Subspace Update (SubUpd).}} 
\label{func_sub_update}
\begin{algorithmic}[0]
\algrenewcommand\algorithmicindent{1em}
\Function{SubUp}{$\Lhat_{t;\alpha}$, $\Phat_{j-1}$, $t$, $t_j$, $j$, $k$, $\tPhat_{(t-1)}$} 
\If{$t = t_j + u \alpha$ for $u = 1,\ 2,\ \cdots, K + 1$}
\State $\bm{B} \leftarrow (\I - \Phat_{j-1} \Phat_{j-1} {}') \Lhat_{t; \alpha}$
\State $\Phat_{j, \add, k} \leftarrow  \SVD_1[\bm{B}]$         \Comment{subspace addition: via $K$ steps of projection-SVD} 
\State $\tPhat_{(t)} \leftarrow [\Phat_{j-1}, \Phat_{j,\add,k}]$, $k \leftarrow k + 1$.
\If{$k=K + 1 $}
\State $\Phat_j \leftarrow \SVD_r[\Lhat_{t; \alpha}]$ \Comment{subspace deletion: via subspace re-estimation using simple SVD} 
\State $\tPhat_{(t)} \leftarrow \Phat_j$, $j\leftarrow j+1$, $k\leftarrow 1$.
\EndIf
\Else
\State $\tPhat_{(t)} \leftarrow \tPhat_{(t-1)}$
\EndIf
\State \Return $\tPhat_{(t)}$, $\Phat_j$, $j$, $k$
\EndFunction
\end{algorithmic}
\end{algorithm}



\newcommand{\vv}{\bm{v}}
\subsection{Discussion}\label{discuss}
In this section, we discuss the various implications of our result, the speed and memory guarantees, explain how to set algorithm parameters, and finally discuss its limitations.%

\subsubsection{Subspace change detection and tracking with short delay}
Theorem \ref{thm1} shows that, whp, the subspace change gets detected within a delay of at most $2\alpha = C f^2 (r \log n)$ frames, and the subspace gets estimated accurately within at most $(K+3) \alpha = C f^2 (r \log n)\log(1/\zz)$ frames. Each column of the low rank matrix is recovered with a small time-invariant bound without any delay. If offline processing is allowed, with a delay of at most $(K+3) \alpha$, we can guarantee all recoveries within normalized error $\zz$, or, in fact, with minor modifications, within any $\epsilon = c \zz$ for $c<1$ (also see the limitations' discussion).
Notice that the required delay between subspace change times is more than $r$ by only logarithmic factors (assuming $f$ does not grow with $n$ or $r$).
Since the previous subspace is not exactly known (is known within error at most $\zz$), at each update step, we {\em do} need to estimate an $r$-dimensional subspace, and not a one-dimensional one. Hence it is not clear if the required delay can be reduced any further. Moreover, the delay required for the deletion step cannot be less than $r$ even in the ideal case when $\lt$ is directly observed.

\subsubsection{Bi-level outliers}
Consider the upper bound on $\Delta$ (amount of subspace change). Observe that the upper bound essentially depends on the ratio between $\xmin$ (minimum outlier magnitude) and $\sqrt{\lambda_\ch}$. Read another way, this means that $\xmin$ needs to be lower bounded. On first glance, this may seem counter-intuitive since sufficiently small magnitude corruptions should not be problematic. This is actually true. Sufficiently small magnitude corruptions get classified as the small noise $\vt$.
Moreover, as noted in Corollary \ref{gen_xmin} and Remark \ref{rem:r_change}, our result actually allows ``bi-level'' corruptions/outliers that need to satisfy a much weaker requirement than this: the large-outliers have magnitude that is ``large enough'', while the rest are such that the squared sum of their magnitudes is ``small enough''.
The threshold for both ``large enough'' and ``small enough'' decreases with each subspace update step.


\subsubsection{Order-wise improvement in allowed upper bound on maximum number of outliers per row}
As pointed out in \cite{robpca_nonconvex}, solutions for standard RPCA (that only assume incoherence of left and right singular vectors of $\L$ and nothing else, i.e., no outlier support model) cannot tolerate\footnote{The reason is this: let $\rrow = \outfracrow$, one can construct a matrix $\X$ with $\rrow$ outliers in some rows that has rank equal to $1/\rrow$. A simple way to do this would be to let the support and nonzero entries of $\X$ be constant for $\rrow \tmax$ columns before letting either of them change. Then the rank of $\X$ will be $\tmax/(\rrow \tmax)=1/\rrow$. If $1/\rrow < \rmat=\rank(\L)$, $\X$ will wrongly get classified as the low-rank component. This is why we need $\rrow=\outfracrow < 1/\rmat$.
A similar argument can be used for $\outfraccol$.} a bound on maximum outlier fractions in any row or any column that is larger than $1/\rmat$.
However observe that simple-ReProCS can tolerate $\outfracrow^\alpha \in O(1)$ (this assumes $f$ is a constant). This is a significant improvement over all existing RPCA results with important practical implications for video analytics.
This is possible is because s-ReProCS uses extra assumptions, we explain their next. 

\subsubsection{The need for extra assumptions}
s-ReProCS recovers the sparse outliers $\xt$ first and then the true data $\lt$, and does this at each time $t$.
Let $\bpsi :=\I - \tPhat_{(t-1)} \tPhat_{(t-1)}{}'$.
When recovering $\xt$, it exploits two facts: (a) the subspace of $\lt$, $\tP_{(t)}$, satisfies the denseness/incoherence property, and (b) ``good'' knowledge of the subspace of $\lt$ (either from initialization or from the previous subspace's estimate and slow subspace change) is available. Using these two facts one can show that $\bpsi$ satisfies the RIP property, and that the ``noise'' seen by the compressive sensing step, $\b_t:=\bpsi (\lt + \vt)$, is small.
This, along with a guarantee for CS, helps ensure that the error in recovering $\xt$ is upper bounded\footnote{Since the individual vector $\b_t$ does not have any structure that can be exploited, the error in recovering $\xt$ cannot be made lower than this. However the $\b_t$'s arranged into a matrix do form a low-rank matrix whose approximate rank can be shown to be one (under our current subspace change model). If we try to exploit this structure we end up with the modified-PCP approach studied in earlier work \cite{zhan_pcp_jp}. This needs the uniform random support assumption \cite{zhan_pcp_jp}.} by $C \|\b_t\|$.
This, in turn, means that, to correctly recover the support of $\xt$, the minimum large-outlier magnitude needs to be larger than $C\|\b_t\|$. This is where the $\xmin$ or $\xmint$ lower bound comes from\footnote{If there were a way to bound the element-wise error of the CS step (instead of the $l_2$ norm of the error), we could relax the $\xmin$ lower bound significantly. It is not clear if this is possible though.}.

Correct outlier support recovery is needed to ensure that the subspace estimate can be improved with each subspace update step. In particular, it helps ensure that the error vectors $\et:=\xt - \xhat_t$ in a given subspace update interval are mutually independent when conditioned on the $\yt$'s from all past intervals. This fact also relies on the mutual independence assumption on the $\at$'s. Moreover, mutual independence, along with the element-wise boundedness and identical covariances assumption, on the $\at$'s helps ensure that we can use matrix Bernstein \cite{tail_bound} and Vershynin's sub-Gaussian result (bounds singular values of matrices with independent sub-Gaussian rows) \cite{vershynin} for obtaining the desired concentration bounds on the subspace recovery error in each step. As explained below, the above assumptions on $\at$ replace the right incoherence assumption. Finally, because ReProCS is updating the subspace using just the past $\alpha$ estimates of $\lhat_t$'s, in order to show that each such step improves the estimate we need to bound $\outfracrow^\alpha$ (instead of just $\outfracrow$).

\subsubsection{Time and Memory Complexity}
Observe that the s-ReProCS algorithm needs memory of order $n \alpha$ in online mode and $K n \alpha$ in offline mode. Assuming $\alpha=\alpha_*$, even in offline mode, its memory complexity is near-optimal and equal to $O(nr \log n \log (1/\zz))$.
Also, observe that the time complexity of s-ReProCS is $O(n \tmax r \log(1/\zz))$. We explain this in Appendix \ref{time_comp_reprocs}.
These claims assume that $f$ is constant with $n,r$. If the dependence on $f$ is included both will be multiplied by $f^2$.

\subsubsection{Subspace and outlier assumptions' tradeoff} 
When there are fewer outliers in the data or when outliers are easy to detect, one would expect to need weaker assumptions on the true data subspace or its rate of change. This is indeed true.
For the original RPCA results, this is encoded in the condition $\max(\outfracrow,\outfraccol) \le c/ (\mu \rmat)$ where $\mu$ quantifies not-denseness of both left and right singular vectors. 
From Theorem \ref{thm1}, this is also how $\outfraccol$, $\mu$ (not-denseness of only left singular vectors) and $\rmat$ are related for dynamic RPCA.
%
On the other hand, for our result, $\outfracrow^\alpha$ and the lower bound on $\xmin$ govern the allowed rate of subspace change. The latter relation is easily evident from the bound on $\Delta$. If $\xmin$ is larger (outliers are large magnitude and hence easy to detect), a larger $\Delta$ can be tolerated.
The relation of $\outfracrow$ to rate of change is not evident from the way the guarantee is stated in Theorem \ref{thm1}. The reason is we have assumed $\outfracrow^\alpha \le \rrow = 0.01 /f^2$ and used that to get a simple expression for $K$. If we did not do this, we would need $K$ to satisfy
\[
c_1 \Delta (c_2 f \sqrt{\rrow} )^K + 0.2 \zz \le \zz.
\]
With this, $K$ needs to be $K = \lceil \frac{1}{-\log (c_2 f \sqrt{\rrow} )} \log (\frac{c_1 \Delta}{0.8\zz}) \rceil$.
Recall that we need $t_{j+1}-t_j \ge (K+3)\alpha$. Thus, a smaller $\rrow$ (smaller $\outfracrow^\alpha$) means one of two things: either a larger $\Delta$ (more change at each subspace change time) can be tolerated while keeping $K$, and hence the lower bound on the delay between change times, the same; or, for $\Delta$ fixed, a smaller lower bound is needed on the delay between change time.
The above can be understood by carefully checking the proof\footnote{The multiplier 0.4 of $q_\rot$ in its first claim is obtained by setting $\rrow=0.01/f^2$. If we do not do this, 0.4 will get replaced by $c_2 f \sqrt{\rrow}$.} of Theorem \ref{thm_corpca}.
%

\subsubsection{The need for detecting subspace change}
As pointed out by an anonymous reviewer, it may not be clear to a reader why we need to explicitly detect subspace change (instead of just always doing subspace update at regularly spaced intervals). The change detection is needed for two key reasons. First, in the projection-SVD step for subspace update, we use $\Phat_{j-1}$ as the best estimate of the previous subspace. We let $\Phat_{j-1}$ be the final $\zz$-accurate subspace estimate obtained after $K$ projection-SVD steps and then one subspace re-estimation step. To know when the $K$ updates are over, we need to know when the first update of the new subspace occurred, or in other words, we need an estimate of when the subspace change occurred.
Second, in the current algorithm, because we detect change, we can choose to use the $\lhat_t$'s from the next $\alpha$-frame interval, i.e. $[\lhat_{\that_j+1},\lhat_{\that_j+2}, \dots, \lhat_{\that_j+\alpha}]$, for the first subspace update. This ensures that the $\lhat_t$'s from the interval that contains $t_j$ (some of the $\lt$'s in this interval come from $\P_{j-1}$ while others come from $\P_j$) is never used further in any subspace update. The is essential because,  if these are used, one will get an incorrect subspace estimate (something in between $\P_{j-1}$ and $\P_j$) and one whose error cannot easily be bounded.
If subspace change is never detected, this cannot be ensured.


\subsubsection{Algorithm parameters}
%
Observe from Theorem \ref{thm1} that we need knowledge of only 4 model parameters - $r$, $\lambda^+$, $\lambda^-$ and $\xmin$ - to set our algorithm parameters.
The initial dataset used for estimating $\Phat_0$ (using PCP/AltProj) can also be used to get an accurate estimate of $r$, $\lambda^-$ and $\lambda^+$ using standard techniques (maximum likelihood applied to the AltProj estimate of $[\l_1,\l_2,\dots,\l_{t_\train}]$).
Thus one really only needs to set $\xmin$. If continuity over time is assumed, a simple heuristic is to let it be time-varying and use $\min_{i \in \That_{t-1}}|(\xhat_{t-1})_i|$  as its estimate at time $t$. This approach in fact allows us to estimate $\xmint$ and thus allows for larger unstructured noise, $\vt$, levels as allowed by Corollary \ref{gen_xmin}.

The most interesting point for practice though is that of Remark \ref{rem:r_change}. It indicates that when a subspace change is detected but not estimated, starting at the previous $2\alpha$ frames, one should use a larger value of the support estimation threshold $\omega_{supp}$. After each subspace update step, $\omega_{supp}$ should be decreased roughly exponentially.



\subsubsection{Dependence on $f$}
Observe that $f$ appears in our guarantee in the bound on $\outfracrow^\alpha$ and in the expression for $\alpha$.
The $\outfracrow^\alpha$ bound is stated that way only for simplicity. Actually, for all time instants except the $\alpha$-length period when the subspace re-estimation (for deletion) step is run, we only need $\outfracrow^\alpha \le 0.01$.
We need the tighter bound $\outfracrow^\alpha \le \rrow=0.01/f^2$ only for the simple SVD based subspace re-estimation (deletion) step to work (i.e., only for $t \in [\that_j+K\alpha,\that_j+K\alpha+\alpha)$). Thus, if offline ReProCS were being used to solve the standard RPCA type problem (where $\rmat$ is nicely bounded), one could choose to never run the subspace deletion step. This will mean that the resulting algorithm (s-ReProCS-no-delete) will need  $\outfraccol < c/\mu \rmat$, but then $\outfracrow < 0.01$ will suffice (the bound would not depend on $1/f^2$).
The $\alpha$ expression governs required delay between subspace change times, tracking delay, and time and memory complexity. If the deletion step is removed, the dependence of $\alpha$ on $f$ will not disappear, but will weaken (it will linearly depend on $f$ not on $f^2$).

We now try to relate $f$ to the condition number of $\L$. Observe that $f$ is the condition number of $\E[\L_j \L_j{}']$ for any $j$. The condition number of the entire matrix $\L$ can be much larger when slow subspace change holds ($\Delta$ is small). To see this, 
let $\kappa^2$ denote the condition number of $\E[\L \L']$, so that, whp, $\kappa$ is approximately the condition number of $\L$. It is not hard to to see that\footnote{To understand this, suppose that there is only one subspace change and suppose that the intervals are equal, i.e., $\tmax-t_1=t_1-t_0$. Then, $\E[\L \L']/\tmax =  \P_{0,\fx} \bm\Lambda_\fx \P_{0,\fx}{}' +  [\P_{0,\ch} \ \P_{1,\new}] \B [\P_{0,\ch}{}' \ \P_{1,\new}{}']'$ where $\B = \lambda_\ch \matr{(0.5+0.5\cos^2\theta_1) \ & \ -0.5\sin\theta_1\cos \theta_1}{-0.5\sin\theta_1\cos \theta_1 \  & \ 0.5 \sin^2\theta_1}$. The maximum eigenvalue of $\E[\L \L']/\tmax$ is $\lambda^+$. Its minimum eigenvalue is the minimum eigenvalue of $\B$ which can be computed as $(1-\cos\theta_1)\lambda_\ch$. In the worst case $ \lambda_\ch=\lambda^-$. 
When the intervals are not equal, this gets replaced by $(1 - \sqrt{1 - 2c \sin^2\theta_1})\lambda^-$ for a $c<1$. This is at most ${ (1 - \sqrt{1 - 2c \Delta^2)}}$.
}, in the worst case (if $\lambda^- = \lambda_\ch$), $\kappa^2 = \frac{f}{1 - \sqrt{1 - 2c \Delta^2}} \approx C \frac{f}{\Delta^2}$ when $\Delta$ is small. Thus, if $\Delta$ is small, $\kappa \approx C \sqrt{f}/\Delta$ can be much larger than $f$. The guarantees of many of the RPCA solutions such as RPCA-GD \cite{rpca_gd} depend on $\kappa$.


\subsubsection{Relating our assumptions to right incoherence of $\L_j := \L_{[t_j, t_{j+1})}$ \cite{rpca_zhang}}\label{rem:bounded}
We repeat this discussion from \cite{rrpcp_icml_trans_it}.
From our assumptions, $\L_j = \P_j \A_j$ with $\A_j:= [\a_{t_j},\a_{t_j+1},\dots \a_{t_{j+1}-1}]$, the columns of $\A_j$ are zero mean, mutually independent, have identical covariance $\Lam$, $\Lam$ is diagonal, and are element-wise bounded as specified by Theorem \ref{thm1}. Let $d_j := t_{j+1}-t_j$.
Define a diagonal matrix $\Sigma$ with $(i,i)$-th entry $\sigma_i$ and with $\sigma_i^2 := \sum_t (a_t)_i^2 / d_j$. Define a $d_j \times r$ matrix $\tilde\V$ with the $t$-th entry of the $i$-th column being $(\tilde\vv_i)_t: = (\a_t)_i/ (\sigma_i \sqrt{d_j})$. Then, $\L_j = \P_j \Sigma \tilde{\V}'$ and each column of $\tilde\V$ is unit 2-norm. Also, from the bounded-ness assumption, $(\tilde\vv_i)_t^2 \leq \eta \frac{\lambda_i}{\sigma_i^2} \cdot \frac{1}{d_j}$ where $\eta$ is a numerical constant.
Observe that $\P_j \Sigma \tilde\V'$ is not exactly the SVD of $\L_j$ since the columns of $\tilde\V$ are not necessarily exactly mutually orthogonal. However, if $d_j$ is large enough, using the assumptions on $\at$, one can argue using any law of large numbers' result (e.g., Hoeffding inequality), that (i) the columns of $\tilde\V$ are approximately mutually orthogonal whp, and (ii) $\sigma_i^2 \ge 0.99 \lambda_i$ whp. 
Thus, our assumptions imply that, whp, $\tilde\V$ is a basis matrix and $(\tilde\vv_i)_t^2 \leq C/{d_j}$.

With the above, one can interpret $\tilde\V$ as an ``approximation'' to the right singular vectors of $\L_j$ and then the above bound on $(\tilde\vv_i)_t^2$ is the same as  the right incoherence condition assumed by \cite{rpca_zhang}. It is slightly stronger than what is assumed by \cite{rpca,robpca_nonconvex} and others (these do not require a bound on each entry but on each row, they require that the squared norm of each row of the matrix of right singular vectors be bounded by $C r / d_j$). 

Ideally we would like to work with the exact SVD of $\L_j$, however this is much harder to analyze using our statistical assumptions on the $\at$'s. To see this,
%
suppose $\A_j \svdeq \bm{U} \Sigma \V'$, then $\L_j \svdeq (\P_j \bm{U})  \Sigma \V'$ is the exact SVD of $\L_j$. Here $\bm{U}$ is an $r \times r$ orthonormal matrix. Now it is not clear how to relate the element-wise bounded-ness assumption on $\at$'s to an assumption on entries of $\V$, since now there is no easy expression for each entry of $\V$ or of the entries of $\Sigma$ in terms of $\at$ ($\bm{U}$ is an unknown matrix that can have all nonzero entries in general).

\subsubsection{Limitations of our guarantees}
s-ReProCS needs a few extra assumptions beyond slow subspace change and what static RPCA solutions need:
(i) instead of a bound on outlier fractions per row of the entire data matrix (which is what standard RPCA methods assume), it needs such a bound for every sub-matrix of $\alpha$ consecutive columns; (ii) it makes statistical assumptions on the principal subspace coefficients $\at$ (with mutual independence being the strongest requirement);
(iii) it needs to lower bound $\xmin$; 
and (iv) it uses $\zz$ to denote both the initial subspace error as well as the final recovery error achieved after a subspace update is complete.
Here (i) is needed because ReProCS is an online algorithm that uses $\alpha$ frames at a time to update the subspace, and one needs to show that each update step provides an improved estimate compared to the previous one. However, since $\alpha$ is large enough, requiring a bound  $\outfracrow^\alpha$ is not too much stronger than requiring the same bound on the outlier fractions per row of the entire $n\times \tmax$ matrix $\Y$.
In fact, if we compare the various RPCA solutions with storage complexity fixed at $O(n\alpha) = O(nr \log n)$, i.e., if we implement the various static RPCA solutions for every new batch of $\alpha$ frames of data, then, the static RPCA solutions will also need to bound $\outfracrow^\alpha$ defined in \eqref{def_outfracrow}. As discussed earlier, these will require a much tighter bound of $c/r$ though.
%
(iv) is assumed for simplicity.  What we can actually prove is something slightly stronger: if the initial error is $\zz$, and if $\epsilon = c \zz$ for a constant $c$ which may be less than one, then, without any changes, we can guarantee the final subspace error to be below such an $\epsilon$. More generally, as long as the initial error $\zz \le \Delta$, it is possible to achieve final error $\epsilon$ for {\em any} $\epsilon>0$ if we assume that $t_1-t_\train > K \alpha$, assume a slightly larger lower bound on $\xmin$, and if we modify our initialization procedure to use the approach of follow-up work \cite{rrpcp_icml_trans_it}.

Limitations (ii) and (iii) are artifacts of our proof techniques. The mutual independence can be replaced by an autoregressive model on the $\at$'s by borrowing similar ideas from \cite{rrpcp_aistats}.
The mutual independence and zero mean assumption on the $\at$'s is valid for the video analytics' application if we let $\lt$ be the mean-subtracted background image at time $t$. Then, $\lt$ models independent zero-mean background image variations about a fixed mean image, e.g., variations due to lighting variations or due to moving curtains; see Fig. \ref{fig:video_res}. This type of mean subtraction (with an estimate of the mean background image computed from training data) is commonly done in practice in many practical applications where PCA is used; it is also done in our video experiments shown later.
%
(iii) is  needed because our proof first tries to show exact outlier support recovery by solving a CS problem to recover the outliers from the projected measurements, followed by thresholding. It should be possible to relax this  by relaxing the exact support recovery requirement which, in turn, will require other significant changes. For example, it may be possible to do this if one is able to do a deterministic analysis. 
It may be possible to also completely eliminate it if we replace the CS step by thresholding with carefully decreasing thresholds in each iteration (borrow the idea of AltProj); however, we may then require the same tight bound on $\outfracrow$ that AltProj needs. By borrowing the stagewise idea of AltProj, it may also be possible to remove all dependence on $f$.

\section{Discussion of Related Work} \label{rel_work}

\subsubsection{Limitations of earlier ReProCS-based guarantees \cite{rrpcp_perf,rrpcp_isit15,rrpcp_aistats}}
In \cite{rrpcp_perf}, we introduced the ReProCS idea and proved a {\em partial} guarantee for it. We call it a partial guarantee because it needed to assume something about the intermediate  subspace estimates returned by the algorithm. However, this work is important because it developed a nice framework for proving guarantees for dynamic RPCA solutions. Both our later complete guarantees \cite{rrpcp_isit15,rrpcp_aistats} as well as the current result build on this framework.

The current work is a significant improvement over the complete guarantees obtained in \cite{rrpcp_isit15,rrpcp_aistats} for two other ReProCS-based algorithms for three reasons.
(i) The earlier works needed very specific assumptions on how the outlier support could change (needed an outlier support model inspired by video moving objects). Our result removes such a requirement and instead only needs a bound on the fraction of outliers per column of the data matrix and on the fraction per row of an $\alpha$-consecutive-column sub-matrix of the data matrix (for $\alpha$ large enough).
(ii) The subspace change model assumed in these earlier papers can be interpreted as the current model (given in Sec. \ref{main_res}) with $\theta_j=90^\circ$ or equivalently with $\Delta=1$. This is an unrealistic model for slow subspace change, e.g., in 3D, it implies that the subspace changes from the x-y plane to the y-z plane. Instead, our current model  allows changes from x-y plane to a slightly tilted x-y plane as shown in Fig. \ref{subs_ch_fig}. This modification is more realistic {\em and} it allows us to replace the upper bound on $\lambda_\ch$ required by the earlier results by a similar bound on $\lambda_\ch \Delta^2$ (see assumption 1b of Theorem \ref{thm1}). Since $\Delta$ quantifies rate of subspace change, this new requirement is much weaker. It can be satisfied by assuming that $\Delta$ is small, without making any assumption on $\lambda_\ch$.
(iii) The required minimum delay between subspace change times in the earlier results depended on $1/\epsilon^2$ where $\epsilon$ is the desired final subspace error after a subspace update is complete. This is a strong requirement. Our current result removes this unnecessarily strong dependence. The delay now only depends on $(-\log \epsilon)$ which makes it much smaller. It also implies that the memory complexity of simple-ReProCS is near-optimal.
(iv) Unlike \cite{rrpcp_isit15,rrpcp_aistats}, we analyze a simple ReProCS-based algorithm that ensures that the estimated subspace dimension is bounded by $(r+1)$, without needing the complicated cluster-SVD algorithm. This is why our guarantee allows outlier fractions per column to be below $c/r$. The work of \cite{rrpcp_isit15} needed this to be below $c/\rmat$ while \cite{rrpcp_aistats} needed an extra assumption (clustered eigenvalues). For long data sequences, $c/r$ can be much larger than $c/\rmat$. We provide a detailed comparison of these assumptions in Table \ref{compare_assu_reprocs}.

\subsubsection{Complete guarantees for other dynamic RPCA or RPCA solutions}
Another approach that solves dynamic RPCA, but in a piecewise batch fashion, is modified-PCP (mod-PCP) \cite{zhan_pcp_jp}.  The guarantee for mod-PCP was proved using ideas borrowed from \cite{rpca} for PCP. Thus, like \cite{rpca}, it also needs uniformly randomly generated support sets which is an unrealistic requirement.  For the video application, this requires that foreground objects are single pixel wide and move around the entire image completely randomly over time. This is highly impractical.
In Table \ref{compare_assu}, we provide a comparison of our current guarantees for simple-ReProCS and its offline version with those for original-ReProCS \cite{rrpcp_perf,rrpcp_isit15,rrpcp_aistats}, modified-PCP \cite{zhan_pcp_jp}, as well as with those for solutions for standard RPCA -
 \cite{rpca} (referred to as PCP(C)), \cite{rpca_zhang} (PCP(H), this strictly improves upon \cite{rpca2}), AltProj \cite{robpca_nonconvex}, RPCA via gradient descent (GD) \cite{rpca_gd} and nearly-optimal robust matrix completion (NO-RMC) \cite{rmc_gd}. The table also contains a speed and memory complexity comparison.
 Offline s-ReProCS  can be interpreted as a solution for standard RPCA.
%
%
%
From Table \ref{compare_assu}, it is clear that for data that satisfies slow subspace change and the assumption that outlier magnitudes are either large or very small, and that is such that its first $t_\train$ frames, $\Y_{[1,t_\train]}$, satisfy AltProj (or PCP) assumptions, s-ReProCS and offline s-ReProCS have the following advantages over other methods.
\ben
\item  For the data matrix after $t_\train$, i.e., for $\Y_{[t_\train+1,\tmax]}$, ReProCS needs the weakest bound on $\outfracrow^\alpha$ without requiring  uniformly randomly generated outlier support sets. This is comparable to the bound needed by PCP(C) or mod-PCP but both assume uniform random outlier supports which is a very strong requirement. 


\item The memory complexity of s-ReProCS is significantly better than that of all other published methods for RPCA that provably work, and is nearly optimal.

\item Both in terms of time complexity order (Table \ref{compare_assu}) and experimentally (see Sec. \ref{sims_detail}), s-ReProCS and its offline counterpart are among the fastest, while having the best, or nearly the best, performance experimentally as well.
    Order-wise, only NO-RMC \cite{rmc_gd} is faster than s-ReProCS. However, NO-RMC needs the data matrix to be nearly square, i.e., it needs $c_1 n \ge \tmax \ge c_2 n$. This is a very strong requirement that often does not hold: for the video application it requires that the number of video frames $\tmax$ be roughly as large as $n$ (number of pixels in one image frame). The reason is that NO-RMC deliberately under-samples the data matrix $\Y$ by randomly throwing away some of its entries and using only the rest even when all are available. In other words, it always solves the robust matrix completion (RPCA with missing entries) problem and this is what results in a significant speed-up, but this is also why it needs $\tmax \approx n$.

\item s-ReProCS can automatically detect subspace change and then also track it with a short delay, while the other approaches (except original-ReProCS) cannot.
    Notice that s-ReProCS also needs a weaker upper bound of $c/r$ on $\outfraccol$ while the batch techniques (PCP, AltProj, GD, NO-RMC) applied to the entire matrix $\L$ need this to be below $c/\rmat$. Of course, if the batch techniques are applied on pieces of data $\Y_{j}:=\Y_{[t_{j},t_{j+1})}$, they also need the same looser bound of $c/r$ on $\outfraccol$ and their memory complexity improves too. However, the batch methods do not have a way to estimate the change times $t_j$, while s-ReProCS does.
    Moreover, since the other methods (except modified-PCP) do not have a way to use the previous subspace information, if the pieces chosen are are too small, e.g., if the methods are applied on $\alpha$-frames at a time, their performance is much worse then when the entire dataset is used jointly. 

\een


\subsubsection{Comparison with follow-up work on ReProCS-NORST \cite{rrpcp_icml_trans_it}}
As compared to ReProCS-NORST, which is the algorithm studied in our follow-up work \cite{rrpcp_icml_trans_it} (which allows all $r$ directions of the subspace to change at each $t_j$),  simple-ReProCS has three advantages: (i) it is faster, (ii) it needs a weaker lower bound on $\xmin$ (its required lower bound essentially does not depend on $r$ if $\zz$ is very small), and (iii) if it is used to solve the standard RPCA problem (estimate span of columns of the entire matrix $\L$), we can eliminate the $r$-SVD based subspace re-estimation (deletion) step. With this change, (a) the required upper bound on $\outfracrow^\alpha$ for s-ReProCS does not depend on the condition number $f$ (just  $\outfracrow^\alpha \le 0.01$ suffices), and (b) its time complexity improves by a factor of $r$ (if the initialization step is ignored) compared to ReProCS-NORST. Of course it will mean s-ReProCS-no-delete will need $\outfraccol < c/(r+J)$ which is slightly stronger.

Since ReProCS-NORST allows all $r$ directions of the subspace to change, it also has many advantages over s-ReProCS: its subspace tracking delay is near-optimal, and it allows for a weaker initialization assumption.


\begin{algorithm}[t!]
\caption{\small{The actual simple-ReProCS algorithm, this is the one that is studied in our guarantees and also implemented in our experiments. Let $\Lhat_{t;\alpha}:= [\lhat_{t-\alpha+1}, \lhat_{t-\alpha+2}, \dots, \lhat_t]$.
}}
\label{simp_reprocs_auto} \label{auto_reprocs}
\begin{algorithmic}[1]
\algrenewcommand\algorithmicindent{0.5em}
\State \textbf{Input}:  $\Phat_0$, $\yt$,  \textbf{Output}:  $\shatt$, $\lhatt$,  $\Phat_{(t)}$
\State \textbf{Parameters:} $\omega_{supp}$, $K$, $\alpha$, $\xi$, $r$, $\lthres$
\State Let $\Lhat_{t;\alpha}:= [\lhat_{t-\alpha+1}, \lhat_{t-\alpha+2}, \dots, \lhat_t]$.
\State $\tPhat_{(t_\train)} \leftarrow \hat{\pt}_{0}$;  $\tildej~\leftarrow~1$, $k~\leftarrow~1$
\For {$t > t_\train$}
\State ($\xhat_t$, $\That_t$) $\leftarrow$ \Call{ProjCS}{$\tPhat_{(t-1)}$,$\yt$}  
\Comment{Algorithm \ref{func_projCS}}
\State $\hat{\bm{\ell}}_t \leftarrow \yt - \hat{\bm{x}}_t$.
\State ($\tPhat_{(t)}$, $\Phat_j$, $\hat{t}_j$, $k$, $j$, $\text{phase}$)  $\leftarrow$ \Call{AutoSubUpd}{$\Lhat_{t; \alpha}$, $\Phat_{j-1}$, $t$, $\hat{t}_{j-1}$, $j$, $k$, $\text{phase}$, $\tPhat_{(t-1)}$}  \Comment{Algorithm \ref{aut_sub_up}}
\EndFor
\State {\bf Offline ReProCS: }
At $t = \that_j + K \alpha$, for all $t \in [\that_{j-1}+ K \alpha,  \that_j + K \alpha-1]$,
\State $\tPhat_{(t)} ^{\mathrm{offline}} \leftarrow [\Phat_{j-1}, \Phat_{j,\rot,K}]$;
\State $\xhatt^{\mathrm{offline}} \leftarrow \I_{\That_t} (\bm\Psi_{\That_t}{}'\bm\Psi_{\That_t})^{-1} \bm{\Psi}_{\That_t}{}' \yt$ where $\bm\Psi:= \I - \Phat_{j-1}\Phat_{j-1}{}' - \Phat_{j,\add,K}\Phat_{j,\add,K}{}'$;
\State  $\lhatt^{\mathrm{offline}} \leftarrow \yt - \xhatt^{\mathrm{offline}}$.
\end{algorithmic}
\end{algorithm}

\begin{algorithm}[t!]
\caption{\small{Automatic Subspace Update}}
\label{aut_sub_up}
\begin{algorithmic}[0]
\Function{AutoSubUpd}{$\Lhat_{t;\alpha}$, $\Phat_{j-1}$, $t$, $\hat{t}_{j-1}$,  $j$, $k$, $\text{phase}$, $\tPhat_{(t-1)}$}
\State $\hat{t}_{j-1, fin} \leftarrow \that_{j-1} + K \alpha + \alphadel-1$
\If{$\text{phase} = \text{detect}$ and $t = \hat{t}_{j-1, fin} + u\alpha$}
\State $\bm{B} \leftarrow (\I - \Phat_{j-1}\Phat_{j-1}{}')\Lhat_{t, \alpha}$
\If {$\sigma_{\max}(\bm{B}) \geq \sqrt{\alpha \lthres}$}
\State $\text{phase} \leftarrow \text{update}$, $\hat{t}_j \leftarrow t$,
\EndIf
\State $\tPhat_{(t)} \leftarrow \tPhat_{(t-1)}$
\EndIf
\If{$\text{phase} = \text{update}$}
\State ($\tPhat_{(t)}$, $\Phat_j$, $k$) $\leftarrow$ \Call{SubUp}{$\Lhat_{t;\alpha}$, $\tPhat_{j-1}$, $t$, $\hat{t}_{j-1}$, $j$, $k$, $\text{phase}$, $\tPhat_{(t-1)}$}  \Comment{Algorithm \ref{func_sub_update}} 
\EndIf
\If{$k= K+1$}
\State $\text{phase} \leftarrow \text{detect}$
\EndIf
\State \Return $\tPhat_{(t)}$, $\Phat_j$, $\hat{t}_j$, $j$, $k$, $\text{phase}$
\EndFunction
\end{algorithmic}
\end{algorithm}

\section{Why s-ReProCS works: main ideas of our proof} \label{proof_idea}
In this section we explain the main ideas of our proof, first for the $t_j$ known case, and then explain why the subspace change detection step works.

\subsection{Why s-ReProCS with $t_j$ known works}
To understand things simply, first assume that $\that_j = t_j$, i.e., the subspace change times are known.
Consider Algorithm \ref{simp_reprocs_tj}. At each time $t$ this consists of three steps - projected Compressive Sensing (CS) to estimate $\xt$, estimating $\lt$ by subtraction, and subspace update. Consider projected CS. This is analyzed in Lemma \ref{CSlem}.
At time $t$, suppose that we have access to $\tPhat_{(t-1)}$ which is a good estimate of the previous subspace, $\Span(\tP_{(t-1)})$. Because of slow subspace change, this is also a good estimate of $\Span(\tP_{(t)})$.
Its first step projects $\yt$ orthogonal to $\tPhat_{(t-1)}$ to get $\tty_t$. Recall that $\tty_t=\bpsi \x_t + \b_t$ where $\b_t:=\bpsi (\lt + \vt)$ is small.
Using the incoherence (denseness) assumption and $\Span(\tP_{(t-1)})$ being a good estimate of $\Span(\tP_{(t)})$, it can be argued that the restricted isometry constant (RIC) \cite{candes_rip} of $\bm\Psi:= \I - \tPhat_{(t-1)}\tPhat_{(t-1)}{}'$ will be small. Using \cite[Theorem 1.2]{candes_rip}, this, along with $\|\b_t\|$ being small, ensures that $l_1$ minimization will produce an accurate estimate, $\xhat_{t,cs}$, of $\xt$. The support estimation step with a carefully chosen threshold, $\omega_{supp} = \xmin/2$, and a lower bound on $\xmin$ then ensures exact support recovery, i.e., $\That_t = \T_t$. With this, the LS step output, $\xhat_t$, satisfies $\xhat_t = \xt + \et$ with
\bea
\et &:= \I_{\T_t} (\bm\Psi_{\T_t}{}'\bm\Psi_{\T_t})^{-1} \bm\Psi_{\T_t}{}' (\lt + \vt) \nn \\
&= \I_{\T_t} (\bm\Psi_{\T_t}{}'\bm\Psi_{\T_t})^{-1} \I_{\Tt}{}' \bm\Psi (\lt + \vt)
\label{etdef0}
\eea
and with $\|\et\|$ being small. 
Computing $\lhat_t := \yt - \xhat_t$, then also gives a good estimate of $\lt$ that satisfies $\lhat_t = \lt  + \vt - \et$ with $\et$ as above.

The subspace update step uses $\lhat_t$'s to update the subspace. Since $\et$ satisfies \eqref{etdef0}, $\et$ depends on $\lt$; thus the error/noise, $\vt-\et$, in the ``observed data'' $\lhat_t$ used for the subspace update step depends on the true data $\lt$. Because of this, the subspace update does not involve a PCA or an incremental PCA problem in the traditionally studied setting (data and corrupting noise/error being independent or uncorrelated). It is, in fact, an instance of PCA when the noise/error, $\vt-\et$, in the observed data $\lhat_t$ depends on the true data $\lt$. This problem was studied in  \cite{corpca_nips,pca_dd} where it was referred to as ``correlated-PCA'' or ``PCA in data-dependent noise''.
Using this terminology, our subspace update problem (estimating $\P_j$ using $\Phat_{j-1}$) is a problem of PCA in data-dependent noise with partial subspace knowledge. To simplify our analysis, we first study this more general problem and obtain a guarantee for it in Theorem \ref{thm_corpca}  in Sec. \ref{corpca_sec}. This theorem along with Lemma \ref{CSlem} (that analyzes the projected-CS step discussed above) help obtain a guarantee for the $k$-th projection-SVD step in Lemma \ref{p_evd_lem}. The $k=1$ and $k>1$ cases are handled separately. The main assumption required for applying Theorem \ref{thm_corpca} holds because $\et$ is sparse with support $\T_t$ that changes enough ($\outfracrow^\alpha$ bound of Theorem \ref{thm1} holds).
The subspace deletion via simple SVD step of subspace update is studied in Lemma \ref{del_evd}. This step solves  a problem of PCA in data-dependent noise and so it directly uses the results from \cite{pca_dd}.

To understand the flow of the proof, consider the interval $[t_j, t_{j+1})$. Assume that, before $t_j$, the previous subspace has been estimated with error $\zz$, i.e., we have $\Phat_{j-1}$ with $\SE(\Phat_{j-1}, \P_{j-1}) \le \zz$. We explain below that this implies that, under the theorem's assumptions, we get $\SE(\Phat_{j}, \P_{j}) \le \zz$ before $t_{j+1}$.
We remove the subscripts $j$ in some of this discussion. Define the interval $\J_k:= [t_j + (k-1)\alpha, t_j +k \alpha)$. Suppose also that $\vt=0$.
\ben
\item Before the first projection-SVD step (which is done at $t=t_j+\alpha$), i.e., for $t \in \J_1$, we have no estimate of $\P_{\new}$, and hence only a crude estimate of $\P_{\rot}$. In particular, we can only get the bound $\SE(\tPhat_{(t)}, \P_\rot)= \SE(\Phat_{j-1},\P_\rot) \le \zz + |\sin \theta|$ for this interval. 
\bi
\item As a result, the bound on the ``noise'', $\bt$, seen by the projected-CS step is also the largest for this interval, we have $\|\b_t\| \le C (\zz \sqrt{r\lambda^+} +  |\sin \theta|\sqrt{\lambda_\ch})$. Using the CS guarantee, followed by ensuring exact support recovery (as explained above), this implies that $\et$ satisfies \eqref{etdef0} and that we get a similar bound on the final CS step error: $\|\e_t\| \le C (\zz \sqrt{r\lambda^+} +  0.11|\sin \theta|\sqrt{\lambda_\ch})$. The factor of $0.11$ in the second term of this bound is obtained because, for this interval, $\bpsi = \I - \Phat_{j-1}\Phat_{j-1}{}'$ and so
    $\bpsi \P_\new \approx \P_\new$ and $\P_\new$ is dense, see \eqref{dense_bnd}. Thus one can show that $\|\I_{\T_t}{}'\bpsi \P_\new\|_2 \le 0.11$. 
\item This bound on $\et$, along with using the critical fact that $\et$ satisfies \eqref{etdef0} (is sparse) and its support $\T_t$ changes enough (the $\outfracrow^\alpha$ bound of  Theorem \ref{thm1} holds), ensures that we get a better estimate of $\P_\rot$ after the first projection-SVD step. This is what allows us to apply Theorem \ref{thm_corpca}. Using it we can show that $\SE(\tPhat_{(t)}, \P_\rot)= \SE([\Phat_{j-1}, \Phat_{j,\rot,1}],\P_\rot) \le  0.1\zz + 0.06|\sin\theta|$ for $t \in \J_2$. See proof of $k=1$ case of Lemma \ref{p_evd_lem} and Fact \ref{algebra_lem}. 
\ei

\item Thus we have a much better estimate of $\P_\rot$ for $t \in \J_2$ than for $\J_1$. 
Because of this, $\|\b_t\|$ is smaller, and hence $\|\et\|$ is smaller for $t \in \J_2$. This, along with the sparsity and changing support, $\T_t$, of $\et$, ensures an even better estimate at the second projection-SVD step. We can show that $\SE(\tPhat_{(t)}, \P_\rot) = \SE([\Phat_{j-1}, \Phat_{j,\rot,2}],\P_\rot) \le  0.1\zz  + 0.5 \cdot 0.06|\sin\theta|$ for $t \in \J_3$. See proof of $k>1$ case of Lemma \ref{p_evd_lem} and Fact \ref{algebra_lem}.

\item Proceeding this way, we show that $\SE(\tPhat_{(t)}, \P_\rot) = \SE([\Phat_{j-1}, \Phat_{j,\rot,k}],\P_\rot) \le 0.1\zz + 0.5^{k-2} (0.06|\sin\theta|) $ after the $k$-th projection-SVD step. Picking $K$ appropriately, gives $\SE(\tPhat_{(t)},\P_\rot) \le  \zz$ after $K$ steps, i.e., at $t=t_j+K\alpha$.
In all the above intervals, $\SE(\tPhat_{(t)},\tP_{(t)}) \le \zz + \SE(\tPhat_{(t)}, \P_\rot)$. Thus at $t=t_j+K\alpha$, $\SE(\tPhat_{(t)},\tP_{(t)}) \le 2 \zz$.

\item At  $t=t_j+K\alpha$, $\tPhat_{(t)}$ contains $(r+1)$ columns. The subspace re-estimation via simple SVD step re-estimates $\P_j$ in order to delete the deleted direction, $\P_\del$, from $\tPhat_{(t)}$. The output of this step is $\Phat_j$ (the final estimate of $\Span(\P_j)$).  Thus, at $t=t_j+K\alpha + \alphadel$, $\tPhat_{(t)} = \Phat_j$ and we can show that it satisfies $\SE(\tPhat_{(t)},\tP_{(t)})=\SE(\Phat_j,\P_j) \le \zz$. See Lemma \ref{del_evd}.
The re-estimation is done at this point because, for times $t$ in this interval, $\|\lhat_t - \lt\| = \|\et\|  \le 2.4 \zz \|\lt\|$. For PCA in data-dependent noise, simple SVD needs $\alpha \ge (q/\epsilon)^2 f^2 (r \log n)$ where $q$ is the error/noise to signal ratio and $\epsilon$ is the final desired error level. For our problem, the ``noise" is $\et$ and thus $q=2.4 \zz$ and $\epsilon = \zz$. Since $q/\epsilon$ is a constant, $\alpha \ge \alpha_* = C f^2 r \log n$ suffices when simple SVD is applied at this time.
%

\een
When $\vt \neq 0$, almost all of the above discussion remains the same. The reason is this: in the main theorem, we assume $\|\vt\|^2 \le 0.1 r \zz^2 \lambda^+$ with $c < 1$ and so even though we have to deal with $\vt$ in $\b_t$ and $\e_t$ expressions, and in the $\alpha$ expression, the changes required are only to the numerical constants. In Corollary \ref{gen_xmin}, we have carefully chosen the bound on $\|\vt\|$ to equal the bound on $\|\e_t\|$ modulo constants. Thus, once again, only numerical constants change, everything else remains the same.

\subsection{Why automatic subspace change detection and Automatic Simple-ReProCS works} \label{det_works}
The subspace change detection approach is summarized in Algorithm \ref{simp_reprocs_auto}. This idea is motivated by a similar idea first used in our earlier works \cite{rrpcp_isit15,rrpcp_aistats}.
The algorithm toggles between the ``detect'' phase and the ``update'' phase. It starts in the ``detect'' phase. If the $j$-th subspace change is detected at time $t$, we set  $\that_j=t$. At this time, the algorithm enters the ``update'' (subspace update) phase. We then repeat the $K$ projection-SVD steps and the one subspace re-estimation via simple SVD step from Algorithm \ref{simp_reprocs_tj} with the following change: the $k$-th projection-SVD step is now done at $t = \that_j + k \alpha-1$ (instead of at $t=t_j+ k\alpha-1$) and the subspace re-estimation is done at $t= \that_j + K \alpha + \alphadel-1:=\that_{j,fin}$.
Thus, at $t =\that_{j,fin}$, the subspace update is complete. At this time, the algorithm enters the ``detect'' phase again.
%
To understand the change detection strategy, consider the $j$-th subspace change. Assume that the previous subspace $\P_{j-1}$ has been accurately estimated by $t= \that_{j-1,fin}$ and that $\that_{j-1,fin} < t_j$. Let $\Phat_*:= \Phat_{j-1}$ denote this estimate.
At this time, the algorithm enters the ``detect'' phase in order to detect the next ($j$-th) change.
Let $\B_t:=(\I-\Phat_*\Phat_*{}')  [\lhat_{t-\alpha+1}, \dots, \lhat_{t}]$. For every $t = \that_{j-1,fin} + u \alpha$, $u=1,2,\dots$, we detect change by checking if the maximum singular value of $\B_t$ is above a pre-set threshold, $ \sqrt{\lthres \alpha}$, or not.

We claim that, whp, under Theorem \ref{thm1} assumptions, this strategy has no false detects and correctly detects change within a delay of at most $2\alpha$ frames.
The former is true because, for any $t$ for which $[t-\alpha+1, t] \subseteq [ \that_{j-1,fin}, t_j)$, all singular values of the matrix $\B_t$ will be close to zero (will be of order $\sqrt{\zz}$) and hence its maximum singular value will be below $ \sqrt{\lthres \alpha}$. Thus, whp, $\that_j \ge t_j$.
%
%
To understand why the change {\em is} correctly detected within $2 \alpha$ frames, first consider $t =\that_{j-1,fin} + \left\lceil \frac{t_j - \that_{j-1,fin}}{\alpha} \right\rceil \alpha:= t_{j,*} $.
Since we assumed that $\that_{j-1,fin}< t_j$ (the previous subspace update is complete before the next change), $t_j$ lie in the interval $[t_{j,*}-\alpha+1, t_{j,*}]$. Thus, not all of the $\lt$'s in this interval will be generated from $\Span(\P_j)$. Thus, depending on where in the interval $t_j$ lies, the algorithm may or may not detect the change at this time.
 However, in the {\em next} interval, i.e., for $t \in [t_{j,*}+1, t_{j,*} + \alpha]$,  all of the $\lt$'s will be generated from $\Span(\P_j)$. We can prove that, whp, $\B_t$ for this time $t$ {\em will} have maximum singular value that is above the threshold.
 Thus, if the change is not detected at $t_{j,*}$, whp, it {\em will} get detected at $t_{j,*} + \alpha$.
Hence one can show that, whp, either $\that_j = t_{j,*}$, or $\that_j =t_{j,*} + \alpha$, i.e., $t_j \le \that_j \le t_j + 2\alpha$. To see the actual proof of these claims, please refer to Appendix \ref{proof_auto_thm1} where we prove our main result without assuming $t_j$ known.

\begin{table*}[t!]
\caption{List of Symbols and Assumptions used in the Main Result \ref{thm1}, and Corollary \ref{gen_xmin}. (Note: We show that whp, $\that_j \geq t_j$ and $\that_j + (K+1)\alpha \leq t_{j+1}$ and hence, whp, $\J_0, \J_{K+2}$ are non-empty intervals.}
\centering
\resizebox{\linewidth}{!}{
\renewcommand{\arraystretch}{1.5}
\begin{tabular}{cccc} \toprule
\multicolumn{4}{c}{{\large {\bf Observations:} $\yt = \lt + \xt + \vt$,  where, $\lt = \P_{(t)} \at = \begin{bmatrix}
\P_{j-1, \fx} & \P_{j, \add}
\end{bmatrix} \begin{bmatrix} \a_{t, \fx} \\ \a_{t, \ch} \end{bmatrix}$ for $t \in [t_j, t_{j+1})$, $\T_t$ is support of $\xt$}.} \\ \toprule
\multicolumn{2}{c}{{\bf Subspace Change}} & \multicolumn{2}{c}{{\bf Principal Subspace Coefficients,} $\at$'s} \\ \cmidrule(lr){1-2} \cmidrule(lr){3-4}
\multicolumn{2}{c}{$\tP_{(t)} = \tP_{(t_j)} = \P_j \ \forall t \in [t_j, t_{j+1}), \ j=0,1,\dots, J$} & \multicolumn{2}{c}{element-wise bounded, zero mean,}  \\
$\P_\ch \equiv \P_{j-1, \ch}$ & Changing direction from $\Span(\P_{j-1})$ at $t_j$ & \multicolumn{2}{c}{mutually independent with identical covariance (See Sec. \ref{assu_main})}   \\
$\P_\fx \equiv \P_{j-1, \fx}$ & Fixed directions from $\Span(\P_{j-1})$ at $t_j$ & \multicolumn{2}{c}{$\ep{[\at \at{}']} := \Lam = \begin{bmatrix} \Lam_\fx & \bm{0} \\ \bm{0} & \lambda_\ch\end{bmatrix}$}   \\
$\P_\new \equiv \P_{j-1, \new}$ & New direction from $\Span(\P_{j-1, \perp})$ added at $t_j$ & $\lambda^+$ &  $\lambda_{\max}(\Lam)$   \\
$\P_\rot \equiv \P_{j-1, \rot}$ & Rotated version of $\P_\ch$& $\lambda^-$ & $\lambda_{\min}(\Lam)$  \\
\multicolumn{2}{c}{{$\SE(\P_{j-1},\P_j) = \SE(\P_{j-1,\ch},\P_{j,\add}) \le \Delta$}}  & $f := \lambda^+/\lambda^-$ & Condition Number of $\Lam$  \\
\multicolumn{2}{c}{{$\P_{j,\new}:=  \frac{(\I - \P_{j-1,\ch} \P_{j-1,\ch}{}') \P_{j,\add}}{\SE(\P_{j-1,\ch},\P_{j,\add})}$}} \\
\multicolumn{2}{c}{See \eqref{ss_ch} for equivalent generative model.} &  \\
\multicolumn{2}{c}{$\max_{j} \max_{i} \| \basis([\P_{j-1},\P_{j}])^i\| \le \sqrt{\frac{\mu (r+1)}{n}}$ which implies \eqref{dense_bnd} holds.} \\ \midrule
\multicolumn{2}{c}{{\bf Outliers}} & \multicolumn{2}{c}{{\bf Intervals for $j$-th subspace change and tracking}} \\
\cmidrule(lr){1-2} \cmidrule(lr){3-4}
$\xmint:= \min_{i \in \T_t}|(\xt)_i|$ & Min. outlier magnitude at $t$  & $\J_0 := [t_j, \that_j)$ & interval before change detected \\
$\xmin:=\min_t \xmint$ & Min. outlier magnitude & $\J_k := [\that_j + (k-1)\alpha, \that_j + k\alpha)$ & $k$-th subspace update interval \\
$s := \outfraccol \cdot n$ & Cardinality of support set of $\xt$ & $\J_{K+1}:=[\that_j + K\alpha, \that_j + (K + 1)\alpha)$ & SVD-re-estimation interval\\
$\outfracrow^{\alpha} \leq \rrow$ & See \eqref{def_outfracrow} & $\J_{K+2}:=[\that_j + (K + 1)\alpha, t_{j+1})$ &  Final interval\\
$\rrow = 0.01/f^2$ \\
$\outfraccol \leq \rcol = \frac{0.01}{2\mu(r+1)}$ & See Theorem \ref{thm1} \\ \bottomrule
\end{tabular}
\label{tab:not1}}
\end{table*}

\begin{table*}
\caption{List of symbols and their associated meaning for understanding the proof of Theorems \ref{thm1} and \ref{thm_corpca}. The complete definitions can be found in Definitions \ref{def:thm_known} and \ref{def:thm_unknown}. We also provide the location of the proof for each of events/scalars where applicable in parenthesis.
}
\centering
\resizebox{\linewidth}{!}{
\renewcommand{\arraystretch}{1}
\begin{tabular}{ll} \toprule
{\bf Symbol } & {\bf Meaning }\\ \midrule
\multicolumn{2}{c}{{\bf \cred Preliminaries (Stated in Definitions \ref{def:thm_known} and \ref{def:thm_unknown})}} \\ \midrule
$\theta: = \theta_j$ & Angle of $j$-th subspace change \\
$\P_*:=\P_{j-1}$, $\P_\new:=\P_{j,\new}$, $\P_\ch := \P_{j-1,\ch}$, $\P_\fx := \P_{j-1,\fx}$ &  Parts of the $j$-th subspace \\
$\P_\rot: = \P_{j,\rot}:=(\P_{\ch} \cos \theta + \P_{\new} \sin \theta)$,  $\P:=\P_j$ &  \\
$\Phat_*:=\Phat_{j-1}$, $\Phat:=\Phat_j$ & Estimates of $j$-th subspace \\
$\Phat_{\rot,k}:=\Phat_{j,\rot,k}$ & $k$-th estimate of $\P_\rot$. \\
$\et := \I_{\T_t} (\bm\Psi_{\T_t}{}'\bm\Psi_{\T_t})^{-1} \I_{\Tt}{}' \bm\Psi (\lt + \vt)$ (Proved in Lemma \ref{CSlem}) & Expression for error, $\et = \xhat_t - \xt = \lt -\lhat + \vt$ \\ \midrule
\multicolumn{2}{c}{{\bf \cred Scalars (Derived in Fact \ref{algebra_lem})}} \\ \midrule
$\zeta_{0}^+:= \zz + |\sin \theta|$ & Bound on $\SE(\Phat, \P_\ast)$, i.e., before subspace update \\
$\zeta_{1}^+:= 0.4 \cdot 1.2  ((0.1 + \zz) |\sin \theta| + \zz) + 0.11 \zz$ & Bound on $\SE([\Phat_\ast, \Phat_{\rot, 1}], \P_\rot)$, i.e., after $1$st subspace update \\
$\zeta_{k}^+:=  0.4 \cdot (1.2 \zeta_{k-1}^+ ) + 0.11 \zz$ & Bound on $\SE([\Phat_\ast, \Phat_{\rot, k}], \P_\rot)$, i.e., after $k$-th subspace update \\ \midrule
\multicolumn{2}{c}{{\bf \cred Events -- $t_j$ known (Proved in Sec. \ref{proof_thm1})}} \\ \midrule
$\Gamma_0:=\{\SE(\Phat_*,\P_*) \le \zz\}$ & Previous subspace, $\P_\ast$ is $\zz$-accurately estimated \\
$\Gamma_k:= \Gamma_{k-1} \cap \{\SE([\Phat_*, \Phat_{\rot,k}],\P_\rot) \le \zeta_{k}^+ \}$ & All $k$ subspace update steps work \\
$\Gamma_{K+1}:= \Gamma_K \cap \{ \SE(\Phat,\P) \le \zz \}$ & Current subspace, $\P$ is $\zz$-accurately estimated \\ \midrule
\multicolumn{2}{c}{{\bf \cred Events -- $t_j$ unknown (This section is only used in Appendix \ref{proof_auto_thm1})}} \\  \midrule
$\that_{j-1,fin}: =  \that_{j-1} + (K + 1) \alpha -1$ & Time at which deletion step is complete \\
$t_{j,*}=\that_{j-1,fin} + \left\lceil \frac{t_j - \that_{j-1,fin}}{\alpha} \right\rceil \alpha$ & First possible time instant at which subspace change can be detected \\
$\mathrm{Det0}:= \{\that_j = t_{j,*} \}$ (Proved in Appendix \ref{proof_auto_thm1}) & Subspace Change detected within $\alpha$-frames \\
$\mathrm{Det1}:= \{\that_j = t_{j,*} + \alpha\}$ (Proved in Appendix \ref{proof_auto_thm1}) & Subspace Change detected after $\alpha$, but before $2\alpha$ frames \\
$\mathrm{Proj\SVD}_k:= \{\SE([\Phat_{\ast},\Phat_{\rot,k}]) \le \zeta_{\rot,k}^+\}$ (Proved in Lemma \ref{p_evd_lem}) & $k$-th Proj SVD works \\
$\mathrm{Proj\SVD}:=\cap_{k=1}^K  \mathrm{Proj\SVD}_k$ & All $K$ Proj SVD steps work \\
$\mathrm{Del}:=\{\SE(\Phat, \P) \le \zz \}$ (Proved in Lemma \ref{del_evd}) & Deletion Step works \\
$\mathrm{NoFalseDets}$ (Proved in Appendix \ref{proof_auto_thm1}) & No false detection of subspace change \\
$\Gamma_{0,\ed}:= \{\SE(\Phat_*, \P_*) \le \zz \}$ (Proved in Claim \ref{claim:equiv}) & Previous subspace estimated within $\zz$-accuracy\\
$\Gamma_{j,\ed}$ (Proved in Claim \ref{claim:equiv}) & All previous $j$ subspaces estimated within $\zz$-accuracy
 \\ \midrule
\multicolumn{2}{c}{{\bf \cred Notation for PCA in data-dependent noise: Theorem \ref{thm_corpca} (Proved in Appendix \ref{proof_thm_corpca})}} \\ \midrule
$\yt = \lt + \wt + \bm{z}_t$ & Observations: True data - $\lt = \P \at = [\P_\fx, \P_\rot] \at$; \\
 & Data-dep noise - $\wt$; Modeling error - $\bm{z}_t$ \\
$\wt = \bm{M}_t \lt = \mtt \mot \lt$ & Data-dependent noise, \\
$\|\mot \P_\ast\| \leq q_0$ and $\|\mot \P_\rot\| \leq q_\rot$ & Assumptions on Data-dependency matrices \\
$\|\mtt\| \leq 1$ and $\norm{\frac{1}{\alpha} \sum_t \mtt \mtt{}'} \leq b_0 = 0.01$ \\
$\|\bm{z}_t\| \leq b_z := q_0 \sqrt{r \lambda^+} + q_\rot \sqrt{\lcp}$, $\|\ep{[\bm{z}_t \bm{z}_t}\| \leq \lambda^+_z := b_z^2/r$ & Assumptions on modeling error. \\
 \bottomrule
\end{tabular}
\label{tab:not2}
}
\end{table*}

\section{Proving Theorem \ref{thm1} with assuming $\that_j = t_j$}\label{proof_thm1_section}


In  Table \ref{tab:not2}, we summarize all the new symbols and terms used in our proof. This, along with Table \ref{tab:not1} given earlier, should help follow the proof details without having to refer back to earlier sections.
To give a simpler proof first, we prove Theorem \ref{thm1} under the assumption that $\that_j= t_j$ below. The proof without this assumption is given in Appendix \ref{proof_auto_thm1}. With assuming $\that_j = t_j$, we are studying Algorithm \ref{simp_reprocs_tj}. Recall from Sec. \ref{proof_idea} that the subspace update step involves solving a problem of PCA in data-dependent noise when partial subspace knowledge is available. We provide a guarantee for this problem in Sec. \ref{corpca_sec} and use it in our proof in Sec. \ref{proof_thm1}.

For the entire proof we will use the equivalent subspace change model described in \eqref{ss_ch}. Clearly $|\sin \theta_j| \le \Delta$ by our assumption.

\color{black}
\subsection{PCA in data-dependent noise with partial subspace knowledge} \label{corpca_sec}
\renewcommand{\zt}{\bm{z}_t}
\subsubsection{The Problem}
We are given a set of $\alpha$ frames of observed data $\yt: = \lt + \wt + \zt$, with $\|\zt\|^2 \leq b_z^2$ and $\|\ep[\zt \zt{}']\| \leq \lambda_z^+ := c_1 b_z^2/ r$ ; $\wt = \M_t \lt$, $\lt = \P \at$ and with $\P$ satisfying
\[
\P =  [\P_{\fx}, \P_{\rot}], \text{ where } \P_\rot:=(\P_{\ch} \cos \theta  +  \P_{\new} \sin \theta),
\]
$\P_{\fx} = (\P_{*} \U_{0}) \I_{[1,r-1]}$, $\P_{\ch} = (\P_{*} \U_{0}) \I_{r}$, and $\U_0$ is an $r \times r$ rotation matrix.
Also, the $\at$'s are zero mean, mutually independent, element-wise bounded r.v.'s with identical and diagonal covariance $\Lam$, and independent of the matrices $\M_t$.
The matrices $\M_t$ are {\em unknown}.

Let $\J^\alpha$ denote the $\alpha$-frame time interval for which the $\yt$'s are available.
We also have access to a partial subspace estimate $\Phat_*$ that satisfies $\SE(\Phat_*, \P_*) \le \zz$, and that is computed using data that is independent of the $\lt$'s (and hence of the $\yt$'s) for $t \in \J^\alpha$.
The goal is to estimate $\Span(\P)$ using $\Phat_*$ and the $\yt$'s for $t \in \J^\alpha$.

\subsubsection{Projection-SVD / Projection-EVD}
Let $\bm\Phi: = \I - \Phat_* \Phat_*{}'$. A natural way to estimate $\P$ is to first compute $\Phat_\rot$ as the top eigenvector of
\[
\bm{D}_{obs}:= \frac{1}{\alpha} \sum_{t \in \J^\alpha} \bm\Phi \yt \yt{}' \bm\Phi.
\]
and set $\Phat = [\Phat_*, \Phat_\rot]$. We refer to this strategy as ``projection-EVD'' or ``projection-SVD''. 
In this paper, we are restricting ourselves to only one changed direction and hence we compute only the top eigenvector (or left singular vector) of $\bm{D}_{obs}$. In general if there were $r_\ch>1$ directions, we would compute all eigenvectors with eigenvalues above a threshold, see \cite{rrpcp_perf, rrpcp_aistats}.


\subsubsection{The Guarantee}
We can prove the following about projection-SVD.
\begin{theorem}\label{thm_corpca}
Consider the above setting for an $\alpha \ge \alpha_0$ where
\[
\alpha_0 := C \eta \max(f (r \log n),  \eta f^2 (r + \log n) ).
\]
%
Assume that the $\M_t$'s can be decomposed as $\M_t = \M_{2,t} \M_{1,t}$ where $\M_{2,t}$ is such that $\norm{\M_{2,t}} \le 1$ but
\begin{align}
\norm{\frac{1}{\alpha} \sum_{t \in \J^\alpha} \M_{2,t} \M_{2,t}{}'} \leq b_0 = 0.01.
\label{M2t_bnd}
\end{align}
Let $q_0$ denote a bound on $\max_t \|\M_{1,t} \P_*\|$ and let $q_\rot$ denote a bound on $\max_t \|\M_{1,t} \P_\rot\| $, i.e., we have
$
\|\M_{1,t} \P_*\| \le q_0 \ \text{ and } \ \|\M_{1,t} \P_\rot\| \le q_\rot
$
for all $t \in \J^\alpha$.
Assume that
\begin{align}
q_0 \le 2 \zz, \  q_\rot \le 0.2 |\sin \theta|,  \zz f \le 0.01 |\sin \theta|, \text{ and } \nn \\
  b_z \le C (q_0 \sqrt{r \lambda^+} + q_\rot  \sqrt{\lcp} )
\label{extra_bnds}
\end{align}

Define the event $\estar:= \{\SE(\Phat_*, \P_*) \le \zz \}$.  The following hold.
\ben
\item Conditioned on $\estar$,  w.p. at least $1 - 12n^{-12}$, \\
$\SE(\Phat,\P) \le \zz + \SE(\Phat,\P_\rot)$ and
\begin{align*}
\SE(\Phat,\P_\rot) &\le (\zz + |\sin\theta|)  \frac{0.39\qa + 0.1 \zz}{|\sin\theta|} \\
&\le 1.01 |\sin\theta|  \frac{0.39\qa + 0.1 \zz}{|\sin\theta|} \\
&<  0.4 q_\rot + 0.11 \zz.
\end{align*}
\item Conditioned on $\estar$,  w.p. at least $1-12n^{-12}$,
\begin{align*}
& \lambda_{\max}(\bm{D}_{obs}) \\
& \ge (0.97 \sin^2 \theta - 0.4 q_\rot |\sin \theta| - 0.15 \zz |\sin \theta|) \lambda_\ch. 
\end{align*}
\een
For large $n,r$, $r\log n > r+\log n$. Thus the following simpler expression for $\alpha_0$ suffices: $\alpha \ge \alpha_0= C \eta^2 f^2 (r \log n)$.
\end{theorem}


\begin{remark} \label{remark_ezero}
Theorem \ref{thm_corpca} holds even when $\estar$ is replaced by $\ezero:= \estar \cap \tilde{\mathcal{E}}(Z)$ where $\tilde{\mathcal{E}}(Z)$ is an event that depends on a r.v. $Z$ that is such that the pair $\{\Phat_*,Z\}$ is still independent of the $\lt$'s (and hence of the $\yt$'s) for $t \in \J^\alpha$.
\end{remark}

\begin{proof}
The proof follows using a careful application of the Davis-Kahan $\sin\theta$ theorem \cite{davis_kahan} followed by using matrix Bernstein \cite{tail_bound} to bound the numerator terms in the $\sin\theta$ theorem bound and Vershynin's sub-Gaussian result \cite{vershynin} to bound the extra terms in its denominator. While the overall approach is similar to that used by \cite{pca_dd} for the basic correlated-PCA problem, this proof requires significantly more work. We give the proof in Appendix \ref{proof_thm_corpca}. The most important idea in the proof is the use of Cauchy-Schwarz to show that the time-averaged projected-data - noise correlation and time-averaged noise power are both $\sqrt{b_0}$ times their instantaneous values. We explain this next.
\end{proof}

In the result above, the bounds assumed in \eqref{extra_bnds} are not critical. They only help to get a simple expression for the subspace error bound. As will be evident from the proof, we can also get a (more complicated) guarantee without assuming \eqref{extra_bnds}, and with any value of $\rrow$.
The main assumption needed by Theorem \ref{thm_corpca} is \eqref{M2t_bnd} on the data-dependency matrices $\M_t$. This is required because the noise $\wt$ depends on the true data $\lt$ and hence the instantaneous values of both the noise power and of the signal-noise correlation (even after being projected orthogonal to $\Phat_*$) can be large if $\lambda_\ch$ is large.
However, \eqref{M2t_bnd} helps ensure that the time-averaged noise power and the time-averaged projected-signal-noise correlation are much smaller.
Using the definitions of $q_0$ and $q_\rot$, $\|\E[\wt \wt{}']\| \le q_0^2 \lambda^+  +  q_\rot^2 \lambda_\ch:= c_w$ and $\|\E[\bm\Phi \lt \wt{}']\| \le \zz q_0 \lambda^+ + (\zz + |\sin \theta|) q_\rot \lambda_\ch := c_{wl}$.
By appropriately applying Cauchy-Schwarz (Theorem \ref{CSmat}),
\[
\norm{\frac{1}{\alpha} \sum_{t \in \J^\alpha} \E[\w_t \w_t{}']} \le \sqrt{b_0}  c_w \text{ and }
\]
\[
\norm{\frac{1}{\alpha} \sum_{t \in \J^\alpha} \E[\bm\Phi \lt \wt{}' \bm\Phi]} \le \sqrt{b_0} c_{wl}.
\]
Since $b_0 = 0.01$, both bounds are $0.1$ times their instantaneous values. See proof of Lemma \ref{lem:concm} for their derivation.

For our problem, \eqref{M2t_bnd} holds because we can let $\M_{2,t} = \I_{\T_t}$ and $\M_{1,t}$ as the rest of the matrix multiplying $\lt$ in \eqref{etdef0}. Then, using the bound on outlier fractions per row from Theorem \ref{thm1}, it is not hard to see that $\norm{\frac{1}{\alpha} \sum_t \M_{2,t} \M_{2,t}{}'} \le \rrow$. We state this formally next in Lemma \ref{outlier-lemma}.
\begin{lem}\label{outlier-lemma}
Assume that the $\outfracrow^\alpha$ bound of Theorem \ref{thm1} holds. Then,  for any $\alpha$-length interval $\J^\alpha \subseteq [t_1,\tmax]$,
\begin{align*}
\norm{\frac{1}{\alpha}\sum_{t \in \J^\alpha} \I_{\T_t} \I_{\T_t}{}'} &= \gamma(\J^\alpha) \le \outfracrow^\alpha \\
&\le \rrow =0.01 / f^2 \le  0.01 = b_0.
\end{align*}
\end{lem}
\begin{proof}[Proof of Lemma \ref{outlier-lemma}]
The proof is straightforward. Let $\bm{C}_t := \itt \itt{}'$. Then,
\begin{align*}
\bm{C}_t = \mathrm{diag}\left((c_t)_1,\ (c_t)_2,\ \cdots,\ (c_t)_n \right),
\end{align*}
where
\begin{align*}
(c_t)_i=
\begin{cases}
1, & \ \text{if} \ i \in \Tt, \\
0, & \  \text{if} \ i \notin \Tt
\end{cases}
= \one_{\{i \in \Tt\}}.
\end{align*}
Since each $\bm{C}_t$ is diagonal, so is $\frac{1}{\alpha} \sum_{t \in \J^\alpha} \bm{C}_t$. The latter has diagonal entries given by
\begin{align*}
\left(\frac{1}{\alpha} \sum_t \bm{C}_t \right)_{i,i} = \frac{1}{\alpha} \sum_{t \in \J^\alpha} (c_t)_i = \frac{1}{\alpha} \sum_{t \in \J^\alpha} \one_{ \{i \in \T_t\} } 
\end{align*}
Thus,
\begin{align*}
\norm{\frac{1}{\alpha} \sum_t \bm{C}_t} = \max_{i=1,2,\dots,n} \left| \left(\frac{1}{\alpha} \sum_t \bm{C}_t \right)_{i,i}\right| \\
 = \max_{i=1,2,\dots,n} \frac{1}{\alpha} \sum_{t \in \J^\alpha} \one_{ \{i \in \T_t\} } = \gamma(\J^\alpha)
\end{align*}
where $\gamma(\J^\alpha)$ defined in \eqref{def_gamma_J} is the outlier fraction per row of $\X_{\J^\alpha}$. From Theorem \ref{thm1}, this is bounded by $\rrow$.
\end{proof}

\subsection{Two simple lemmas from \cite{rrpcp_perf}}
The following two lemmas taken from \cite{rrpcp_perf} will be used in proving the main lemmas that together imply Theorem \ref{thm1} with $\that_j=t_j$.
\begin{lem} \cite[Lemma 2.10]{rrpcp_perf} \label{hatswitch}
Suppose that $\bm{P}$, $\Phat$ and $\bm{Q}$ are three basis matrices. Also, $\bm{P}$ and $\Phat$ are of the same size, $\bm{Q}' \bm{P} = \bm{0}$ and $\|(\I-\Phat \Phat{}' ) \bm{P} \| = \zeta_*$. Then,
\ben
  \item $\|(\I-\Phat\Phat{}')\bm{P}\bm{P}'\| =\|( \I - \bm{P}\bm{P}' ) \Phat \Phat{}'\| =  \|( \I - \bm{P} \bm{P}' ) \Phat \| = \| ( \I - \Phat \Phat{}' ) \bm{P}\| =  \zeta_*$
  \item $\|\bm{P} \bm{P}' - \Phat \Phat{}'\| \leq 2 \|(\I-\Phat \Phat{}')\bm{P}\| = 2 \zeta_*$
  \item $\|\Phat{}' \bm{Q}\| \leq \zeta_*$ \label{lem_cross}
  \item $ \sqrt{1-\zeta_*^2} \leq \sigma_i\left((\I-\Phat \Phat{}')\bm{Q}\right)\leq 1 $
\een
\end{lem}

\begin{lem}\cite[Lemma 3.7]{rrpcp_perf}\label{kappadelta}
For an $n \times r$ basis matrix $\P$,
\ben
\item $\max_{|\T| \le s} \|\I_\T{}' \P\|^2 \le s \max_{i=1,2,\dots,r} \|\I_i{}' \P\|^2$
\item $\delta_s(\I - \P \P') = \max_{|\T| \le s} \|\I_\T{}' \P\|^2$
\item If $\P = [\P_1, \P_2]$ then $\|\I_\T{}' \P\|^2 \le \|\I_\T{}' \P_1\|^2 + \|\I_\T{}' \P_2\|^2$. Also, $\|\I_\T{}' \P_1\|^2 \le \|\I_\T{}' \P\|^2$.
\een
Recall that $\delta_s(\M)$ is the $s$-restricted isometry constant \cite{candes_rip} of a matrix $\M$, i.e., it is the smallest real number for which the following holds for all $s$-sparse vectors $\bm{z}$:
$(1-\delta_s)\|\bm{z}\|^2 \le \|\M \bm{z}\|^2 \le (1+\delta_s)\|\bm{z}\|^2$.
\end{lem}

\subsection{Definitions and main claim needed for Theorem \ref{thm1} and Corollary \ref{gen_xmin} with $\that_j=t_j$} 


\begin{definition} \label{def:thm_known} We will use the following definitions in our proof.
\ben
\item Let $\theta: = \theta_j$,
$\P_*:=\P_{j-1}$, $\P_\new:=\P_{j,\new}$, $\P_\ch := P_{j-1, \ch}$, $\P_\rot: = \P_{j,\rot}:=(\P_{\ch} \cos \theta + \P_{\new} \sin \theta)$,  $\P:=\P_j$.
Similarly define $\Phat_*:=\Phat_{j-1}$, $\Phat:=\Phat_j$, and
\\ let $\Phat_{\rot,k}:=\Phat_{j,\rot,k}$ denote the $k$-th estimate of $\P_\rot$ with $\Phat_{\rot,0} = [.]$.

\color{black}
\item The scalars
\\ $\zeta_{0}^+:= \zz + |\sin \theta|$,
\\ $\zeta_{1}^+:= 0.4 \cdot 1.2  ((0.1 + \zz) |\sin \theta| + \zz) + 0.11 \zz$ and
\\ $\zeta_{k}^+:=  0.4 \cdot (1.2 \zeta_{k-1}^+ ) + 0.11 \zz$ for $k=2,3,\dots,K$.
\\ We will show that these are high probability bounds on $\SE([\Phat_*, \Phat_{\rot,k}], \P_\rot)$.

\color{black}
\item The events
\\ $\Gamma_0:=\{\SE(\Phat_*,\P_*) \le \zz\}$: clearly $\Gamma_0$ implies that $\SE(\Phat_*,\P_\rot) \le \zeta_{0}^+:= \zz + |\sin \theta|$,
\\ $\Gamma_k:= \Gamma_{k-1} \cap \{\SE([\Phat_*, \Phat_{\rot,k}],\P_\rot) \le \zeta_{k}^+ \}$ for $k=1,2,\dots, K$,
and
\\ $\Gamma_{K+1}:= \Gamma_K \cap \{ \SE(\Phat,\P) \le \zz \}$.

\item The time intervals:
\\ $\J_k:=[t_j + (k-1)\alpha, t_j + k\alpha)$ for $k=1,2,\dots,K$: the projection-SVD intervals,
\\ $\J_{K+1}:=[t_j +K \alpha, t_j +(K + 1) \alpha)$: the subspace re-estimation interval,
\\ $\J_{K+2}:=[t_j +K \alpha+\alphadel, t_{j+1})$: the interval when the current subspace update is complete and before the next change.
\een
\end{definition}

We first prove the $\SE$ bounds of Theorem \ref{thm1}. With these, the other bounds follow easily. To obtain the $\SE$ bounds, we will be done if we prove the following claim.
\begin{claim}\label{claim:equiv}
Given $\SE(\Phat_*, \P_*) \le \zz$, w.p. at least $1-(K+1)12n^{-12}$, 
\ben
\item $\SE(\Phat_*,\P_\rot) \le \zeta_{0}^+ $,
$\SE([\Phat_*, \Phat_{\rot,1}],\P_\rot) \le \zeta_{1}^+ $,

\item for $k>1$, $\SE([\Phat_*, \Phat_{\rot,k}],\P_\rot) \le \zeta_{k}^+$,

\item and so $\SE([\Phat_*, \Phat_{\rot,K}],\P_\rot) \le \zeta_{K}^+ \le \zz$ (using definition of $K$) and $\SE([\Phat_*, \Phat_{\rot,K}],\P)\le 2\zz$.

\item  Further,  after the deletion step, $\SE(\Phat, \P) \le \zz$.
\een
Proving the above claim is equivalent to showing that $\Pr(\Gamma_{K+1} | \Gamma_0) \ge 1 - (K+1)12n^{-12}$.
\label{main_claim}
\end{claim}

This claim is an easy consequence of the three main lemmas and Fact \ref{algebra_lem} given below. Fact \ref{algebra_lem} provides simple upper bounds on $\zeta_{k}^+$ that will be used at various places. The first lemma, Lemma \ref{CSlem}, shows that, assuming that the ``subspace estimates so far are good enough'', the projected CS step ``works'' for the next $\alpha$ frames, i.e., for all $t \in \J_k$, $\That_t = \T_t$; $\et$ is sparse and supported on $\T_t$ and satisfies \eqref{etdef0}, and $\|\et\|$ is bounded.
The second lemma, Lemma \ref{p_evd_lem}, uses Lemma \ref{CSlem} and Theorem \ref{thm_corpca} to show that, assuming that the ``subspace estimates so far are good enough'', with high probability (whp), the subspace estimate at the next projection-SVD step is even better than the previous ones. Applying Lemma \ref{p_evd_lem} for each $k=1,2,\dots,K$ proves the first two parts of Claim \ref{main_claim}. The third part follows easily from the first two and the definition of $K$. The fourth part follows using Lemma \ref{del_evd}, which shows that, assuming that the $K$-th projection-SVD step produces a subspace estimate that is within $2\zz$ of the true subspace, the subspace re-estimation step returns an estimate that is within $\zz$ of the true subspace.%

\begin{fact}\label{algebra_lem}
Using $\zz \le \zz f \le 0.01 |\sin \theta|$,
\ben
\item $\zeta_{0}^+:= \zz + |\sin \theta| \le 1.01 |\sin \theta|$,
\item $ \zeta_{1}^+:= 0.4 \cdot 1.2 ((0.1 + \zz)|\sin\theta| + \zz) + 0.11 \zz \le 0.06|\sin\theta|$,
\item $ \zeta_{k}^+:= 0.4 \cdot 1.2 \zeta_{k-1}^+ + 0.11 \zz
\le 0.5^{k-1} \zeta_{1}^+ + \frac{0.11}{1-0.5} \zz
\le 0.5^{k-1} (0.06|\sin\theta|) + 0.11\zz
\le 0.03 |\sin \theta|
$
\een
\end{fact}

This claim essentially implies Theorem \ref{thm1} and Corollary \ref{gen_xmin} with $\that_j=t_j$. We prove these without this assumption in Appendix \ref{proof_auto_thm1}.

\subsection{The three main lemmas needed to prove the main claim and their proofs} \label{proof_thm1}  
\color{black}
\begin{lem}[Projected CS] \label{CSlem}
Recall from Sec. \ref{proof_thm1_section} that $s$ is an upper bound on $|\T_t|$.
Under assumptions of Theorem \ref{thm1} or  or Corollary \ref{gen_xmin}, the following hold for $k=1,2,\dots, K+2$.
Let $\bm\Psi_1:= \I - \Phat_* \Phat_*{}'$,
$\bm\Psi_k:= \I - \Phat_* \Phat_*{}' - \Phat_{\rot,k-1} \Phat_{\rot,k-1}{}'$ for $k=2,3,\dots,K+1$, and
$\bm\Psi_{K+2}:=\I - \Phat \Phat'$. From Algorithm \ref{simp_reprocs_tj},
\[
\bm\Psi = \bm\Psi_k \ \text{ for } \ t \in \J_k, \ k=1,2,\dots,K+2.
\]
Assume that $\Gamma_{k-1}$ holds. Then,
\ben
\item $\max_{|\T| \le 2s} \|\I_\T{}'\Phat_*\| \le 0.3 + \zz \le 0.31$.

\item $\max_{|\T| \le 2s} \|\I_\T{}'\Phat_{\rot,k-1}\| \le 0.1 + \zz + \frac{\zeta_{k-1}^+ + \zz}{|\sin \theta|} \le 0.1 + 0.01 + 0.04 < 0.15$.

\item $\delta_{2s}(\bm\Psi_1) 
\le 0.31^2  < 0.12$,
 $\delta_{2s}(\bm\Psi_k) \le  0.31^2 + 0.15^2 < 0.12$ for $k=2,3,\dots,K+2$,

\item for all $t \in \J_k$,  
$\| (\bm\Psi_{\T_t}{}' \bm\Psi_{\T_t})^{-1} \| \le  1.2$ 

\item for all $t \in \J_k$, $\That_t = \T_t$
\color{black}
\item for all $t \in \J_k$, $\et:=\xhat_t - \xt = \lt - \lhat_t + \vt$ satisfies \eqref{etdef0}
with $\bm\Psi = \bm\Psi_k$ for $t \in \J_k$

\item for $t \in \J_1$, $\|\et\| \le 2.4 \sqrt{\eta} (\zz\sqrt{r \lambda^+} +  0.11 \Delta \sqrt{\lambda_\ch})$;
\\ for $t \in \J_k$, $\|\et\| \le 2.4 \sqrt{\eta} (\zz\sqrt{r \lambda^+} +  \zeta_{k-1}^+ \sqrt{\lambda_\ch})$ for $k=2,3,\dots K$;
\\ for $t \in \J_{K+1}$, $\|\et\| \le 4.8 \sqrt{\eta} \zz\sqrt{r \lambda^+}$.
\\ for $t \in \J_{K+2}$, $\|\et\| \le 2.4 \sqrt{\eta} \zz\sqrt{r \lambda^+}$.
\color{black}
\een
\end{lem}

\begin{proof}
Since the noise bound of Theorem \ref{thm1} is much smaller or equal to those assumed by Corollary \ref{gen_xmin}, if we can prove the latter, we would have also proved the former.
Using the first claim of Lemma \ref{kappadelta}, the $\outfraccol$ bound of Theorem \ref{thm1} and the incoherence assumption \eqref{incoh_2} imply that, for any set $\T$ with $|\T|\le 2s$, 
\bea
\|\I_\T{}' \P_{*}\| \le 0.1 < 0.3  \text{ and } \|\I_\T{}' \P_{\new}\|  \le 0.1
\label{dense_bnd}
\eea
(In order to simplify our assumptions, we have simplified the incoherence/denseness assumption from what it was in the original version of this work; as a result, even the first term above is bounded by 0.1 (not 0.3 as before). To not have to change the rest of the proof given below, we still use the 0.3 bound in the writing below).
%
Using \eqref{dense_bnd},
for any set $\T$ with $|\T| \le 2s$,
\begin{align}
\|\I_\T{}' \Phat_* \| &\le  \|\I_\T{}' (\I - \P_* \P_*{}') \Phat_* \| + \|\I_\T{}' \P_* \P_*{}' \Phat_*\| \nn \\
&\le  \| (\I - \P_* \P_*{}') \Phat_* \| + \|\I_\T{}' \P_*\| \nn \\
&=  \| (\I - \Phat_* \Phat_*{}') \P_* \| + 0.3 \nn \\
&\le  \zz + 0.3 \le 0.31.
\end{align}
The second row used \eqref{dense_bnd}  and the following: since $\Phat_*$ and $\P_*$ have the same dimension, $\SE(\P_*, \Phat_*) = \SE(\Phat_*, \P_*)$ (follows by Lemma \ref{hatswitch}, item 1). The third row follows using the definition of event $\Gamma_{k-1}$.
Proceeding similarly for $\Phat_{\rot,k-1}$ and $\P_\new$ (both have the same dimension),
\begin{align*}
\|\I_\T{}'\Phat_{\rot,k-1}\| &\le  \|(\I - \P_\new \P_\new{}')\Phat_{\rot,k-1}  \| + \|\I_\T{}' \P_\new\| \nn \\
&=  \|(\I - \Phat_{\rot,k-1} \Phat_{\rot,k-1}{}') \P_\new \| + 0.1 \nn \\
&\le  \|(\I - \Phat_*\Phat_*{}' - \Phat_{\rot,k-1} \Phat_{\rot,k-1}{}') \P_\new \|  \nn \\
&+ \|\Phat_*{}' \P_\new\| + 0.1 \nn \\
&\le  \frac{\zeta_{k-1}^+ + \zz}{|\sin \theta|} + \zz + 0.1 \le  0.04  + 0.01 + 0.1 \\
&\le 0.15.
\nn
\end{align*}
The second row used item 1 of Lemma \ref{hatswitch} and \eqref{dense_bnd}. The third row used triangle inequality.
The last row follows using $\P_\new = \frac{\P_\rot - \P_{\ch} \cos \theta}{\sin \theta}$ and the definition of event $\Gamma_{k-1}$. Using this, triangle inequality and $|\cos\theta|\le 1$, we bound the first term. Using item 3 of Lemma \ref{hatswitch}, we bound the second term. The final bounds use Fact \ref{algebra_lem}.

To get the above bound, we use $\P_\new$ (and not $\P_\rot$) because 
 $\|\Phat_*{}' \P_\new\| \le \zz$ since $\P_*$ is orthogonal to $\P_\new$. But we do not have a small upper bound on $\|\Phat_*{}' \P_\rot\| $.

The third claim follows using the first two claims and Lemma \ref{kappadelta}.
The fourth claim follows from the third claim as follows:
\begin{align*}
\norm{\left(\bpsi_{\Tt}{}'\bpsi_{\Tt}\right)^{-1}} \leq \frac{1}{1 - \delta_s(\bpsi)} &\leq \frac{1}{1 - \delta_{2s}(\bpsi)} \\
& \leq \frac{1}{1- 0.12} < 1.2.
\end{align*}
The last three claims follow the approach of the proof of \cite[Lemma 6.4]{rrpcp_perf}. There are minor differences because we set $\xi$ a little differently now and because we assume $\vt \neq 0$. We provide the proof in Appendix \ref{proof_CSlem}.
\end{proof}

\begin{lem}[Projection-\SVD]\label{p_evd_lem}
Under the assumptions of Theorem \ref{thm1} or Corollary \ref{gen_xmin}, the following holds for $k=1,2,\dots K$.
Conditioned on $\Gamma_{k-1}$, w.p. at least $1 -  12 n^{-12}$,
$
\SE([\Phat_*, \Phat_{\rot,k}],\P_\rot) \le \zeta_{k}^+
$, i.e., $\Gamma_k$ holds.
\end{lem}

\color{black}
\begin{proof}
Since the noise bound of Theorem \ref{thm1} is much smaller or equal to those assumed by Corollary \ref{gen_xmin}, if we can prove the latter, we would have also proved the former.

Assume that $\Gamma_{k-1}$ holds.
The proof first uses Lemma \ref{CSlem} to get an expression for $\et = \lt- \lhat_t + \vt$ and then applies Theorem \ref{thm_corpca} with the modification given in Remark \ref{remark_ezero}.
%
Using Lemma \ref{CSlem}, for all $t \in \J_k$,
\begin{align*}
\lhat_t = \lt - \et + \vt  &= \lt - \I_{\T_t} (\bm\Psi_{\T_t}{}' \bm\Psi_{\T_t})^{-1} \I_{\T_t}{}' \bm\Psi (\lt + \vt) + \vt \\
&:= \lt - \e_{l,t} - \e_{v,t} + \vt
\end{align*}
where $\bm\Psi =  \I - \Phat_* \Phat_*{}' - \Phat_{\rot,k-1} \Phat_{\rot,k-1}{}'$ with $\Phat_{\rot,0} = [.]$.

In the $k$-th projection-SVD step, we use these $\lhat_t$'s and $\Phat_*$ to get a new estimate of $\P_\rot$ using projection-SVD.
To bound $\SE([\Phat_*, \Phat_{\rot,k}],\P_\rot)$, we apply Theorem \ref{thm_corpca} (Remark \ref{remark_ezero})\footnote{We use Remark \ref{remark_ezero} with $\estar \equiv \Gamma_0$, $Z \equiv \{\lhat_1,\lhat_2, \dots, \lhat_{t_j+(k-1)\alpha-1}\}$, and $\tilde{\mathcal{E}}(Z) = \Gamma_{k-1} \setminus \Gamma_0$.}
with
$\ezero \equiv \Gamma_{k-1}$, $\yt \equiv \lhat_t$, $\wt \equiv -\e_{l,t}$, $\zt \equiv -\e_{v, t} + \vt$, $\alpha \ge \alpha_0 \equiv \alpha_*$, and $\J^\alpha \equiv \J_k$. We can let $\M_{2,t} = -\I_{\T_t}$ which implies $b_0 \equiv \outfracrow^\alpha$ and $\M_{1,t} = (\bm\Psi_{\T_t}{}' \bm\Psi_{\T_t})^{-1} \I_{\T_t}{}' \bm\Psi$. Using the $\outfracrow^\alpha$ bound of Theorem \ref{thm1} and Lemma \ref{outlier-lemma}, the main assumption needed by Theorem \ref{thm_corpca}, \eqref{M2t_bnd}, holds.  
With $\P = \P_j$ satisfying \eqref{ss_ch}, and $\alpha_*$ defined in \eqref{def_alphas}, all the key assumptions of Theorem \ref{thm_corpca} hold. The simpler expression of $\alpha_*$ suffices because we treat $\eta$ as a numerical constant and so $f^2 (r \log n ) > f^2 (r + \log n)$ for large $n,r$.

We now just need to compute $q_0$ and $q_\rot$ for each $k$, ensure that they satisfy \eqref{extra_bnds}, and apply the result.
The computation for $k=1$ is different from the rest. When $k=1$, $\bm\Psi = \I - \Phat_* \Phat_*{}'$. Thus, using item 4 of Lemma \ref{CSlem} and the definition of event $\Gamma_{k-1}$, $\|\M_{1,t}\P_*\| \le 1.2 \zz = q_0$, $q_0 < 2 \zz$, and
\bea
\|\M_{1,t}\P_\rot\| & \le & 1.2  (\| \I_{\T_t}{}' (\I - \Phat_* \Phat_*{}')\P_\new\| |\sin \theta| + \zz) \nn \\ 
& \le & 1.2  (\| \I_{\T_t}{}'\P_\new\| +  \|\Phat_*{}'\P_\new\|) |\sin \theta| + 1.2 \zz \nn \\
& \le & 1.2  ( (0.1 + \zz) |\sin \theta| + \zz  )  = q_\rot. \nn
\eea
The third row follows using \eqref{dense_bnd} and $\|\Phat_*{}'\P_\new\| \le \zz$ (folows by item 3 of Lemma \ref{hatswitch}).
Using $\zz \le \zz f \le 0.01 |\sin \theta|$, clearly $q_\rot < 0.2 |\sin \theta|$. Finally, in this interval, the bound on $b_z$ is satisfied since $b_z = b_{v,t}$ and the expression for $b_{v,t}$ in $\J_1$ is equal to the required upper bound on $b_z$.
Applying Theorem \ref{thm_corpca}, $\SE([\Phat_*, \Phat_{\rot,1}],\P_\rot) \le 0.4 q_\rot + 0.11 \zz = 0.4 \cdot 1.2  ((0.1 + \zz) |\sin \theta| +  \zz) + 0.11 \zz = \zeta_{1}^+$.%

Consider $k>1$. Now $\bm\Psi = \I - \Phat_* \Phat_*{}' - \Phat_{\rot,k-1} \Phat_{\rot,k-1}{}'$. With this, we still have $\|\M_{1,t}\P_*\| \le 1.2 \zz = q_0$ and $q_0 < 2 \zz$. But, to bound $\|\M_{1,t}\P_\rot\|$ we cannot use the approach that worked for $k=1$. The reason is that $\|[\Phat_*, \Phat_{\rot,k-1}]'\P_\new\|$ is not small. 
However, instead, we can now use the fact that $[\Phat_*, \Phat_{\rot,k-1}]$ is a good estimate of $\P_\rot$, with $\SE([\Phat_*,\Phat_{\rot,k-1}], \P_\rot) \le \zeta_{k-1}^+$ (from definition of event $\Gamma_{k-1}$). Thus,
\bea
&\|\M_{1,t}\P_\rot\|  \le  1.2 \SE([\Phat_*,\Phat_{\rot,k-1}], \P_\rot) 
\\
&\le  1.2 \zeta_{k-1}^+ = q_\rot. \nn
\eea
By Fact \ref{algebra_lem}, $q_\rot < 0.2 |\sin \theta|$. Even in this interval, the required bound on $b_z$ holds. Thus, applying Theorem \ref{thm_corpca},
$\SE([\Phat_*, \Phat_{\rot,k}],\P_\rot)
\le 0.4 q_\rot + 0.11 \zz = 0.4 \cdot 1.2 \zeta_{k-1}^+ + 0.11 \zz = \zeta_{k}^+$.
\end{proof}
\color{black}

\begin{lem}[Simple SVD based subspace re-estimation]\label{del_evd}
Under the assumptions of Theorem \ref{thm1} or Corollary \ref{gen_xmin}, the following holds.
Conditioned on $\Gamma_{K}$, w.p. at least $1 -  12 n^{-12}$,
$\SE(\Phat,\P) \le \zz$, i.e., $\Gamma_{K+1}$ holds.
\end{lem}

\color{black}
\begin{proof}
Assume that $\Gamma_K$ holds.
Using Lemma \ref{CSlem}, for all $t \in \J_{K+1}$,
\begin{align*}
\lhat_t = \lt - \et + \vt &= \lt - \I_{\T_t} (\bm\Psi_{\T_t}{}' \bm\Psi_{\T_t})^{-1} \I_{\T_t}{}' \bm\Psi (\lt + \vt) + \vt  \\
&:= \lt -\e_{l,t} - \e_{v,t} + \vt
\end{align*}
where $\bm\Psi = \I - \Phat_* \Phat_*{}' - \Phat_{\rot,K}\Phat_{\rot,K}{}'$. Re-estimating the entire subspace using simple SVD applied to these $\lhat_t$'s
is an instance of correlated-PCA with $\yt \equiv \lhat_t$, $\wt \equiv -\e_{l,t}$ and $\zt \equiv -\e_{v,t} + \vt$. We can apply the following result for correlated-PCA \cite[Theorem 2.13]{pca_dd} to bound $\SE(\Phat,\P)$. Recall $\Phat$ contains the top $r$ eigenvectors of $\sum_{t \in \J_{K+1}} \lhat_t \lhat_t'$. The following is a simplified version of  \cite[Theorem 2.13]{pca_dd}. It follows by upper bounding $\lambda_{z,\P,\P_\perp}$ and $\lambda_{z,rest}^+$ by $\lambda_z^+$ and lower bound $\lambda_{z,\P}^-$ by zero in \cite[Theorem 2.13]{pca_dd}.

\begin{theorem}
For $t \in \J^\alpha$, we are given data vectors $\yt := \lt + \wt + \zt$ where $\wt= \M_t \lt$, $\lt = \P \at$ and $\zt$ is small unstructured noise. Let $\Phat$ be the matrix of top $r$ eigenvectors of $\frac{1}{\alpha} \sum_{t \in \J^\alpha} \yt \yt{}'$.
Assume that $\M_t$ can be decomposed as $\M_t = \M_{2,t} \M_{1,t}$ so that $\|\M_{2,t}\| \le 1$ but $\norm{\frac{1}{\alpha} \sum_t \M_{2,t} \M_{2,t}{}'} \leq b$ for a $b < 1$.
Let $q$ be an upper bound on $\max_{t \in \J^\alpha} \|\M_{1,t} \P\|$. We assume that $\|\zt\| \leq b_z$ and define $\| \ep{[\zt\zt{}']} \| \leq \lambda_z^+ := b_z^2/r$.
For an $\varepsilon_\SE > 0$, define
\begin{align*}
\alpha_0 := C \eta  \max\bigg(& f^2 ( r \log n)    \frac{q^2}{\varepsilon_\SE^2}, \\
&\frac{\lambda_z^+ q^2}{\lambda^- \varepsilon_\SE^2} f r (\log n),
\eta f^2 (r\log9 + 10 \log n) \bigg).
\end{align*}
If $\alpha \ge \alpha_0$, and $3 \sqrt{b} qf + \frac{\lambda_z^+}{\lambda^-} \le 0.46 \varepsilon_\SE,$
then, w.p. at least $1-12n^{-12}$,
$
\SE(\Phat,\P) \le \varepsilon_\SE.
$
\label{thm_pca_dd_Thm_2_13}.
\end{theorem}
%
Apply the above result with $\yt \equiv \lhat_t$, $\wt \equiv -\e_{l,t}$, $\zt \equiv - \e_{v,t} +  \vt$, $\alpha \ge \alpha_*$, and $\J^\alpha \equiv \J_{K+1}$.
From the expression for $\et$, we can let $\M_{2,t} \equiv - \I_{\T_t}$, $\M_{1,t}  \equiv (\bm\Psi_{\T_t}{}' \bm\Psi_{\T_t})^{-1} \I_{\T_t}{}' \bm\Psi$.
Next we compute $q$.
Since $\Gamma_K$ holds, $\SE([\Phat_*,\Phat_{\rot,K}], \P) \le \zz + \zeta_{K}^+ \le 2 \zz$. Thus, $\|\M_{1,t} \P\| \le 1.2 \SE([\Phat_*,\Phat_{\rot,K}], \P) \le 1.2 \cdot  2 \zz = q$. The final desired error is $\varepsilon_\SE = \zz$.
Using Lemma \ref{outlier-lemma} and the $\outfracrow^\alpha$ bound from Theorem \ref{thm1}, the bound on the time-average of $\M_{2,t} \M_{2,t}{}'$ holds with $b \equiv \rrow = \frac{0.01}{f^2} < \frac{0.5^2}{(3 \cdot 2.4 f)^2}$. Also, in this interval $b_z^2 = C r \lambda^+$ and so $\lambda_z = C \lambda^+$ and thus the second term in the $\alpha_0$ expression equals the first term. The third term can of course be ignored in the large $n,r$ regime.
Applying the above result with $\varepsilon_\SE = \zz$, and $q=2.4\zz$, we conclude the following:
for $\alpha \ge \alpha_*$, w.p. at least $1-12n^{-12}$, $\SE(\Phat,\P) \le \zz$.
The simpler expression of $\alpha_*$ suffices because $\eta$ is treated as a numerical constant and so $f^2 (r \log n ) > f^2 (r + \log n)$ for large $n,r$. Also under the assumption of Corollary \ref{gen_xmin}, the second term is dominated by the first term. 
\end{proof}

\color{black}
\begin{proof}[Proof of Claim \ref{main_claim}]
Lemma \ref{p_evd_lem} tells us that $\Pr(\Gamma_k | \Gamma_{k-1}) \ge 1 -  12 n^{-12}$.
Lemma \ref{del_evd} tells us that $\Pr(\Gamma_{K+1} | \Gamma_K) \ge  1-12n^{-12}$.
Thus,
$
\Pr(\Gamma_{K+1} | \Gamma_0) = \Pr(\Gamma_{K+1}, \Gamma_{K}, \dots \Gamma_1 | \Gamma_0)
= \Pr(\Gamma_1 | \Gamma_0) \Pr(\Gamma_2 |\Gamma_1)  \dots \Pr(\Gamma_{K+1}|\Gamma_{K}) \ge (1 - 12 n^{-12})^K (1-12n^{-12}).  
$
since $\Gamma_{K+1} \subseteq \Gamma_K \subseteq \Gamma_{K-1} \dots \subseteq \Gamma_0$. The result follows since $(1 - 12 n^{-12})^K (1-12n^{-12}) \ge 1 - (K+1) 12 n^{-12}$.
\end{proof}

\begin{proof}[Proof of Theorem \ref{thm1} with $\that_j=t_j$]
Define the events $\Gamma_{1,0}:= \{\SE(\Phat_0, \P_0) \le \zz \}$,  $\Gamma_{j,k}:=\Gamma_{j,k-1} \cap \{\SE([\Phat_{j-1},\Phat_{j,\rot,k}]) \le \zeta_{k}^+\}$, for $k=1,2,\dots,K$, $\Gamma_{j,K+1}:=\Gamma_{j,K} \cap \{\SE(\Phat_{j}, \P_{j}) \le \zz \}$ and $\Gamma_{j+1,0}:=\Gamma_{j,K+1} $.
We can state and prove Lemmas \ref{p_evd_lem} and \ref{del_evd} with $\Gamma_{k}$ replaced by $\Gamma_{j,k}$. Then Claim \ref{main_claim} implies that $\Pr(\Gamma_{j,K+1} | \Gamma_{j,0}) \ge 1 - 12n^{-12}$.
Using $\Gamma_{J,K+1} \subseteq \Gamma_{J-1,K+1} \dots \subseteq \Gamma_{1,K+1} \subseteq \Gamma_{1,0}$ and $\Gamma_{j+1,0}:=\Gamma_{j,K+1} $,
$\Pr(\Gamma_{J,K+1} | \Gamma_{1,0}) =\Pr(\Gamma_{J,K+1},\Gamma_{J-1,K+1}, \dots \Gamma_{1,K+1} | \Gamma_{1,0}) = \Pr(\Gamma_{1,K+1} | \Gamma_{1,0}) \Pr(\Gamma_{2,K+1} | \Gamma_{2,0}) \dots\Pr(\Gamma_{J,K+1} | \Gamma_{J,0}) \ge (1 - (K+1) 12 n^{-12})^J \ge 1 - J (K+1) 12 n^{-12} \ge 1 - d n^{-12}$.

Event $\Gamma_{J,K+1}$ implies that $\Gamma_{j,k}$ holds for all $j$ and for all $k$. Thus, all the $\SE$ bounds given in Theorem \ref{thm1} hold. Using Lemma \ref{CSlem}, $\That_t = \T_t$ for all the time intervals of interest, and the bounds on $\|\et\|$ hold. 
\end{proof}

\section{Empirical Evaluation} \label{sims_detail}
In this section we illustrate the superiority of s-ReProCS over existing state of the art methods on synthetic and real data. In particular, we consider the task of background subtraction. All time comparisons are performed on a Desktop Computer with Intel$^{\textsuperscript{\textregistered}}$ Xeon E$3$-$1240$ $8$-core CPU @ $3.50$GHz and $32$GB RAM. And all experiments with synthetic data are averaged over $100$ independent trials. All codes are available at \url{https://github.com/praneethmurthy/ReProCS}.

Similar experiments have been shown in the earlier ReProCS works (original-ReProCS) \cite{rrpcp_perf,rrpcp_tsp,rrpcp_isit15,rrpcp_aistats}.  The purpose of this section is to illustrate that, even though s-ReProCS is much simpler, is provably faster and memory efficient, and provably works under much simpler assumptions, its experimental performance is still similar to that of original-ReProCS. It outperforms existing works for the same classes of videos and simulated data for which original-reprocs outperforms them. 

\subsection{Synthetic Data}
Our first simulation experiment is done to illustrate the advantage of s-ReProCS over existing batch and online RPCA techniques. As explained earlier, because s-ReProCS exploits dynamics (slow subspace change), it is provably able to tolerate a much larger fraction of outliers per row than all the existing techniques without needing uniformly randomly generated support sets. When the number of subspace changes, $J$, is large, it also tolerates a significantly larger fraction of outliers per column. The latter is hard to demonstrate via simulations (making $J$ large will require a very long sequence). Thus we demonstrate only the former. Our second experiment shows results with using an i.i.d. Bernoulli model on support change (which is the model assumed in the other works).

One practical instance where outlier fractions per row can be larger than those per column is in the case of video moving objects that are either occasionally static or slow moving \cite{rrpcp_isit15,rrpcp_aistats}. The outlier support model for our first and second experiments is inspired by this example and the model used in \cite{rrpcp_isit15,rrpcp_aistats}. It models a 1D video consisting of a person/object of length $s$ pacing up and down in a room with frequent stops. The object is static for $\bbeta$ frames at a time and then moves down. It keeps moving down for a period of $\tau$ frames, after which it turns back up and does the same thing in the other direction. We let $\bbeta =  \lceil c_0 \tau \rceil$ for a $c_0<1$. With this model, for any interval of the form $[(k_1-1)\tau+1, k_2 \tau]$ for $k_1,k_2$ integers, the outlier fraction per row is bounded by $c_0$. For any general interval of length $\alpha \ge \tau$, this $\outfracrow^{\alpha}$ is still bounded by $2c_0$ while $\outfraccol$ is bounded by $s/n$.

\begin{sigmodel}\label{mod:moving_object}
Let $\bbeta = \lceil c_0 \tau \rceil$. Assume that the $\Tt$ satisfies the following. For the first $\tau$ frames (downward motion),
\begin{align*}
\mathcal{T}_t = \begin{cases}
[1,\ s], &\quad t \in [1, \beta] \\
[s+1,\ 2s], &\quad t \in [\beta + 1, 2 \beta] \\
\vdots \\
[(1/c_0 - 1)s + 1,\ s/c_0], &\quad t \in [\tau - \beta + 1, \tau] \\
\end{cases}
\end{align*}
for the next $\tau$ frames (upward motion), $\mathcal{T}_t = $
\begin{align*}
\begin{cases}
[(1/c_0 - 1)s + 1,\ s/c_0], &\ \ t \in [\tau + 1,\ \tau + \beta] \\
[(1/c_0 - 2)s + 1,\ (1/c_0 - 1) s], & \ \ t \in [\tau + \beta + 1, \tau + 2\beta] \\
\vdots \\
[1,\ s], & \ \ t \in [2\tau - \beta + 1, 2\tau]. \\
\end{cases}
\end{align*}
Starting at $t=2\tau+1$, the above pattern is repeated every $2\tau$ frames until the end, $t=\tmax$.
\end{sigmodel}
This model is motivated by the model assumed for the guarantees in older works \cite{rrpcp_isit15,rrpcp_aistats}. 
The above model is one practically motivated way to simulate data that is not not generated uniformly at random (or as i.i.d. Bernoulli, which is approximately the same as the uniform model for large $n$). It also provides a way to generate data with a different bounds on outlier fractions per row and per column. The maximum outlier fraction per column is $s/n$. For any time interval of length $\alpha \ge \tau$, the outlier fraction per row is bounded by $2c_0$. Thus, for Theorem \ref{thm1}, with this model, $\rrow = 2c_0/f^2$. By picking $2c_0$ larger than $s/n$ we can ensure larger outlier fractions per row than per column.

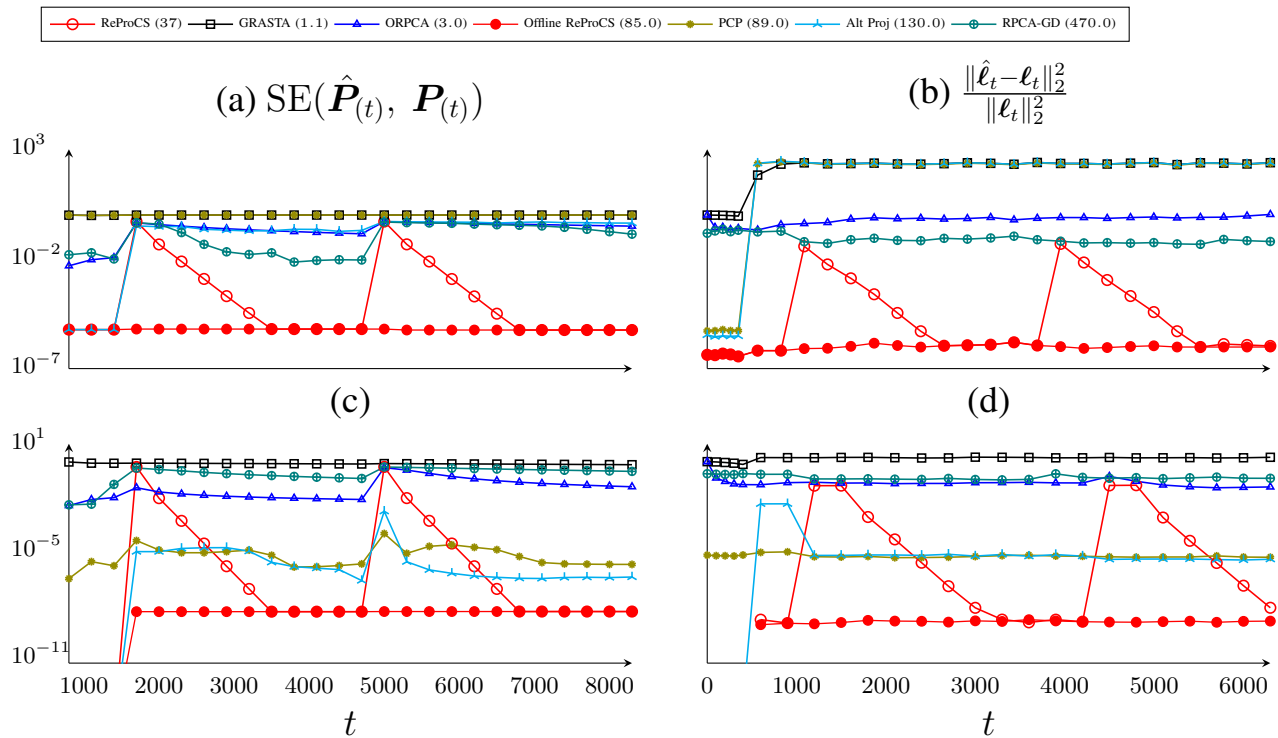
\begin{figure*}[t!]
\begin{center}
\begin{tikzpicture}
    \begin{groupplot}[
        group style={
            group size=2 by 2,
            x descriptions at=edge bottom,
            y descriptions at=edge left,
        },
        my stylecompare,
        %
        ymin=1e-7, ymax=1e3,
        enlargelimits=false,
    ]
       \nextgroupplot[
            my legend style compare,
            ymode=log,
            title={(a) $\SE(\tPhat_{(t)},\ \tP_{(t)})$}
        ]
	        \addplot table[x index = {0}, y index = {1}]{\senfivethousand};
	        \addplot table[x index = {0}, y index = {2}]{\senfivethousand};
	        \addplot table[x index = {0}, y index = {3}]{\senfivethousand};
	        \addplot table[x index = {0}, y index = {4}]{\senfivethousand};
	        \addplot table[x index = {0}, y index = {5}]{\senfivethousand};
	        \addplot table[x index = {0}, y index = {6}]{\senfivethousand};
	        \addplot table[x index = {0}, y index = {7}]{\senfivethousand};

        \nextgroupplot[
            ymode=log,
            title={(b) $\frac{\|\lhatt - \lt \|_2^2}{\|\lt\|_2^2}$}
        ]
	        \addplot table[x index = {0}, y index = {1}]{\ltnfivethousand};
	        \addplot table[x index = {0}, y index = {2}]{\ltnfivethousand};
	        \addplot table[x index = {0}, y index = {3}]{\ltnfivethousand};
	        \addplot table[x index = {0}, y index = {4}]{\ltnfivethousand};
	        \addplot table[x index = {0}, y index = {5}]{\ltnfivethousand};
	        \addplot table[x index = {0}, y index = {6}]{\ltnfivethousand};
	        \addplot table[x index = {0}, y index = {7}]{\ltnfivethousand};
	
       \nextgroupplot[
            ymode=log,
            ymin=5e-12, ymax=1e1,
            xlabel=$t$,
			title={(c)}
        ]
	
	        \addplot table[x index = {0}, y index = {1}]{\sebern};
	        \addplot table[x index = {0}, y index = {2}]{\sebern};
	        \addplot table[x index = {0}, y index = {3}]{\sebern};
	        \addplot table[x index = {0}, y index = {4}]{\sebern};
	        \addplot table[x index = {0}, y index = {5}]{\sebern};
	        \addplot table[x index = {0}, y index = {6}]{\sebern};
	        \addplot table[x index = {0}, y index = {7}]{\sebern};

        \nextgroupplot[
            ymode=log,
            ymin=5e-12, ymax=1e1,
            xlabel=$t$,
            title={(d)}
        ]
	        \addplot table[x index = {0}, y index = {1}]{\ltbern};
	        \addplot table[x index = {0}, y index = {2}]{\ltbern};
	        \addplot table[x index = {0}, y index = {3}]{\ltbern};
	        \addplot table[x index = {0}, y index = {4}]{\ltbern};
	        \addplot table[x index = {0}, y index = {5}]{\ltbern};
	        \addplot table[x index = {0}, y index = {6}]{\ltbern};
	        \addplot table[x index = {0}, y index = {7}]{\ltbern};
    \end{groupplot}
\end{tikzpicture}
\end{center}
\caption{First row ((a), (b)): Illustrate the subspace error and the normalized $\lt$ error for $n=5000$ and outlier supports generated using Model \ref{mod:moving_object}. Both the metrics are plotted every $k\alpha - 1$ time-frames. The results are averaged over $100$ iterations. Second row ((c), (d)) illustrate the subspace error and the normalized $\lt$ error for $n=500$ and Bernoulli outlier support model. They are plotted every $k\alpha - 1$ time-frames. The plots clearly corroborates the nearly-exponential decay of the subspace error as well as the error in $\lt$.}
\label{fig:Comparisonn5000}
\end{figure*}

We compare Algorithm \ref{auto_reprocs} and its offline counterpart with three of the batch methods with provably guarantees discussed in Sec. \ref{discuss} -  PCP \cite{rpca}, AltProj \cite{robpca_nonconvex} and RPCA-GD \cite{rpca_gd} - and with two recently proposed online algorithms known to have good experimental performance and for which code was available - ORPCA \cite{xu_nips2013_1} and GRASTA \cite{grass_undersampled}. The code for all these techniques are cloned from the Low-Rank and Sparse library (\url{https://github.com/andrewssobral/lrslibrary}).

For generating data we used $\tmax = 8000$, $t_\train = 500$, $J=2$, $r = 5$, $f = 16$ with $t_1 = 1000$, $\theta_1 = 30^{\circ}$, $t_2 = 4300$, $\theta_2 = 1.01\theta_1$ and varying $n$. The $\at$'s are zero mean i.i.d uniform random variables generated exactly as described before and so is $\tP_{(t)}$.
We generated a basis matrix $\bm{Q}$ by ortho-normalizing the first $r + 2$ columns of a $n \times n$ i.i.d. standard Gaussian matrix. For $t \in [1,\ t_1)$, we set $\tP_{(t)} = \P_0$ with $\P_0$ being the first $r$ columns of $\bm{Q}$. We let $\P_{1,\new}$ be $(r+1)$-th column of $\bm{Q}$, and rotated it in using \eqref{ss_ch} with $\U_{1} = \I$ and with angle $\theta_1$ to get $\P_1$. We set $\tP_{(t)} = \P_1$ for $t \in [t_1, t_2)$. We set $\P_{2, \new}$ to be the last column of $\bm{Q}$, $\U_2 = \I$, and rotate using angle $\theta_2$ just as done in the first subspace change and finally for $t\in [t_2, \tmax]$ we set $\P_{(t)} = \P_2$.
At all times $t$, we  let $\lt = \tP_{(t)} \at$ with $\at$ being zero mean, i.i.d uniform random variables such that $(\at)_i \sim unif\left(-\sqrt{f},\ \sqrt{f}\right)$ for $i = 1,\ \cdots,\ r - 2$ and $(\at)_{[r-1, r]} \sim unif(-1, 1)$. With this the condition number is $f$, the covariance matrix, $\Lambda = \mathrm{diag}(f,f, \cdots, f, 1, 1)/3$, $\lambda^+ = f/3$, $\lambda_\ch=\lambda^- = 1/3$, and $\eta=3$.
We generate $\T_t$ using Model \ref{mod:moving_object} as follows. For $t \in [t_\train, \tmax]$, we used $s = 0.1n$,  $c_0 = 0.2$ and $\tau = 100$. Thus $\rrow = 0.4 / f^2$.
For $t \in [1,t_\train]$, we used $s = 0.05n$ and $c_0 = 0.02$. This was done to ensure that AltProj (or any other batch technique works well for this period and provides a good initialization).
The magnitudes of the nonzero entries of $\xt$ (outliers) were generated s i.i.d uniform r.v.'s between $x_{\min} = 10$ and $x_{\max} = 25$.  


We implemented Algorithm \ref{auto_reprocs} for s-ReProCS (with initialization using AltProj) with $\alpha = C f^2 r \log n = 500$, $K = \lceil -0.8\log(0.9 \zz) \rceil = 5$, $\omega_{supp} =x_{\min}/2$, and $\omega_{evals} = 0.0025 \lambda^-$. We initialized using AltProj applied to $\Y_{[1,t_\train]}$. For the batch methods used in the comparisons -- PCP, AltProj and  RPCA-GD, we implement the algorithms on $\Y_{[1, t]}$ every $t= t_\train + k\alpha - 1$ frames. Further, we set the regularization parameter for PCP $\lambda = 1/\sqrt{n}$ in accordance with \cite{rpca}. The other known parameters, $r$ for Alt-Proj, outlier-fraction for RPCA-GD, are set using the true values. For online methods we implement the algorithms without modifications. The regularization parameter for ORPCA was set as with $\lambda_1 = 1 / \sqrt{n}$ and $\lambda_2 = 1 / \sqrt{d}$ according to \cite{xu_nips2013_1}.
We plot the subspace error and the normalized error of $\lt$ over time in Fig. \ref{fig:Comparisonn5000}(a) and \ref{fig:Comparisonn5000}(b) for $n=5000$. We display the time-averaged error for other values of $n$ in Table \ref{tab:results1}. This table also contains the time comparisons.

\begin{table*}[t!]
\centering
\caption{\small{Average subspace error $\SE(\tPhat_{(t)},\ \tP_{(t)})$ and time comparison for different values of signal size $n$. The values in brackets denote average time taken per frame (-- indicates that the algorithm does not work).}}
\resizebox{\linewidth}{!}{
\begin{tabular}{@{}llllllll@{}} \toprule
& \textbf{ReProCS} & GRASTA & ORPCA & \textbf{Offline ReProCS} & PCP & AltProj & RPCA-GD \\ 
\midrule
$n = 500$ (in $10^{-4}s$) & $\mathbf{0.066}$ ($\mathbf{3.1}$) & $0.996$ ($2.8$) & $0.320$ ($10$) & $\mathbf{8.25 \times 10^{-5}}$ ($\mathbf{6.3}$) & $1.00$ ($51$) & $0.176$ ($104$) & $0.215$ ($454$) \\ 
$n = 500$, Bern. (in $10^{-4}s$) & $\mathbf{0.044}$ ($\mathbf{4.8}$) & $0.747$ ($1.9$) & $0.078$ ($1.8$) & $\mathbf{3.9 \times 10^{-7}}$ ($\mathbf{9.2}$) & $1.2 \times 10^{-4}$ ($395$) & $0.0001$ ($32$) &$0.303$ ($329$) \\
$n = 5000$ (in $10^{-2}s$) & $\mathbf{0.048}$ ($\mathbf{3.7}$) & $0.999$ ($0.11$) & $0.322$ ($0.30$) & $\mathbf{6.05 \times 10^{-5}}$ ($\mathbf{8.5}$) & $0.999$ ($8.9$) & $0.354$ ($13.0$) & $0.223$ ($47.0$) \\
$n = 10,000$ (in $10^{-2}s$)  & $\mathbf{0.090}$ ($\mathbf{15.6}$) & $0.999$ ($0.25$) &  $0.3235$ ($0.68$) & $\mathbf{0.0006}$ ($\mathbf{36.8}$) & -- & -- & -- \\
\hline
\end{tabular}
}
\label{tab:results1}
\end{table*}

As can be seen, s-ReProCS outperforms all the other methods and offline s-ReProCS significantly outperforms all the other methods for this experiment. The reason is that the outlier fraction per row are quite large, but s-ReProCS exploits slow subspace change. In principle, even GRASTA exploits slow subspace change, however, it uses approximate methods for computing the SVD and does not use projection-SVD and hence it fails. s-ReProCS and offline s-ReProCS are faster than all the batch methods especially for large $n$. In fact when $n=10000$, the batch methods are out of memory and cannot work, while s-ReProCS still can. But s-ReProCS is slower than GRASTA and ORPCA.

{\em Comparison with other algorithms - random outlier support using the i.i.d. Bernoulli model. }
We generated data exactly as described above with the following change: $\T_t$ was now generated as i.i.d. Bernoulli with probability of any index $i$ being in $\cup_{t\in[1,n]} \T_t$ being $\rho_s = 0.02$ for the first $t_\train$ frames and $\rho_s = 0.2$ for the subsequent data. Notice that under the Bernoulli model, $\rrow = \rcol = \rho_s$. We used $n=500$. We show the results in Fig. \ref{fig:Comparisonn5000}(c) and \ref{fig:Comparisonn5000}(d).
For this experiment, the batch methods PCP and AltProj have good performance, that is better than s-ReProCS at most time instants. Offline s-ReProCS still outperforms all the other methods.

\begin{figure*}[t!]
\centering
\resizebox{.9\linewidth}{!}{
\begin{tabular}{@{}c@{}c@{}c@{}c@{}c@{}c@{}}
\tiny{Original} & \tiny{{\bf s-ReProCS}($\bm{16.5}${\bf ms})} & \tiny{AltProj ($26.0$ms)} & \tiny{RPCA-GD($29.5$ms)} & \tiny{GRASTA ($2.5$ms)} & \tiny{PCP ($44.6$ms)} \\
	\includegraphics[width=0.11\linewidth, height=1.4cm, trim={2cm, 0cm, 2cm, 0cm}, clip]{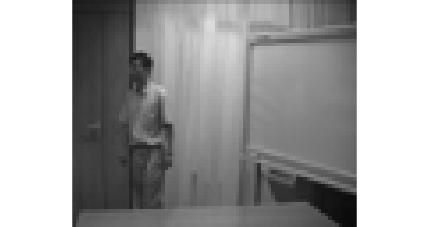}
&
	\includegraphics[width=0.11\linewidth, height=1.4cm, trim={1.7cm, 0cm, 1.7cm, 0cm}, clip]{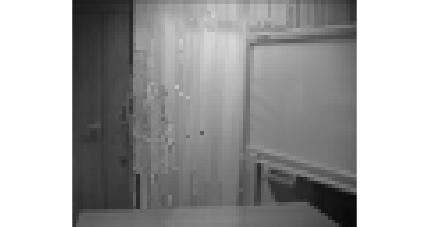}
&
	\includegraphics[width=0.11\linewidth, height=1.4cm, trim={1.8cm, 0cm, 1.8cm, 0cm}, clip]{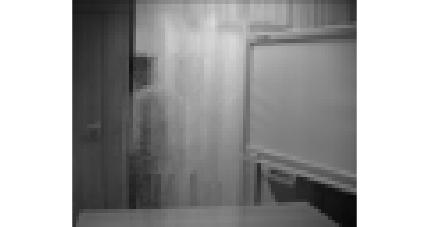}
&
	\includegraphics[width=0.11\linewidth, height=1.4cm, trim={1.8cm, 0cm, 1.8cm, 0cm}, clip]{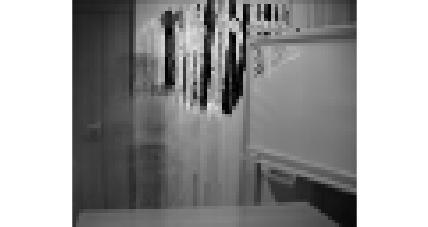}
&
	\includegraphics[width=0.11\linewidth, height=1.4cm, trim={1.2cm, 0cm, 1.25cm, 0cm}, clip]{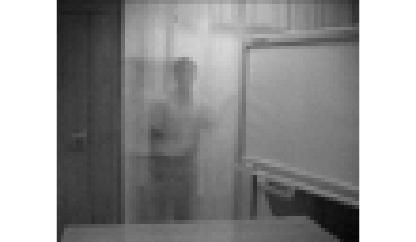}
&
	\includegraphics[width=0.11\linewidth, height=1.4cm, trim={1.4cm, 0cm, 1.4cm, 0cm}, clip]{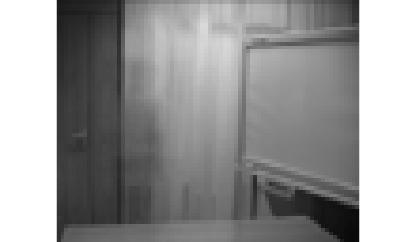}
\\    \newline
	\includegraphics[width=0.11\linewidth, height=1.4cm, trim={2cm, 0cm, 2cm, 0cm}, clip]{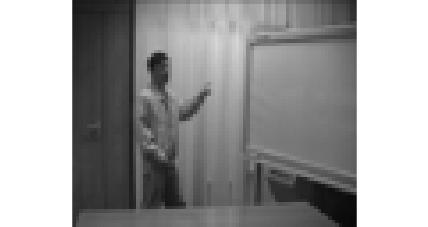}
&
	\includegraphics[width=0.11\linewidth, height=1.4cm, trim={1.7cm, 0cm, 1.7cm, 0cm}, clip]{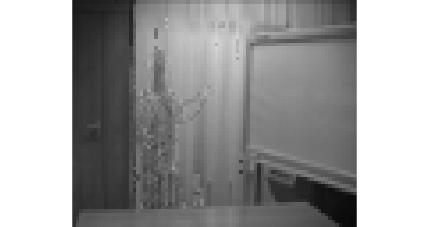}
&
	\includegraphics[width=0.11\linewidth, height=1.4cm, trim={1.2cm, 0cm, 1.2cm, 0cm}, clip]{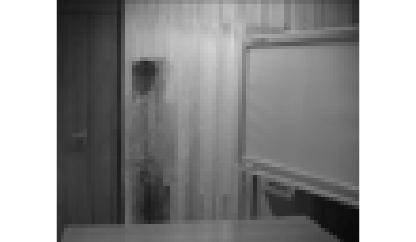}
&
	\includegraphics[width=0.11\linewidth, height=1.4cm, trim={1.8cm, 0cm, 1.8cm, 0cm}, clip]{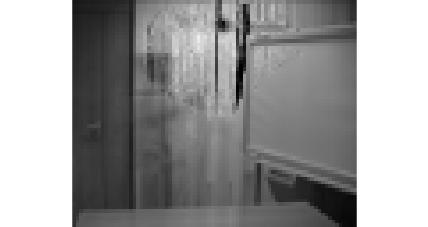}
&
	\includegraphics[width=0.11\linewidth, height=1.4cm, trim={1.2cm, 0cm, 1.25cm, 0cm}, clip]{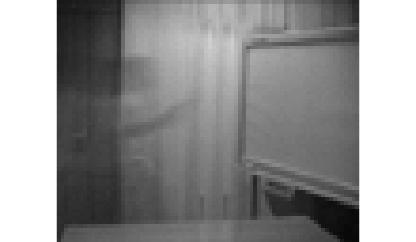}
&
	\includegraphics[width=0.11\linewidth, height=1.4cm, trim={1.4cm, 0cm, 1.4cm, 0cm}, clip]{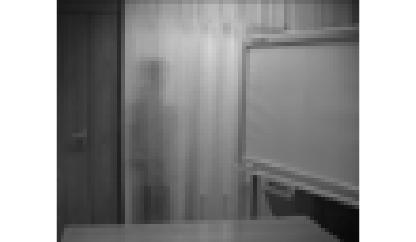}
\\    \newline
\tiny{Original} & \tiny{{\bf s-ReProCS}($\bm{85.4}${\bf ms})} & \tiny{AltProj($95.7$ms)} & \tiny{RPCA-GD($122.5$ms)} & \tiny{GRASTA ($22.6$ms)}  & \tiny{PCP ($318.3$ms)} \\
	\includegraphics[width=0.11\linewidth, height=1.4cm, trim={.7cm, 0cm, .7cm, 0cm}, clip]{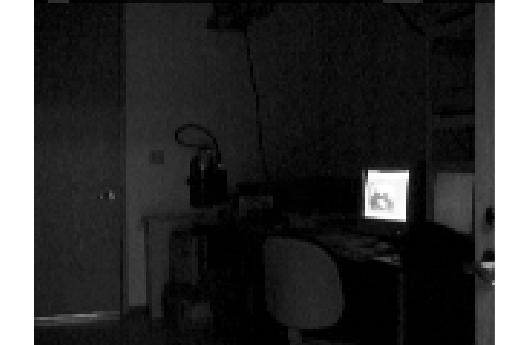}
&
	\includegraphics[width=0.11\linewidth, height=1.4cm, trim={.7cm, 0cm, .7cm, 0cm}, clip]{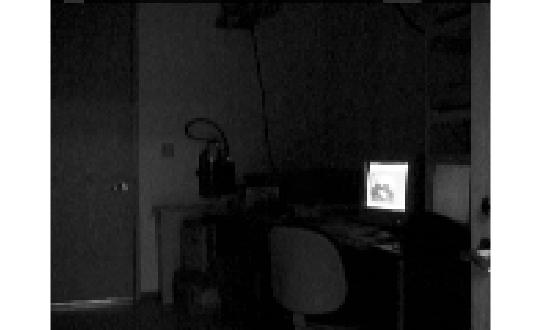}
&
	\includegraphics[width=0.11\linewidth, height=1.4cm, trim={1.5cm, 0cm, 1.5cm, 0cm}, clip]{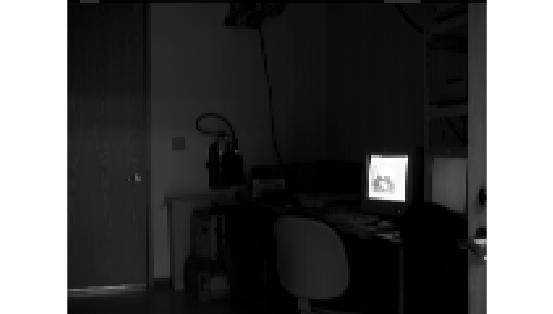}
&
	\includegraphics[width=0.11\linewidth, height=1.4cm, trim={1.4cm, 0cm, 1.4cm, 0cm}, clip]{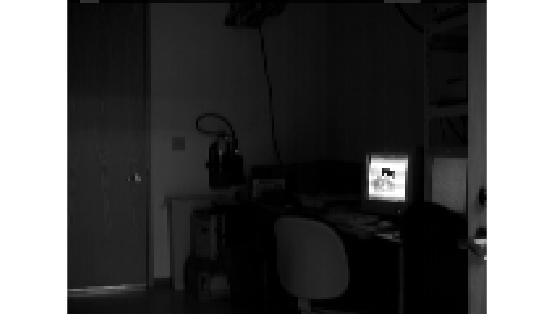}
&
	\includegraphics[width=0.11\linewidth, height=1.4cm, trim={.5cm, 0cm, .5cm, 0cm}, clip]{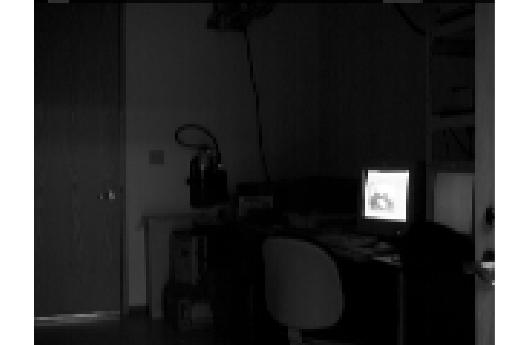}
&
	\includegraphics[width=0.11\linewidth, height=1.4cm, trim={.8cm, 0cm, .8cm, 0cm}, clip]{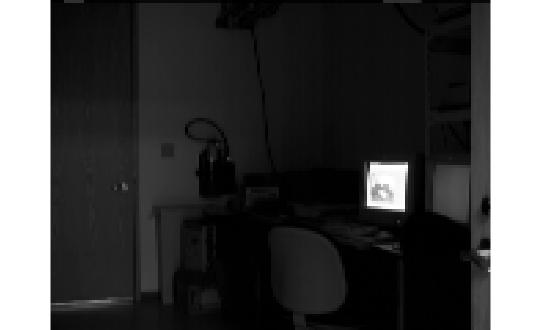}
\\    \newline
		\includegraphics[width=0.11\linewidth, height=1.4cm, trim={1.3cm, 0cm, 1.3cm, 0cm}, clip]{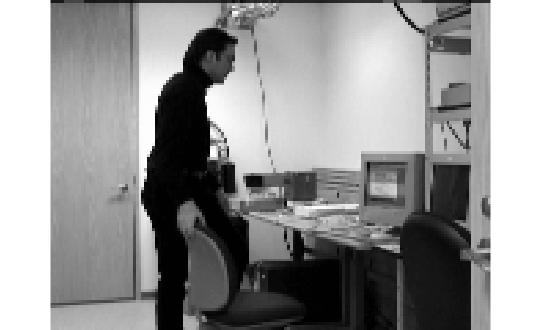}
&
	\includegraphics[width=0.11\linewidth, height=1.4cm, trim={1.3cm, 0cm, 1.3cm, 0cm}, clip]{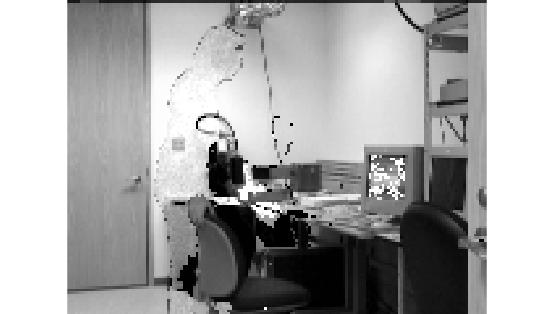}
&
	\includegraphics[width=0.11\linewidth, height=1.4cm, trim={1.5cm, 0cm, 1.5cm, 0cm}, clip]{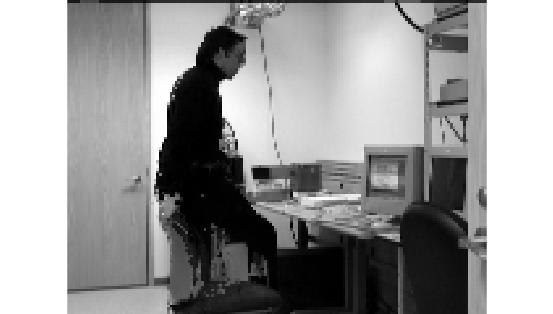}
&
	\includegraphics[width=0.11\linewidth, height=1.4cm, trim={.6cm, 0cm, .6cm, 0cm}, clip]{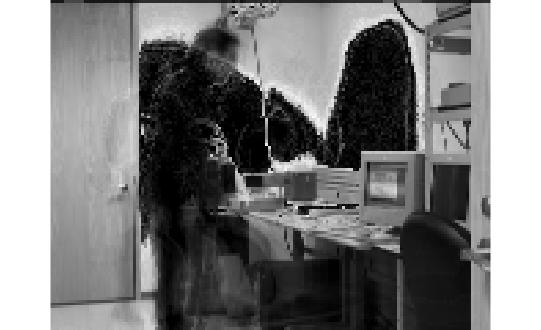}
&
	\includegraphics[width=0.11\linewidth, height=1.4cm, trim={.5cm, 0cm, .5cm, 0cm}, clip]{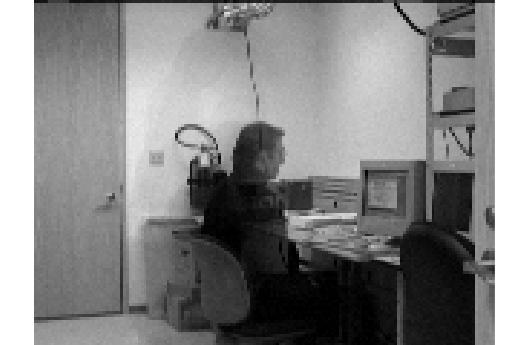}
&
	\includegraphics[width=0.11\linewidth, height=1.4cm, trim={.8cm, 0cm, .8cm, 0cm}, clip]{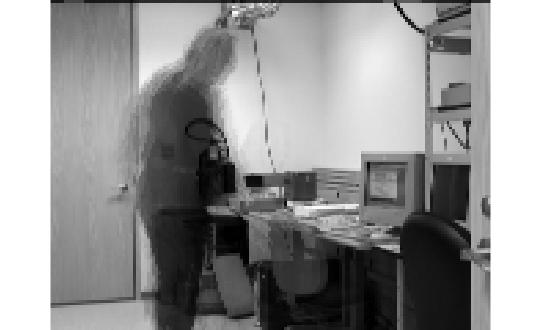}
\\    \newline
\tiny{Original} & \tiny{{\bf s-ReProCS}($\bm{72.5}${\bf ms})} & \tiny{AltProj ($133.1$ms)} & \tiny{RPCA-GD($113.6$ms)} & \tiny{GRASTA ($18.9$ms)}  & \tiny{PCP ($240.7$ms)} \\
	\includegraphics[width=0.11\linewidth, height=1.4cm, trim={1cm, 0cm, 1cm, 0cm}, clip]{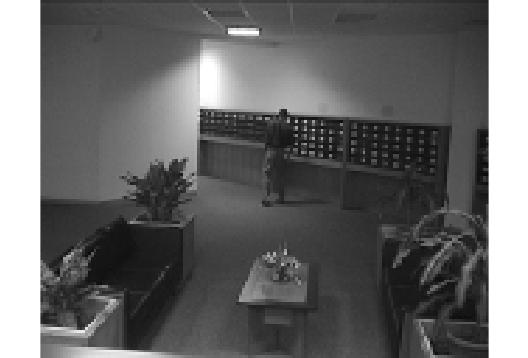}
&
	\includegraphics[width=0.11\linewidth, height=1.4cm, trim={1cm, 0cm, 1cm, 0cm}, clip]{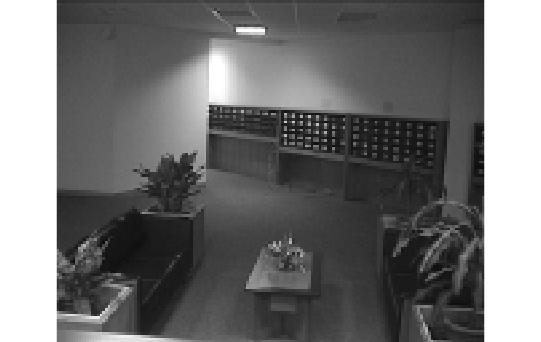}
&
	\includegraphics[width=0.11\linewidth, height=1.4cm, trim={1cm, 0cm, 1cm, 0cm}, clip]{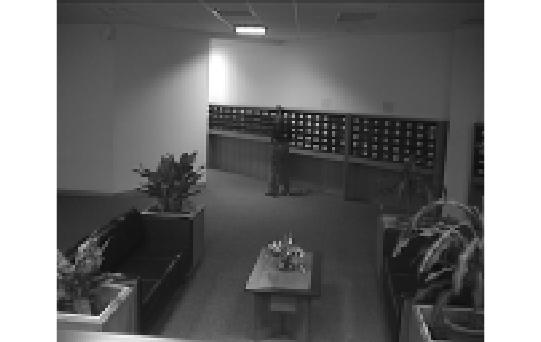}
&
	\includegraphics[width=0.11\linewidth, height=1.4cm, trim={1cm, 0cm, 1cm, 0cm}, clip]{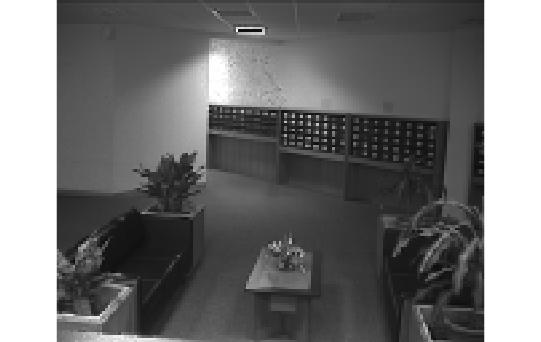}
&
	\includegraphics[width=0.11\linewidth, height=1.4cm, trim={1cm, 0cm, 1cm, 0cm}, clip]{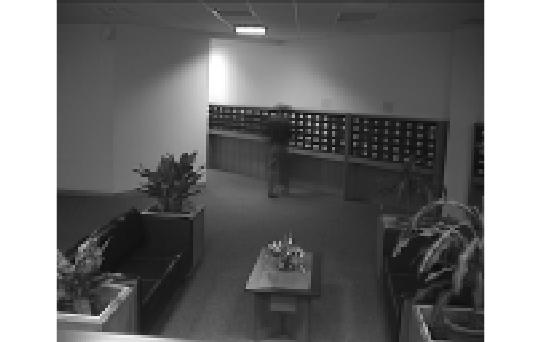}
&
	\includegraphics[width=0.11\linewidth, height=1.4cm, trim={1cm, 0cm, 1cm, 0cm}, clip]{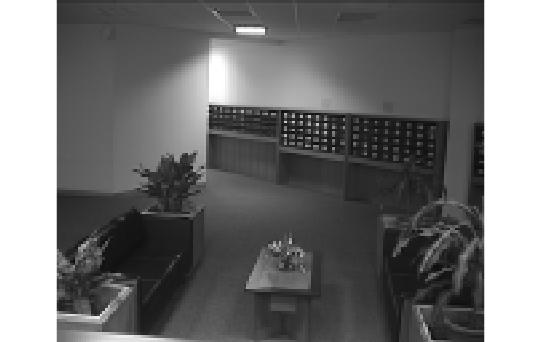}
\\    \newline

	\includegraphics[width=0.11\linewidth, height=1.4cm, trim={1.5cm, 0cm, 1.5cm, 0cm}, clip]{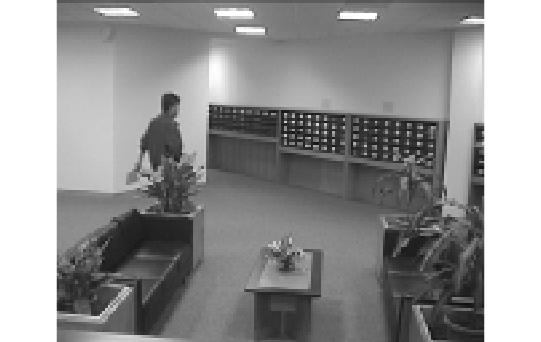}
&
	\includegraphics[width=0.11\linewidth, height=1.4cm, trim={1cm, 0cm, 1cm, 0cm}, clip]{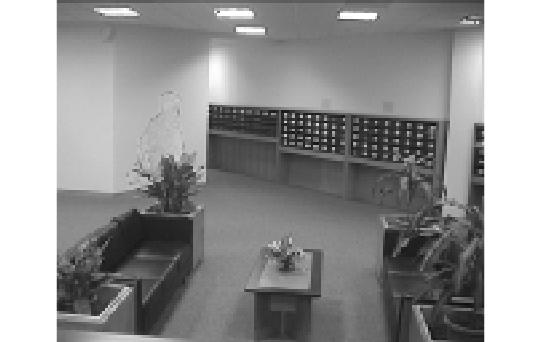}
&
	\includegraphics[width=0.11\linewidth, height=1.4cm, trim={1.45cm, 0cm, 1.5cm, 0cm}, clip]{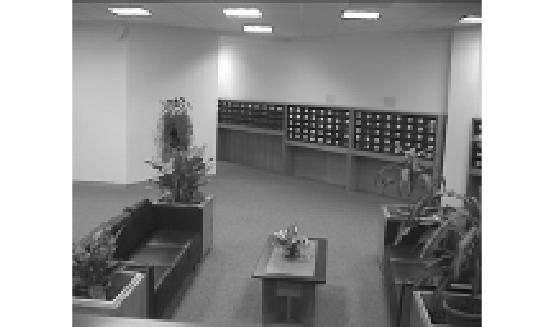}
&
	\includegraphics[width=0.11\linewidth, height=1.4cm, trim={1.4cm, 0cm, 1.4cm, 0cm}, clip]{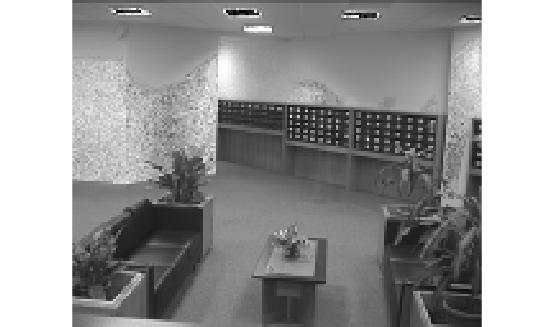}
&
	\includegraphics[width=0.11\linewidth, height=1.4cm, trim={1.3cm, 0cm, 1.5cm, 0cm}, clip]{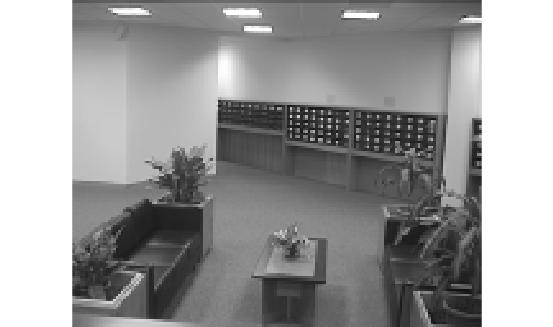}
&
	\includegraphics[width=0.11\linewidth, height=1.4cm, trim={.3cm, 0cm, .3cm, 0cm}, clip]{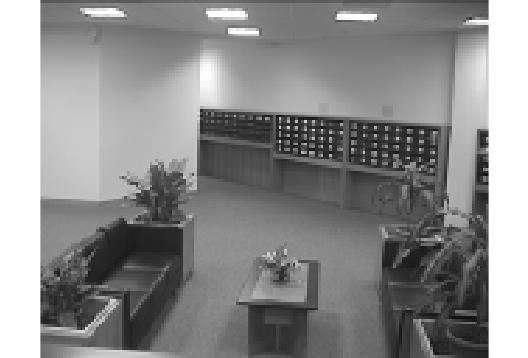}
\end{tabular}
}
\caption{{Comparison of background recovery performance is Foreground-Background Separation tasks for \texttt{MR} (first two rows), \texttt{SL} (middle two rows) and \texttt{LB} (last two rows) sequences (first two rows). The recovered background images are shown at $t = t_\train + 140, 630$ for MR, $t = t_\train + 200, 999$ for SL, and $t = t_\train + 260, 610$ for LB. Notice that for the LB sequence, all algorithms work fairly well. In the MR sequence, since the s-ReProCS is able to tolerate larger $\outfracrow$, it is able to completely remove the person. Further, only s-ReProCS background does not contain the person or even his shadow. All others do. Finally, in the SL sequence, it is demonstrated that the changing subspace model is much more appropriate for long sequences since only s-ReProCS and GRASTA are able to recognize that the background has changed. GRASTA contains some artifacts, but s-ReProCS is able to clearly isolate the person. The time taken per frame (in milliseconds) is shown in parentheses above the respective video sequence. In all the videos, notice that s-ReProCS is also faster than all algorithms with the exception of GRASTA which only works for the lobby sequence that involves very little background changes.}}
\label{fig:video_res}
\end{figure*}

\subsection{Real Data: Background Subtraction}
In this section we provide simulation results for on real videos on three benchmark datasets. For all the sequences, to implement s-ReProCS, we obtained an estimate using the AltProj algorithm. For the initialization we set $r = 40$ and the other parameters for the proposed algorithm, were set as follows. We used $\alpha = 60$, $K = 3$, $\xi_t = \|\bm{\Psi} \hat{\bm{\ell}}_{t-1}\|_2$ and $\lthres = 0.0011\lambda^-$. We found that these parameters work for most videos that we verified our algorithm on. For a more detailed empirical evaluation on real world data-sets, please see \cite{rrpcp_review}. The other state-of-the-art algorithms were implemented using the default setting. In algorithms where we are required to provide an esimate of the rank, we used $r = 40$ consistently. Additionally, for RPCA-GD we set the corruption fraction, $\alpha = 0.2$ as described in the paper. We must also mention that the performance of s-ReProCS w.r.t. original-ReProCS is very similar, but is provably fast, and needs fewer assumptions. Again, a more detailed comparison is presented in \cite{rrpcp_review}.

\texttt{Meeting Room (MR) dataset}: The meeting room sequence is a set of $1964$ images of resolution $64 \times 80$. The first $1755$ frames consists of outlier-free data and so we only consider the last $1209$ frames. Here and below, we use the first $400$ ``noisy'' frames as the training data and the algorithm parameters are set as mentioned before. This is a challenging video sequence because the color of the person and the color of the curtain are hard to distinguish. s-ReProCS algorithm is able to perform the separation at around $43$ frames per second. The recovered background images are shown in the first two rows of Fig. \ref{fig:video_res}.

\texttt{Switch Light (SL) dataset}: This dataset contains $2100$ images of resolution $120 \times 160$. The first $770$ frames are outlier free. This is a challenging sequence because there are drastic changes in the subspace as indicated in the last two rows of Fig. \ref{fig:video_res}. This causes all the batch techniques to fail. For this sequence, s-ReProCS achieves a ``test'' processing rate of $16$ frames-per-second. The recovered background images are shown in the middle two rows of Fig. \ref{fig:video_res}.

\texttt{Lobby (LB) dataset}: This dataset contains $1555$ images of resolution $128 \times 160$. The first $341$ frames are outlier free. This is a challenging sequence, as the background changes often due to illumination changes, and there are multiple objects in the foreground to detect and subtract. For this sequence, s-ReProCS achieves a ``test'' processing rate of $12$ frames-per-second. The images are shown in the last two rows of Fig. \ref{fig:video_res}.

\section{Conclusions and Future Work} \label{conclude}
We obtained the first complete guarantee for any online, streaming or dynamic RPCA algorithm that holds under weakened versions of standard RPCA assumptions, slow subspace change, and outlier magnitudes are either large or very small. Our guarantee implies that, by exploiting these extra assumptions, one can significantly weaken the required bound on outlier fractions per row. This has many important practical implications especially for video analytics. 
We analyzed a simple algorithm based on the Recursive Projected Compressive Sensing (ReProCS) framework introduced in \cite{rrpcp_perf}. The algorithm itself is simpler than other previously studied ReProCS-based methods, it is provably faster, and has near-optimal memory complexity. Moreover, our guarantee removes all the strong assumptions made by the previous two guarantees for ReProCS-based methods.

As described earlier, our current result still has limitations, some of which can be removed with a little more work. For example, it assumes a very simple model on subspace change in which only one direction can change at any given change time. Of course the changing direction could be different at different change times, and hence over a long period, the entire subspace could change. In follow-up work \cite{rrpcp_icml_trans_it}, we have studied another algorithm that removes this limitation. 
Another issue that we would like to study is whether the lower bound on outlier magnitudes can be relaxed further if we use the stronger assumption on outlier fractions per row (assume they are of order $1/r$). It may be possible to do this by borrowing the AltProj \cite{robpca_nonconvex} proof idea.

A question of practical and theoretical interest is to develop a streaming version of simple-ReProCS for dynamic RPCA. A preprint that studies a streaming algorithm for standard RPCA but only for the restrictive $\rmat=r=1$ setting is \cite{streaming_rpca}.
By streaming we mean that the algorithm makes only one pass through the data and needs storage of order exactly $nr$. Simple-ReProCS needs only a little more storage than this, however, it makes multiple passes through the data in the SVD steps. Algorithmically, streaming ReProCS is easy to develop: one can replace the projection SVD and SVD steps in the subspace update by their streaming versions, e.g., block stochastic power method. However, in order to prove that this still works (with maybe an extra factor of $\log n$ in the delay), one would need to analyze the block stochastic power method for the problems of PCA in data-dependent noise, and for its extension that assumes availability of partial subspace knowledge. %
%
Finally, as explained in \cite{rrpcp_isit15}, any guarantee for dynamic RPCA also provides a guarantee for dynamic Matrix Completion (MC) as an almost direct corollary. The reason is that MC can be interpreted as RPCA with outlier supports $\T_t$ being known. 

\appendices

\renewcommand{\thetheorem}{\thesection.\arabic{theorem}}

\section{Proof of Theorem \ref{thm1} or Corollary \ref{gen_xmin} without assuming $t_j$ known} \label{proof_auto_thm1}
The key results needed for this and later proofs -- Cauchy-Schwarz for sums of matrices, matrix Bernstein, and Vershynin's sub-Gaussian result -- are summarized in the last appendix, Appendix \ref{sec:LAandProb}. For a summary of notation used for this and later proofs, please see Table \ref{tab:not2}.

Here we prove Theorem \ref{thm1} in the general case. The main idea is explained in Sec. \ref{det_works}. Define
\begin{align*}
\that_{j-1,fin}: =  \that_{j-1} + K \alpha + \alphadel-1, \\ 
 t_{j,*}=  \that_{j-1,fin} + \left\lceil \frac{t_j - \that_{j-1,fin}}{\alpha} \right\rceil \alpha
\end{align*}
Thus, $\that_{j-1,fin}$ is the time at which the $(j-1)$-th subspace update is complete; whp, this occurs before $t_j$. Under this assumption, $t_{j,*}$ is such that $t_j$ lies in the interval $[t_{j,*}-\alpha+1,t_{j,*}]$.
Recall from the algorithm that we increment $j$ to $j+1$ at $t= \that_j+K\alpha+\alphadel:= \that_{j,fin}$. Thus, for $t \in  [t_j, \that_{j,fin})$, $\bphi =  \I -  \Phat_{\ast} \Phat_{\ast}{}'$, while for $t \in  [\that_{j,fin}, t_{j+1})$, $\bphi =  \I -  \Phat \Phat{}'$.

\begin{definition} \label{def:thm_unknown} Define the events
\ben
\item $\mathrm{Det0}:= \{\that_j = t_{j,*} \} = \{\lambda_{\max}(\frac{1}{\alpha} \sum_{t= t_{j,*}-\alpha+1}^{t_{j,*}} (\I - \Phat_{\ast} \Phat_{\ast}{}') \lhat_t \lhat_t'(\I - \Phat_{\ast} \Phat_{\ast}{}')) > \lthres\}$ and
\\ $\mathrm{Det1}:= \{\that_j = t_{j,*} + \alpha\} = \{\lambda_{\max}(\frac{1}{\alpha} \sum_{t= t_{j,*} +1}^{t_{j,*}+\alpha} (\I - \Phat_{\ast} \Phat_{\ast}{}') \lhat_t \lhat_t'(\I - \Phat_{\ast} \Phat_{\ast}{}')) > \lthres\} $,
\item $\mathrm{Proj\SVD}:=\cap_{k=1}^K  \mathrm{Proj\SVD}_k$ where $\mathrm{Proj\SVD}_k:= \{\SE([\Phat_{\ast},\Phat_{\rot,k}]) \le \zeta_{k}^+\}$,

\item $\mathrm{Del}:=\{\SE(\Phat, \P) \le \zz \}$,
\item $\mathrm{NoFalseDets}:= \{\text{for all $\J^\alpha \subseteq [\that_{j,fin}, t_{j+1})$, } \\ \lambda_{\max}(\frac{1}{\alpha} \sum_{t \in \J^\alpha} (\I - \Phat \Phat{}') \lhat_t \lhat_t'(\I - \Phat \Phat{}')) \le \lthres\}$
\item $\Gamma_{0,\ed}:= \{\SE(\Phat_*, \P_*) \le \zz \}$,
\item  $\Gamma_{j,\ed}:= \Gamma_{j-1,\ed} \cap
\big( (\mathrm{Det0} \cap \mathrm{Proj\SVD} \cap \mathrm{Del} \cap \mathrm{NoFalseDets}) \cup
(\overline{\mathrm{Det0}} \cap \mathrm{Det1} \cap \mathrm{Proj\SVD} \cap \mathrm{Del} \cap \mathrm{NoFalseDets}) \big)$.
\een
\end{definition}
Let $p_0$ denote the probability that, conditioned on $\Gamma_{j-1,\ed}$, the change got detected at $t=t_{j,*}$, i.e., let
\[
p_0:= \Pr(\mathrm{Det0}|\Gamma_{j-1,\ed}).
\]
Thus, $\Pr(\overline{\mathrm{Det0}}|\Gamma_{j-1,\ed}) = 1- p_0$. It is not easy to bound $p_0$. However, as we will see, this will not be needed.

Assume that $\Gamma_{j-1,\ed} \cap \overline{\mathrm{Det0}}$ holds. Consider the interval $\J^\alpha: = [t_{j,*}, t_{j,*}+\alpha)$. This interval starts at or after $t_j$, so, for all $t$ in this interval, the subspace has changed. For this interval, $\bpsi = \bphi  = \I - \Phat_{\ast} \Phat_{\ast}{}'$. 
Applying the last item of Theorem \ref{thm_corpca}, w.p. at least $1-12n^{-12}$,
\begin{align*}
&\lambda_{\max} \left(\frac{1}{\alpha} \sum_{t \in \J^\alpha} \bphi \lhat_t \lhat_t'\bphi \right)  \\ 
&\ge (0.97 \sin^2 \theta - 0.4 q_\rot |\sin \theta| - 0.15 \zz |\sin \theta|) \lambda_\ch
\end{align*}
where $q_\rot$ is the bound $\|(\bpsi_{\That_t}{}'\bpsi_{\That_t})^{-1} \I_{\That_t}{}' \bpsi \P_\rot\|$. Theorem \ref{thm_corpca} is applicable for the reasons given in the proof of Lemma \ref{p_evd_lem}.
Proceeding as in the proof of Lemma \ref{p_evd_lem} for $k=1$, we get that $q_\rot = 1.2( (0.1 + \zz) |\sin \theta| + \zz)$. Thus, using the bound on $\zz$, we can conclude that, w.p. at least $1-12n^{-12}$,
\begin{align*}
\lambda_{\max} \left(\frac{1}{\alpha} \sum_{t \in \J^\alpha} \bphi \lhat_t \lhat_t'\bphi \right) &\ge 0.91 \sin^2 \theta \lambda_\ch \\ 
&\ge 0.9 \sin^2 \theta \lambda^- >  \lthres
\end{align*}
and thus $\that_j = t_{j,*} + \alpha$. This follows since $\lthres = 5 \zz^2 \lambda^+ = 5 \zz^2 f \lambda^- \le  5 \zz^2 f^2 \lambda^- \le  5 (0.01 \min_j \SE(\P_{j-1},\P_j))^2 \lambda^-$ and $\sin \theta = \SE(\P_{j-1},\P_j)$.
In other words,
\[
\Pr(\mathrm{Det1} | \Gamma_{j-1,\ed} \cap \overline{\mathrm{Det0}}) \ge 1 - 12n^{-12}.
\]

Conditioned on $\Gamma_{j-1,\ed} \cap \overline{\mathrm{Det0}} \cap \mathrm{Det1}$, the first projection-SVD step is done at $t= \that_j + \alpha = t_{j,*} + 2\alpha$ and so on. We can state and prove Lemma \ref{p_evd_lem} with $\Gamma_{k}$ replaced by $\Gamma_{j,\ed} \cap \overline{\mathrm{Det0}} \cap \mathrm{Det1} \cap \mathrm{Proj\SVD}_1 \cap \mathrm{Proj\SVD}_2 \dots \mathrm{Proj\SVD}_k$ and with the $k$-th projection-SVD interval being $\J_k:=[\that_j+(k-1)\alpha, \that_j + k \alpha)$. We can state and prove a similarly changed version of Lemma \ref{del_evd} for the simple SVD based deletion step.
Applying Lemma \ref{p_evd_lem} for each $k$, and then apply Lemma \ref{del_evd},
\[
\Pr(\mathrm{Proj\SVD} \cap \mathrm{Del} |\Gamma_{j-1,\ed} \cap \overline{\mathrm{Det0}} \cap \mathrm{Det1}) \ge ( 1 - 12n^{-12})^{K+1}.
\]
We can also do a similar thing for the case when the  change is detected at $t_{j,*}$, i.e. when $\mathrm{Det0}$ holds. In this case, we replace  $\Gamma_{k}$ by $\Gamma_{j,\ed} \cap \mathrm{Det0} \cap \mathrm{Proj\SVD}_1 \cap \mathrm{Proj\SVD}_2 \dots \mathrm{Proj\SVD}_k$ and conclude that
\[
\Pr(\mathrm{Proj\SVD} \cap \mathrm{Del}|\Gamma_{j-1,\ed} \cap \mathrm{Det0}) \ge ( 1 - 12n^{-12})^{K+1}.
\]

Finally consider the $\mathrm{NoFalseDets}$ event. First, assume that $\Gamma_{j-1,\ed} \cap \mathrm{Det0} \cap \mathrm{Proj\SVD} \cap \mathrm{Del}$ holds.  Consider any interval $\J^\alpha \subseteq [\that_{j,fin}, t_{j+1})$. In this interval, $\tPhat_{(t)} = \Phat$, $\bpsi = \bphi = \I -  \Phat \Phat{}'$ and $\SE(\Phat,\P) \le \zz$.
Also, using Lemma \ref{CSlem}, $\et$ satisfies \eqref{etdef0} for $t$ in this interval. Thus, defining $\e_{l,t} = \bphi \I_{\T_t} (\bpsi_{\T_t}{}'\bpsi_{\T_t})^{-1} \I_{\T_t}{}' \bpsi \lt$, $\e_{v,t} = \bphi \I_{\T_t} (\bpsi_{\T_t}{}'\bpsi_{\T_t})^{-1} \I_{\T_t}{}' \bpsi \vt$ and $\zt = \e_{v,t} + \bphi \vt$
\begin{align*}
&\frac{1}{\alpha} \sum_{t \in \J^\alpha} \bphi \lhat_t \lhat_t{}'\bphi =  \frac{1}{\alpha} \bphi \P \left(\sum_{t \in \J^\alpha} \at \at{}' \right) \P'\bphi \\ 
&+ \frac{1}{\alpha} \sum_{t \in \J^\alpha} \bphi \lt \e_{l,t}{}' + (.)' + \frac{1}{\alpha} \sum_{t \in \J^\alpha} \bphi \lt \zt{}' + (.)' \\
& + \frac{1}{\alpha} \sum_{t \in \J^\alpha} \e_{l,t} \e_{l,t}{}' + \frac{1}{\alpha} \sum_{t \in \J^\alpha} \zt \zt{}'
\end{align*}
We can bound the first term using Vershynin's sub-Gaussian result (Theorem \ref{versh}) and the other terms using matrix Bernstein (Theorem \ref{matrix_bern}). The approach is similar to that of the proof of Lemma \ref{lem:concm}. The derivation is more straightforward in this case, since for the above interval $\|\bpsi \P\|=\|\bphi \P\| \le \zz$. The required bounds on $\alpha$ are also the same as those needed for Lemma \ref{lem:concm} to hold.
We conclude that, w.p. at least $1- 12n^{-12}$,
\begin{align*}
\lambda_{\max} &\left(\frac{1}{\alpha} \sum_{t \in \J^\alpha} \bphi \lhat_t \lhat_t'\bphi \right)  \\
& \leq \zz^2 (\lambda^+ + 0.01\lambda^+)[1 + 6 \sqrt{\rrow} f (1.2)^2 + 6 f \sqrt{\rrow}  1.2] \\
& \le 2.6 \zz^2 f \lambda^-
 < \lthres
\end{align*}
This follows since $\lthres = 5 \zz^2 \lambda^+ =  5 \zz^2 f \lambda^-$.
Since $\mathrm{Det0}$ holds, $\that_j = t_{j,*}$.
Thus, we have a total of $\lfloor \frac{t_{j+1} - t_{j,*} - K \alpha - \alphadel}{\alpha} \rfloor$ intervals $\J^\alpha$ that are subsets of $[\that_{j,fin}, t_{j+1})$. Moreover, $\lfloor \frac{t_{j+1} - t_{j,*} - K \alpha - \alphadel}{\alpha} \rfloor \le \lfloor \frac{t_{j+1} - t_j - K \alpha - \alphadel}{\alpha} \rfloor \le \lfloor \frac{t_{j+1} - t_j}{\alpha} \rfloor - (K+1)$ since $\alpha \le \alphadel$.
Thus,
\begin{align*}
\Pr(\mathrm{NoFalseDets} | \Gamma_{j-1,\ed} \cap \mathrm{Det0} \cap \mathrm{Proj\SVD} \cap \mathrm{Del}) \\ 
\ge (1 - 12n^{-12})^{\lfloor \frac{t_{j+1} - t_j}{\alpha} \rfloor - (K+1)}
\end{align*}
On the other hand, if we condition on $\Gamma_{j-1,\ed} \cap \overline{\mathrm{Det0}} \cap \mathrm{Det1} \cap \mathrm{Proj\SVD} \cap \mathrm{Del}$, then $\that_j = t_{j,*} + \alpha$. Thus,
\begin{align*}
\Pr(\mathrm{NoFalseDets} | \Gamma_{j-1,\ed} \cap \overline{\mathrm{Det0}} \cap \mathrm{Det1} \cap \mathrm{Proj\SVD} \cap \mathrm{Del}) \\ 
\ge (1 - 12n^{-12})^{\lfloor \frac{t_{j+1} - t_j}{\alpha} \rfloor - (K+2)}
\end{align*}
We can now combine the above facts to bound $\Pr(\Gamma_{j,\ed}|\Gamma_{j-1,\ed})$. Recall that $p_0:= \Pr(\mathrm{Det0}|\Gamma_{j-1,\ed})$.
Clearly, the events $(\mathrm{Det0} \cap \mathrm{Proj\SVD} \cap \mathrm{Del} \cap \mathrm{NoFalseDets})$ and $(\overline{\mathrm{Det0}} \cap \mathrm{Det1} \cap \mathrm{Proj\SVD} \cap \mathrm{Del} \cap \mathrm{NoFalseDets})$ are disjoint. Thus,
\begin{align*}
& \Pr(\Gamma_{j,\ed}|\Gamma_{j-1,\ed}) \\
& = p_0 \Pr(\mathrm{Proj\SVD} \cap \mathrm{Del} \cap \mathrm{NoFalseDets} |\Gamma_{j-1,\ed} \cap \mathrm{Det0})  \\
& + (1-p_0) \Pr(\mathrm{Det1}|\Gamma_{j-1,\ed} \cap \overline{\mathrm{Det0}}) \times \\
 &\Pr(\mathrm{Proj\SVD} \cap \mathrm{Del} \cap \mathrm{NoFalseDets} |\Gamma_{j-1,\ed}\cap \overline{\mathrm{Det0}} \cap \mathrm{Det1}) \\
& \ge p_0 ( 1 - 12n^{-12})^{K+1} (1 - 12n^{-12})^{\lfloor \frac{t_{j+1} - t_j}{\alpha} \rfloor - (K+1)} \\
& + (1-p_0) ( 1 - 12n^{-12}) ( 1 - 12n^{-12})^{K+1} \times  \\
&(1 - 12n^{-12})^{\lfloor \frac{t_{j+1} - t_j }{\alpha} \rfloor - (K+2)}  \\
& =  ( 1 - 12n^{-12})^{\lfloor \frac{t_{j+1} - t_j}{\alpha} \rfloor}
\ge ( 1 - 12n^{-12})^{t_{j+1}-t_j}.
\end{align*}
Since the events $\Gamma_{j,\ed}$ are nested, the above implies that
\begin{align*}
\Pr(\Gamma_{J,\ed}|\Gamma_{0,\ed}) &= \prod_j \Pr(\Gamma_{j,\ed}|\Gamma_{j-1,\ed}) \\
&\ge \prod_j ( 1 - 12n^{-12})^{t_{j+1}-t_j} \\
& = ( 1 - 12n^{-12})^d \ge  1 - 12d n^{-12}.
\end{align*}




\section{Proof of Theorem \ref{thm_corpca}: PCA in data-dependent noise with partial subspace knowledge}\label{proof_thm_corpca}

We prove Theorem \ref{thm_corpca} with the modification given in Remark \ref{remark_ezero}. Thus we condition on $\ezero$ defined in the remark. Recall that $\bphi:= \I - \Phat_* \Phat_*{}'$.
Let
\bea
\bphi \pa \qreq \enew \Rnew
\label{def_enew}
\eea
denote the reduced QR decomposition of $(\bphi \pa)$. Here, and in the rest of this proof, we write things in a general fashion to allow $\P_\rot$ to contain {\em more} than one direction. This makes it easier to understand how our guarantees extend to the more general case ($\P_\rot$ being an $n \times r_\ch$ basis matrix with $r_\ch > 1$) easier. 
The proof uses the following simple lemma at various places.
\begin{lem} \label{lem:simp} Assume that $\ezero$ holds. Then,
\begin{enumerate}
\item $\norm{\mot\pfix} \leq \qfix$, $\norm{\mot\pch} \leq \qfix$ and $\norm{\mot \pa} \leq \qa$
\item $\norm{\bphi \pfix} \leq \zz$, $\norm{\bphi \pch} \leq \zz$, $\norm{\bphi \pnew} \leq 1$, 
\item $\norm{\Rnew} = \norm{\bphi \pa} \leq \zz |\cos \theta| + |\sin \theta| \leq \zz  + |\sin \theta|$
\item $\sigma_{\min}(\Rnew) = \sigma_{\min}(\bphi \pa) \geq \sqrt{\sin^2\theta(1 - \zz^2) -2\zz|\sin \theta|}$
\item $\norm{\bphi \lt} \leq  2\zz \sqrt{\eta \rfix \lfp} + |\sin \theta|\sqrt{\eta \lcp} $.
\end{enumerate}
\end{lem}
\begin{proof}[Proof of  Lemma \ref{lem:simp}]
Item 1 follows because $\norm{\mot\pt_*} \le q_0 \le 2 \zz$ and $\norm{\mot\pt_*}  = \norm{\mot [\pfix,  \pch]} \geq \norm{\mot \pfix}$. Similarly,  $\norm{\mot\pt_*}  \geq \norm{\mot \pch}$. The first two claims of item 2 follow because $\norm{\bphi\pt_*} \le \zz$ and the bound on item 1. Third claim uses $\|\bphi \P_\new\| \le \|\bphi\| \|\P_\new\|=1$. The fourth claim uses triangle inequality and definition of $\P_\rot$. For Item 3, recall that $\bphi \pa \qreq \enew \Rnew$. Thus, $\sigma_i(\Rnew) = \sigma_i(\bphi \pa)$. Thus $\|\Rnew\| = \norm{\bphi \pa} \le \zz + |\sin \theta|$

{\em Item 4}: From above, $\sigma_{\min}(\Rnew) = \sigma_{\min}(\bphi \pa)$. Moreover, $\sigma_{\min}(\bphi \pa) = \sqrt{\lambda_{\min}(\pa{}' \bphi'\bphi \pa)} = \sqrt{\lambda_{\min}(\pa{}' \bphi\pa)} $. We bound this as follows. Recall that $\bphi = \I - \Phat_* \Phat_*{}'$.
\begin{align*}
&\lambda_{\min}(\pa{}' \bphi\pa)
 \ge \lambda_{\min}(\cos^2\theta \pch{}'\bphi \pch) \\ 
&+  \lambda_{\min}(\sin^2\theta \pnew{}'\bphi \pnew) \\
&- 2|\sin \theta| |\cos \theta| \ \|\pch{}'\bphi \pnew\| \\
& \ge 0 + \lambda_{\min}(\sin^2\theta \pnew{}'\bphi \pnew) - 2\zz|\sin \theta| \\
& = \sin^2\theta  \lambda_{\min}( \I - \pnew{}'\Phat_* \Phat_*{}' \pnew) - 2\zz|\sin \theta| \\
& = \sin^2\theta (1 - \|\pnew{}'\Phat_*\|^2) - 2\zz|\sin \theta| \\
&\ge \sin^2\theta (1 - \zz^2) - 2\zz|\sin \theta|
\end{align*}
The last inequality used Lemma \ref{hatswitch}.

{\em Item 5}: Using the previous items and the definition of $\eta$,
\begin{align*}
\norm{\bphi \lt} &:= \norm{ \bphi (\pfix \atf + \pa \atr)} \\
&\leq \norm{\bphi \pfix \atf} + \norm{\bphi \pa \atr} \\
&\leq \left( \zz \sqrt{\eta \rfix \lfp} + (\zz |\cos \theta | + |\sin \theta|)\sqrt{\eta \rch \lcp} \right)
\end{align*}
\end{proof}

\subsection{Proof of Theorem \ref{thm_corpca}}
\begin{proof}[Proof of Theorem \ref{thm_corpca}]
We have
\begin{align*}
\SE(\phatt, \pt) &= \SE([\phatz,  \phata], [\pfix,  \pa] )  \\
& \leq  \SE([\phatz,  \phata],  \pfix ) + \SE([\phatz,  \phata], \pa ) \\
  &  \le \zz + \SE([\phatz,  \phata],  \pa )
\end{align*}
where the last inequality used Lemma \ref{lem:simp}. Consider $ \SE([\phatz,  \phata], \pa )$.
\begin{align}
 \SE([\phatz , \phata], \pa )
 &\le  \norm{ ( \bm{I} - \phata \phata{}'  ) \enew}\norm{\Rnew} \nn \\ 
 &\le  \SE(\phata, \enew) (\zz + |\sin \theta|)
\label{SE_Prot_bnd}
\end{align}
The last inequality used Lemma \ref{lem:simp}.
To bound $\SE(\phata, \enew)$, we use the Davis-Kahan $\sin \theta$ theorem \cite{davis_kahan} given below.
\begin{theorem}[Davis-Kahan $\sin \theta$ theorem]
Consider $n \times n$ Hermitian matrices, $\bm{D}$ and $\hat{\bm{D}}$ such that
\begin{align*}
\bm{D} &= \begin{bmatrix}
\bm{E} & \bm{E}_{\perp}
\end{bmatrix} \begin{bmatrix}
\bm{A} & \bm{0} \\ \bm{0} & \A_\rest
\end{bmatrix} \begin{bmatrix}
\bm{E}{}' \\ \bm{E}_{\perp}{}'
\end{bmatrix} \nonumber \\
\hat{\bm{D}} &= \begin{bmatrix}
\bm{F} & \bm{F}_{\perp}
\end{bmatrix} \begin{bmatrix}
\bm{\Lambda} & \bm{0} \\ \bm{0} & \bm{\Lambda}_\rest
\end{bmatrix} \begin{bmatrix}
\bm{F}{}' \\ \bm{F}_{\perp}{}'
\end{bmatrix} \nonumber
\end{align*}
where $[\bm{E}, \bm{E}_{\perp}]$ and $[\bm{F}, \bm{F}_{\perp}]$ are orthogonal matrices and $\mathrm{rank}(\bm{F}) = \mathrm{rank}(\bm{E})$. Let
\[
\bm{H} = \hat{\bm{D}} - \bm{D}.
\]
If $\lambda_{\min}(\bm{A}) - \lambda_{\max}(\A_\rest) - \|\bm{H}\| > 0$ and $\mathrm{rank}(\bm{E}) = \mathrm{rank}(\bm{F})$, then,
\begin{align}\label{eq:sincor}
\norm{\left(\bm{I} - \bm{F} \bm{F}{}'\right)\bm{E}} \leq \frac{\|\bm{H}\|}{\lambda_{\min}(\bm{A}) - \lambda_{\max}(\A_\rest) - \|\bm{H}\|}.
\end{align}
\end{theorem}
To use this result to bound $\SE(\Phat_\rot, \bE_\rot)$,  let $\hat{\bm{D}} := \bm{D}_{obs} =  \frac{1}{\alpha}\sum_t \bphi \yt \yt{}' \bphi$. Its top eigenvector is $\Phat_\rot$.
We need to define a matrix $\bm{D}$ that is such that its top eigenvector is $\enew$ and the gap between its first and second eigenvalues is more than $\|\bm{H}\|$.
Consider the matrix
\begin{align*}
& \bm{D} := \enew \anew \enew{}' + \enperp \anperp \enperp{}' \text{  where}   \\
& \anew := \enew{}' \left(\frac{1}{\alpha} \sum \bphi \lt \lt{}' \bphi \right) \enew,  \\
&\anperp := \enperp{}' \left(\frac{1}{\alpha} \sum \bphi \lt \lt{}' \bphi \right) \enperp.
\end{align*}
If $\lambda_{\max}(\A_\rest) < \lambda_{\min}(\A)$, then $\enew$ is the top eigenvector of $\bm{D}$.
Moreover, if $\lambda_{\max}(\A_\rest) < \lambda_{\min}(\A) - \|\bm{H}\|$, then the gap requirement holds too. Thus, by the $\sin \theta$ theorem,
\begin{align}
\SE(\phata, \enew) &= \norm{\left(\bm{I} - \phata \phata{}'\right) \enew} \nn \\ 
&\leq \frac{\norm{\bm{H}}}{\lambda_{\min}(\bm{A}) - \lambda_{\max}(\A_\rest) - \|\bm{H}\|}
\label{SE_bnd}
\end{align}
Here again, we should point out that, in the simple case that we consider where $\P_\rot$ is a vector (only one direction changes), $\A$ is a non-negative scalar and $\lambda_{\min}(\A)=\A$. However the above discussion applies even in the general case when $r_\ch>1$.
The rest of the proof obtains high probability bounds on the terms in the above expression. $\norm{\bm{H}}=\norm{\D - \hat{\D}}$ can be bounded as follows.
\color{black}
\begin{lem} \label{lem:sin} 
Let \\  $\termoneone = \frac{1}{\alpha} \sum_t \enew \enew{}' \bphi \lt \lt{}' \bphi \enperp \enperp{}'$. Then,
\begin{align*}
\norm{\bm{H}} &\leq 2\norm{\frac{1}{\alpha}\sum\bphi  \lt \wt{}' \bphi} + 2\norm{\termoneone}  \\
& + \norm{\frac{1}{\alpha}\sum\bphi \wt \wt{}' \bphi} + 2 \norm{\frac{1}{\alpha}\sum\bphi  \lt \zt{}' \bphi} \\ 
& + \norm{\frac{1}{\alpha}\sum\bphi  \zt \zt{}' \bphi}
\end{align*}
\end{lem}

\begin{proof} [Proof of Lemma \ref{lem:sin}]
Recall that $\bm{H}=\hat{\bm{D}} - \bm{D}$. Thus
\begin{align*}
\bm{H} &= \left(\hat{\bm{D}} - \frac{1}{\alpha}\sum\bphi \lt \lt{}' \bphi \right) + \left( \frac{1}{\alpha}\sum\bphi \lt \lt{}' \bphi - \bm{D} \right) \nonumber \\
&= \left(\frac{1}{\alpha}\sum\bphi \yt \yt{}' \bphi - \frac{1}{\alpha}\sum\bphi \lt \lt{}' \bphi \right) \nonumber \\
&+ \bigg( \left( \enew\enew{}' + \enperp \enperp{}'\right) \frac{1}{\alpha}\sum\bphi \lt \lt{}' \bphi \times \\ 
&\left( \enew\enew{}' + \enperp \enperp{}'\right) - \bm{D} \bigg) \nonumber \\
&= \left(\frac{1}{\alpha}\sum\bphi \wt \lt{}' \bphi + \frac{1}{\alpha}\sum\bphi \lt \wt{}'\bphi \right) \\
&+ \left(\frac{1}{\alpha}\sum\bphi \zt \lt{}' \bphi + \frac{1}{\alpha}\sum\bphi \lt \zt{}'\bphi \right) \nonumber \\
&+ \left(\frac{1}{\alpha}\sum\bphi \wt \wt{}' \bphi + \frac{1}{\alpha}\sum\bphi \zt \zt{}' \bphi \right) \nonumber \\
&+ \left( \frac{1}{\alpha} \sum \enew \enew{}' \bphi \lt \lt{}' \bphi \enperp \enperp{}' \right) \\
&+ \left( \frac{1}{\alpha} \sum \enperp \enperp{}' \bphi \lt \lt{}' \bphi \enew \enew{}' \right) \nonumber
\end{align*}
Using triangle inequality the bound follows.
\end{proof}

The next lemma obtains high probability bounds on the above terms and the two other terms from \eqref{SE_bnd}.
\color{black}
\begin{lem}\label{lem:concm}
Assume that the assumptions of Theorem \ref{thm_corpca} with the modification given in Remark \ref{remark_ezero} hold.
Let $\epsilon_0 = 0.01 |\sin \theta|(\zz + \qa)$, $\epsilon_1 = 0.01(\qa^2 + \zz^2)$, and $\epsilon_2 = 0.01$.
For an $\alpha \geq \alpha_0 := C\eta\max \left\{f r \log n,\ \eta f^2(r + \log n) \right\}$, conditioned on $\ezero$, all the following hold w.p. at least $1-12n^{-12}$:
\ben
\item $\norm{\frac{1}{\alpha} \sum_t \bphi \lt \wt{}' \bphi} \leq \left[ \sqrt{b_0} \left(2\zz^2 f  + (\zz + |\sin\theta|)\qa \right) + \epsilon_0 \right]\lcm$,
\item  $\norm{\frac{1}{\alpha} \sum_t \bphi \wt \wt{}' \bphi} \leq \left[ \sqrt{b_0} \left(4\zz^2 f + \qa^2 \right) + \epsilon_1\right]\lcm$,
\item  $\lambda_{\min}(\bm A) \geq (\sin^2\theta(1 - \zz^2) - 2 \zz |\sin \theta|)(1 - \epsilon_2)\lcm - 2 \zz (\zz + |\sin\theta|)\epsilon_2 \lcm$,
\item  $\lambda_{\max}(\A_\rest) \leq \zz^2\lfp + \zz^2\epsilon_2\lcm$,
\item  $\norm{\termoneone} \leq \left[\zz^2 f + 2\zz^2\epsilon_2 + \zz|\sin \theta|\epsilon_2|\right] \lcm.$
\item $\norm{\frac{1}{\alpha} \sum_t \bphi \lt \zt{}' \bphi} \leq \epsilon_0 \lcp$
\item $\norm{\frac{1}{\alpha} \sum_t \bphi \zt \zt{}' \bphi} \leq  \left[ \left(4\zz^2 f + \qa^2 \right) + \epsilon_1\right]\lcm$  
\een
\end{lem}

Using Lemma \ref{lem:concm} and substituting for $\epsilon_0, \epsilon_1, \epsilon_2$, we conclude the following. Conditioned on $\ezero$, with probability at least $1 - 12n^{-12}$,
\begin{align*}
&\SE(\phata, \enew) \leq \\
& \frac{\splitfrac{2 \sqrt{b_0} \left[2\zz^2 f  + (\zz + |\sin\theta|) \qa \right]+ \sqrt{b_0} \left[ 4\zz^2 f + \qa^2\right]}{ + 2 \left[\zz^2 f + 2\zz^2\epsilon_2  + \zz|\sin \theta|\epsilon_2\right] + 4\epsilon_0 + 2\epsilon_1}}{\splitfrac{(\sin^2\theta(1 - \zz^2) - 2 \zz |\sin \theta|)(1 - \epsilon_2)}{- 2 \zz (\zz + |\sin\theta|)\epsilon_2   - (\zz^2f + \zz^2\epsilon_2) - \mathrm{numer}}},
\end{align*}
where  $\mathrm{numer}$ denotes the numerator expression.
The numerator, $\mathrm{numer}$, expression can be simplified to
\begin{align*}
\mathrm{numer} &\leq \qa \left[ 2\sqrt{\rrow}f(\zz + |\sin\theta|) + 0.04|\sin\theta| \right] \\
&+ \qa^2(\sqrt{b_0} + 0.02) +  \zz\big[ (8\sqrt{\rrow}f + 2 ) \zz f \\
&+ (2\zz + |\sin\theta|) 0.01 + 0.04|\sin\theta| + 0.02\zz  \big].
\end{align*}
Further, using $\zz f \leq 0.01 |\sin \theta|$, $\sqrt{b_0} \le 0.1$ and $\qa \leq 0.2|\sin\theta|$,
\begin{align*}
\mathrm{numer} &\leq |\sin\theta|( 0.242\qa + 0.07\zz)  + 0.12\qa^2 \\
&\leq  |\sin\theta|(0.27\qa + 0.07\zz)
\end{align*}
This can be loosely upper bounded by $0.26\sin^2\theta$. We use this loose upper bound when this term appears in the denominator. Following a similar approach for the denominator, denoted $\mathrm{denom}$,
\begin{align*}
& \mathrm{denom} \\
& \geq  \sin^2\theta \left[ 1 - \zz^2 -\frac{2\zz}{|\sin \theta|} - \frac{3\zz^2\epsilon_2}{\sin^2\theta} - \frac{\zz^2 f}{\sin^2\theta}  - \frac{\mathrm{numer}}{\sin^2\theta}\right] \nn \\
& \geq  \sin^2\theta \left[0.95   - \frac{\mathrm{numer}}{\sin^2\theta}\right] \ge 0.69\sin^2\theta
\end{align*}
Thus,
\begin{align*}
\SE(\phata, \enew) &\leq \frac{(0.27\qa + 0.07 \zz)|\sin \theta|}{0.69\sin^2\theta} \\
&\leq \frac{0.39\qa + 0.1 \zz}{|\sin\theta|},
\end{align*}
Using \eqref{SE_Prot_bnd} and  $\zz \le \zz f \le 0.01 |\sin \theta|$,
\begin{align*}
\SE(\Phat,\P_\rot) &\le (\zz + |\sin\theta|)  \frac{0.39\qa + 0.1 \zz}{|\sin\theta|} \\ 
&\le 1.01 |\sin \theta|\frac{0.39\qa + 0.1 \zz}{|\sin\theta|} \\
&\le  0.40 q_\rot + 0.11 \zz.
\end{align*}

{\em Proof of the last claim: lower bound on $\lambda_{\max}(\bm{D}_{obs})$. }
Using Weyl's inequality,
\begin{align*}
\lambda_{\max}(\bm{D}_{obs}) \ge \lambda_{\max}(\bm{D}) - \|\bm{H}\| &\ge  \lambda_{\max}(\bm{A}) - \|\bm{H}\| \\
&\ge  \lambda_{\min}(\bm{A}) - \|\bm{H}\|.
\end{align*}
Using the bounds from Lemmas \ref{lem:sin} and \ref{lem:concm} and \eqref{extra_bnds}, we get the lower bound. 
\end{proof}
\color{black}

\subsection{Proof of Lemma \ref{lem:concm}: high probability bounds on the $\sin \theta$ theorem bound terms}\label{app:lemmaproof}
\begin{proof}[Proof of Lemma \ref{lem:concm}]
Recall the definition of the event $\ezero$ from Remark \ref{remark_ezero}.
To prove this lemma, we first bound the probabilities of all the events conditioned on $\{\Phat_*,Z\}$, for values of $\{\Phat_*,Z\} \in \ezero$.
Then we use the following simple fact.
\begin{fact}
\label{simple_fact}
If
$
\Pr(\mathrm{Event}  | \{\Phat_*,Z\}) \ge p_0
$
for all $\{\Phat_*,Z\} \in \ezero$, then, 
\[
\Pr ( \mathrm{Event}  | \ezero) \ge p_0.
\]
\end{fact}
In the discussion below, we condition on $\{\Phat_*,Z\}$, for values of $\{\Phat_*,Z\}$ in $\ezero$. Conditioned on $\{\Phat_*,Z\}$, the matrices $\bE_\rot$, $\bE_{\rot,\perp}$, $\bm\Phi$, etc, are constants (not random). All the terms that we bound in this lemma are either of the form $\sum_{t \in \J^\alpha} g_1(\Phat_*,Z)  \lt \lt{}' g_2(\Phat_*,Z) $, for some functions $g_1(.),g_2(.)$, or are sub-matrices of such a term. 

Since the pair $\{\Phat_*,Z\}$ is independent of the $\lt$'s for $t \in \J^\alpha$, and these $\lt$'s are mutually independent, hence, even conditioned on $\{\Phat_*,Z\}$, the same holds: the $\lt$'s for $t \in \J^\alpha$ are mutually independent. Thus, once we condition on $\{\Phat_*,Z\}$, the summands in the terms we need to bound are mutually independent.
As a result, matrix Bernstein (Theorem \ref{matrix_bern}) or Vershynin's sub-Gaussian result (Theorem \ref{versh}) are applicable. 

\renewcommand{\zt}{\bm{Z}_t}
\emph{{Item 1}}: {\em In the proof of this and later items, we condition on $\{\Phat_*,Z\}$, for values of $\{\Phat_*,Z\}$ in $\ezero$.}

Since $\norm{\bphi} = 1$,
\begin{align*}
\norm{\frac{1}{\alpha} \sum_t \bphi \lt \wt{}'\bphi} \leq \norm{\frac{1}{\alpha} \sum_t \bphi \lt \wt{}'}.
\end{align*}
To bound the RHS above, we will apply \emph{matrix Bernstein} (Theorem \ref{matrix_bern}) with $\bm{Z}_t =  \bphi \lt \wt{}'$. As explained above, conditioned on $\{\Phat_*,Z\}$, the $\bm{Z}_t$'s are mutually independent.
We first obtain a bound on the expected value of the time average of the $\bm{Z}_t$'s and then compute $R$ and $\sigma^2$ needed by Theorem \ref{matrix_bern}.
By Cauchy-Schwarz,
\begin{align}
&\norm{\ep\left[\frac{1}{\alpha} \sum_t \bphi \lt \wt{}'  \right]}^2 = \norm{\frac{1}{\alpha} \sum_t \bphi \pt \bm{\Lambda} \pt{}' \mot{}' \mtt{}'}^2 \nonumber \\
&\overset{(a)}{\leq} \norm{\frac{1}{\alpha} \sum_t \left(\bphi \pt \bm{\Lambda} \pt{}' \mot{}'\right)\left( \mot \pt\bm{\Lambda}\pt{}'\bphi\right)} \times  \nonumber \\
 &\norm{\frac{1}{\alpha} \sum_t \mtt \mtt{}'} \nonumber \\
&\overset{(b)}{\leq}  b_0 \left[ \max_t \norm{\bphi \pt \bm{\Lambda} \pt{}' \mot{}'}^2 \right] \nonumber \\
&\leq b_0 \bigg[ \max_t \big(\norm{\bphi \pfix \lfix \pfix{}' \mot{}'} \nonumber \\ &+ \norm{\bphi \pa \lch \pa{}'\mot{}'} \big)^2\bigg] \nonumber \\
&\leq b_0 \left[\zz \qfix \lfp  + \left( \zz + |\sin\theta| \right) \qa \lcp \right]^2
\label{lt_wt_bnd}
\end{align}
where (a) follows by Cauchy-Schwarz (Theorem \ref{CSmat}) with $\bm{X}_t = \bphi \pt \bm{\Lambda} \pt{}' \mot{}'$ and $\bm{Y}_t =\mtt$, (b) follows from the assumption on $\M_{2,t}$, and the last inequality follows from Lemma \ref{lem:simp}. Using $\qch \leq 2 \zz$, 
\begin{align*}
\norm{\ep\left[\frac{1}{\alpha} \sum \bphi \lt \wt{}' \right]} &\leq \sqrt{b_0} \left[2\zz^2 \lfp  + \left( \zz + |\sin\theta|  \right) \qa\lcp \right].
\end{align*}

To compute $R$, using Lemma \ref{lem:simp} and using $q_0 \le 2 \zz$ and $q_\rot < |\sin \theta|$,
\begin{align*}
\norm{\zt} \leq \norm{\bphi \lt} \norm{\wt} &\leq \left( \zz \sqrt{\eta \rfix \lfp} + (\zz + |\sin \theta|)\sqrt{\eta \lcp} \right) \\
& \left( \qfix \sqrt{\eta r \lfp} + \qa \sqrt{\eta \lcp} \right) \\
&\leq 4\zz^2\eta r \lfp + |\sin\theta|\qa \eta \lcp \\ 
&+ 2 \zz \eta \sqrt{r \lfp \lcp}(\qa + |\sin\theta|) \\
&\leq c_1\zz|\sin\theta|\eta r \lfp + c_2|\sin\theta|\qa \eta \lcp \\
&:= R
\end{align*}
for numerical constants $c_1, c_2$.
Next we compute $\sigma^2$. Since $\wt$'s are bounded r.v.'s, we have
\begin{align*}
\norm{\frac{1}{\alpha}\sum_t\ep{ [\zt \zt{}' ]}} &= \norm{\frac{1}{\alpha} \sum_t \ep{\left[\bphi \lt \wt{}' \wt \lt{}' \bphi  \right]}} \\
& = \norm{ \frac{1}{\alpha}\ep{[ \norm{\wt}^2  \bphi \lt \lt{}' \bphi ]}} \\
&\leq \left( \max_{\wt} \norm{\wt}^2 \right)\norm{\frac{1}{\alpha}\sum_t\ep{\left[\bphi \lt \lt{}' \bphi \right]}}\\
&\leq \left(8\zz^2 \eta \rfix \lfp + 2\qa^2\eta \lcp \right) \\
& \left(2\zz^2 \lfp + \sin^2\theta\lcp \right) \\
&\leq c_1\qa^2\sin^2\theta\eta (\lcp)^2 + c_2 \zz^2 \eta \rfix \sin^2\theta \lfp \lcp \\
&:= \sigma_1^2
\end{align*}
for numerical constants $c_1$ and $c_2$. The above bounds again used $q_0 \le 2 \zz$ and $q_\rot < |\sin \theta|$.
For bounding $\norm{\frac{1}{\alpha}\sum_t \ep{\left[\zt{}' \zt \right]}}$ we get the same expression except for the values of $c_1$, $c_2$. Thus, applying matrix Bernstein (Theorem \ref{matrix_bern}) followed by Fact \ref{simple_fact},
\begin{align*}
&\Pr \bigg( \norm{\frac{1}{\alpha}\sum_t \bphi \lt \wt{}'} \\ 
&\leq \sqrt{\rrow}f \left[2\zz^2 \lfp  + (\zz + |\sin\theta|)\qa \lcp \right] + \epsilon \bigg| \ezero \bigg) \\
&\geq 1- 2n \exp \left( \frac{-\alpha}{4\max\left\{\frac{\sigma_1^2}{\epsilon^2},\ \frac{R}{\epsilon}\right\} }\right).
\end{align*}
Let $\epsilon = \epsilon_0 \lcm$ where $\epsilon_0 = 0.01 \sin\theta(\qa + \zz)$. Then, clearly,
\begin{align*}
\frac{\sigma^2}{\epsilon^2} \leq c \eta\max \{1,\ f r \} = c \eta f r , \ \text{ and } \\
\frac{R}{\epsilon} \leq c \eta\max\{ 1,\ f r \} = c \eta f  r .
\end{align*}
Hence, for the probability to be of the form $1 - 2 n^{-12}$ we require that $\alpha \geq \alpha_{(1)}$ where
\begin{align*}
\alpha_{(1)} := C \cdot  \eta  f  (r \log n)
\end{align*}
Thus, if $\alpha \ge \alpha_{(1)}$, conditioned on $\ezero$, the bound on $ \norm{\frac{1}{\alpha}\sum_t \bphi \lt \wt{}' \bphi}$ given in Lemma \ref{lem:concm} holds w.p. at least $1 - 2 n^{-12}$.

\emph{{Item 2}}:
We use Theorem \ref{matrix_bern} (matrix Bernstein) with $\bm{Z}_t:= \bphi \wt \wt{}' \bphi$. The proof approach is similar to that of the proof of item 1.
First we bound the norm of the expectation of the time average of $\zt$:
\begin{align*}
&\norm{\ep\left[ \frac{1}{\alpha} \sum \bphi \wt \wt{}' \bphi  \right]} \\
&= \norm{\frac{1}{\alpha} \sum  \bphi \mtt \mot \pt \bm{\Lambda} \pt{}' \mot{}' \mtt{}' \bphi} \\
&\leq \norm{\frac{1}{\alpha} \sum  \mtt \mot \pt \bm{\Lambda} \pt{}' \mot{}' \mtt{}'} \\
&\overset{(a)}{\leq} \bigg(\norm{\frac{1}{\alpha} \sum_t \mtt \mtt{}'} \\
& \left[ \max_t \norm{\mtt \mot \pt \bm{\Lambda} \pt{}' \mot{} (\cdot){}'}^2\right]\bigg)^{1/2} \\
&\overset{(b)}{\leq} \sqrt{b_0} \left[\max_t \norm{\mot \pt \bm{\Lambda} \pt{}' \mot{}'\mtt{}'}\right] \\
&\overset{(c)}{\leq} \sqrt{b_0} \left[ \qfix^2 \lfp + \qa^2 \lcp \right] \leq \sqrt{b_0} \left[4\zz^2 \lfp + \qa^2 \lcp \right]. \nonumber
\end{align*}
(a) follows from Cauchy-Schwarz (Theorem \ref{CSmat}) with $\bm{X}_t = \M_{2,t}$ and $\bm{Y}_t = \mot \pt \bm{\Lambda} \pt{}' \mot{}'\mtt{}'$, (b) follows from the assumption on $\M_{2,t}$, and (c) follows from Lemma \ref{lem:simp}. The last inequality used $q_0 \le 2\zz$.
To obtain $R$,
\begin{align*}
\norm{\zt} &= \norm{\bphi \wt \wt{}' \bphi} \\ 
&\leq 2\left(\norm{\bphi \Mt \pfix \atf}^2 + \norm{\bphi \Mt \pa \atr}^2\right) \\
&\leq 2\left( \qfix^2 \eta  r \lfp + \qa^2\eta \lcp \right) \\ 
&\leq 8 \zz^2 r \eta \lfp + 2 \qa^2 \eta \lcp  := R
\end{align*}
To obtain $\sigma^2$,
\begin{align*}
&\norm{\frac{1}{\alpha}\sum_t\ep{\left[\bphi\wt (\bphi\wt){}'(\bphi \wt) \wt{}' \bphi \right]}} \\
&= \norm{\frac{1}{\alpha} \sum_t \ep{\left[\bphi\wt \wt{}' \bphi \norm{\bphi\wt}^2 \right]}} \\
&\leq \left(\max_{\wt} \norm{\bphi \wt}^2 \right) \norm{\bphi \Mt \pt \bm{\Lambda} \pt{}' \Mt{}' \bphi}  \\
&\leq 2\left( \qfix^2 r \eta \lfp + \qa^2\eta \lcp \right) \left(\qfix^2 \lfp + \qa^2 \lcp\right) \\
&\leq c_1\qa^4 \eta (\lcp)^2 + c_2 \qa^2\zz^2\eta r \lfp \lcp := \sigma^2
\end{align*}
Applying matrix Bernstein (Theorem  \ref{matrix_bern}) followed by Fact \ref{simple_fact}, we have
\begin{align*}
\Pr\left( \norm{\frac{1}{\alpha} \sum_t \bphi \wt \wt{}' \bphi} \leq \sqrt{b_0} \left[4\zz^2 \lfp + \qa^2 \lcp \right] + \epsilon \bigg| \ezero  \right) \\
\geq 1- n \exp \left( \frac{- \alpha \epsilon^2}{2(\sigma^2 + R\epsilon)} \right)
\end{align*}
Let $\epsilon = \epsilon_1 \lcm$, $\epsilon_1 = 0.01 (\qa^2 + \zz^2)$. Then we get
\begin{align*}
\frac{R}{\epsilon} \leq c\eta\max\{1,\ rf\}, \ \text{ and } \  \frac{\sigma^2}{\epsilon^2} \leq c\eta\max\{1,\ rf\}.
\end{align*}
For the success probability to be of the form $1 - 2 n^{-12}$ we require $\alpha \geq \alpha_{(2)}$ where
\begin{align*}
\alpha_{(2)} := C \eta \cdot 13 f (r \log n)
\end{align*}
Thus, if $ \alpha \ge \alpha_{(2)}$, \\
$
\Pr\left(\norm{\frac{1}{\alpha} \sum_t \bphi \wt \wt{}' \bphi} \leq \left[\sqrt{b_0} \left(4\zz^2 f + \qa^2 \right) + \epsilon_1\right]\lcm | \ezero\right) \geq 1 - n^{-12}.
$

\emph{{Item 3}}:
Expanding the expression for $\anew$,
\begin{align*}
\anew &= \enew{}'\bphi \pfix \left(\frac{1}{\alpha}\sum_t \atf \atf{}'\right) \pfix{}' \bphi \enew \\
&+ \enew{}'\bphi \pa \left(\frac{1}{\alpha}\sum_t \atr \atr{}'\right) \pa{}' \bphi \enew  \\
& + \mathrm{term1} + \mathrm{term1}{}'
\end{align*}
where $\mathrm{term1} := \enew{}'\bphi \pfix \left(\frac{1}{\alpha}\sum_t \atf \atr{}'\right) \pa{}' \bphi \enew$. Since the first term on the RHS is positive semi-definite,
\begin{align}
&\lambda_{\min}(\bm A) \nn \\
&\geq \lambda_{\min}\left(\enew{}'\bphi \pa \left(\frac{1}{\alpha}\sum_t \atr \atr{}'\right) \pa{}' \bphi \enew\right) \nn \\
&+ \lambda_{\min} (\mathrm{term1} + \mathrm{term1}{}') \nn \\
&\geq \lambda_{\min}\left(\enew{}'\bphi \pa \left(\frac{1}{\alpha}\sum_t \atr \atr{}'\right) \pa{}' \bphi \enew\right) \nn  \\
& - 2 \norm{\enew{}'\bphi \pa \left(\frac{1}{\alpha}\sum_t \atr \atf{}'\right) \pfix{}' \bphi \enew}
\label{A_bnd}
\end{align}
Under our current assumptions, the $\atr$'s are scalars, so $\anew$ and $\frac{1}{\alpha}\sum_t \atr \atr{}'$ are actually scalars. However, we write things in a general fashion (allowing $\atr$'s to be $r_\ch$ length vectors), so as to make our later discussion of the $r_\ch>1$ case easier.
Using \eqref{A_bnd},
\begin{align}\label{eq:lminexp}
&\lambda_{\min}(\bm A) \nn \\ 
& \geq \lambda_{\min}\left(\enew{}'\bphi \pa \pa{}' \bphi \enew\right) \lambda_{\min}\left(\frac{1}{\alpha}\sum_t \atr \atr{}'\right)  \nn \\
& - 2 \norm{\enew{}'\bphi \pa} \norm{\pfix{}' \bphi \enew} \norm{\left(\frac{1}{\alpha}\sum_t \atr \atf{}'\right)} \nonumber \\
&\geq (\sin^2\theta(1 - \zz^2) - 2 \zz |\sin \theta| )\lambda_{\min}\left(\frac{1}{\alpha}\sum_t \atr \atr{}'\right) \nn \\
&- 2 \zz (\zz + |\sin \theta|) \norm{\left(\frac{1}{\alpha}\sum_t \atr \atf{}'\right)}.
\end{align}
The second inequality follows using $\enew{}'\bphi \pa = \enew{}' \enew \R_\rot = \R_\rot$ and Lemma \ref{lem:simp}. The first inequality is straightforward if $\atr$'s are scalars (current setting); it follows using Ostrowski's theorem \cite{hornjohnson} in the general case.%

To bound the remaining terms in the above expression, we use Vershynin's sub-Gaussian result \cite[Theorem 5.39]{vershynin} summarized in Theorem \ref{versh}. To apply this, recall that $(\at)_i$ are bounded random variables satisfying $|(\at)_i|\le \sqrt{\eta\lambda_i}$. Hence they are sub-Gaussian with sub-Gaussian norm $\sqrt{\eta\lambda_i}$ \cite{vershynin}.
Using \cite[Lemma 5.24]{vershynin}, the vectors $\at$ are also sub-Gaussian with sub-Gaussian norm bounded by $\max_i \sqrt{\eta \lambda_i} = \sqrt{\eta\lambda^+}$. Thus, applying Theorem \ref{versh} with $K \equiv \sqrt{\eta \lambda^+}$, $\epsilon \equiv \epsilon_2 \lambda_\ch$, $N \equiv \alpha$, $n_w \equiv r$, followed by using Fact \ref{simple_fact},
if
$
\alpha \geq \alpha_{(3)} := \frac{C(r\log 9 + 10 \log n)  f^2}{\epsilon_2^2},
$
then,
\begin{align}\label{eq:simpversh}
\Pr\left( \norm{ \frac{1}{\alpha} \sum_t \at \at{}' - \bm{\Lambda} }  \leq \epsilon_2\lcm \bigg| \ezero \right) \geq 1- 2n^{-12}.
%
\end{align}
We could also have used matrix Bernstein to bound $\norm{\sum_t \at \at{}'}$. However, since the $\at$'s are $r$-length vectors and $r \ll n$, the Vershynin result requires a smaller lower bound on $\alpha$.

If $\bm{B}_1$ is a sub-matrix of a matrix $\bm{B}$, then $\norm{\bm{B}_1} \le \norm{\bm{B}}$. Thus, we can also use \eqref{eq:simpversh} for bounding the norm of various sub-matrices of $\left( \frac{1}{\alpha} \sum_t \at \at{}' - \bm{\Lambda} \right)$.
 Doing this, we get 
\begin{align}
\label{eq:lmaxversh}
& \Pr\left(\lambda_{\max}\left(\frac{1}{\alpha}\sum_t \atf  \atf{}' \right) \le \lambda^+ + \epsilon_2\lcm  \bigg| \ezero \right) \nn \\
 &\geq 1 - 2n^{-12}, \\
\label{eq:lminversh_2}
& \Pr\left(\lambda_{\max}\left(\frac{1}{\alpha}\sum_t \atr \atr{}' \right) \le \lcm + \epsilon_2\lcm  \bigg| \ezero \right) \nn \\ 
&\geq 1 - 2n^{-12}, \\
\label{eq:lminversh}
& \Pr\left(\lambda_{\min}\left(\frac{1}{\alpha}\sum_t \atr \atr{}'  \right) \geq \lcm - \epsilon_2\lcm  \bigg| \ezero \right) \nn \\
 &\geq 1 - 2n^{-12}, \text{ and} \\
\label{eq:lminbern}
& \Pr\left(\norm{\frac{1}{\alpha}\sum_t \atr \atf{}' } \leq \epsilon_2\lcm   \bigg| \ezero  \right) \nn \\
 &\geq 1 - 2n^{-12}.
\end{align}
Combining \eqref{eq:lminexp}, \eqref{eq:lminversh} and \eqref{eq:lminbern}, if $\alpha \geq \alpha_{(3)}$, 
\begin{align}\label{eq:bndlmina}
\Pr\bigg( \lambda_{\min}(\bm A) \geq (\sin^2\theta(1 - \zz^2) - 2 \zz |\sin \theta|)(1 - \epsilon_2)\lcm \nn \\
- 2 \zz (\zz + |\sin\theta|)\epsilon_2 \lcm \bigg| \ezero  \bigg) \geq 1 - 4n^{-12}
\end{align}

\emph{{Item 4}}:
Recall that $\enperp{}'\bphi \pa = 0$. Thus,
\begin{align}
\lambda_{\max}(\A_\rest) &\leq \lambda_{\max}\left(\enperp{}'\bphi \pfix \pfix{}' \bphi \enperp\right) \times \nn \\ &\lambda_{\max}\left(\frac{1}{\alpha}\sum_t \atf \atf{}'\right)
\end{align}
where the last inequality follows from Ostrowski's theorem \cite{hornjohnson}. Using this and \eqref{eq:lmaxversh},  if $\alpha \geq \alpha_{(3)}$,  
\begin{align*}
\Pr\left( \lambda_{\max}(\A_\rest) \leq \zz^2\lfp + \zz^2\epsilon_2\lcm \bigg| \ezero  \right) &\geq 1 - 2n^{-12}
\end{align*}

\emph{{Item 5}}:
Recall that \\ $\termoneone~=~\frac{1}{\alpha} \sum \enew \enew{}' \bphi \lt \lt{}' \bphi \enperp \enperp{}'$. As in earlier items, we can expand this into a sum of four terms using $\lt = \pfix \atf + \pa \atr$.
Then
using $\enperp{}'\bphi \pa = 0$ and $\norm{\enew} = \norm{\enperp} = 1$, we get
\begin{align}
\norm{\termoneone} &\leq \norm{\bphi\pfix}\norm{ \pfix{}'\bphi}\lambda_{\max}\left(\frac{1}{\alpha}\sum_t \atf \atf{}'\right) \nonumber \\
&+ \norm{\bphi\pa}\norm{ \pfix{}'\bphi}\norm{\frac{1}{\alpha}\sum_t \atf \atr{}'} 
\end{align}
Using \eqref{eq:lmaxversh} and \eqref{eq:lminbern}, if $\alpha \geq \alpha_{(3)}$, w.p. at least $1 - 4n^{-12}$, conditioned on $\ezero$, 
$
\norm{\termoneone} \leq \zz^2 (\lfp + \epsilon_2\lcm) + (\zz (\zz + |\sin\theta|))\epsilon_2 \lcm \nonumber.
$
%

\renewcommand{\zt}{\bm{z}_t}

\emph{Item 6}:
Consider {$\norm{\frac{1}{\alpha} \sum_t \bphi \lt \zt{}'}$}. We will apply matrix Bernstein (Theorem \ref{matrix_bern}).
We have $\norm{\ep{\left[\frac{1}{\alpha} \sum_t \bphi \lt \zt{}'\right]}} = 0$ since $\lt$'s are independent of $\zt$'s and both are zero mean. We obtain $R$ as follows
\begin{align*}
&\norm{\bphi \lt \zt{}'} = \norm{\bphi \lt}\norm{\zt} \\
&\leq \left(\zz \sqrt{\eta r \lambda^+} + (\zz + |\sin\theta|) \sqrt{\eta \lcp} \right) b_z \\
&\leq \left(2\zz \sqrt{\eta r \lambda^+} + |\sin\theta| \sqrt{\eta \lcp} \right) \left( q_0 \sqrt{r \lambda^+} + q_\rot \sqrt{\lcp} \right) \\
&\leq 4\zz^2 \sqrt{\eta} r \lfp + |\sin\theta|\qa \sqrt{\eta} \lcp \\
&+ 2 \zz \sqrt{\eta} \sqrt{r \lfp \lcp}(\qa + |\sin\theta|) \\
&\leq c_1\zz|\sin\theta|\sqrt{\eta} r \lfp + c_2|\sin\theta|\qa \sqrt{\eta} \lcp := R
\end{align*}
for numerical constants $c_1, c_2$.
Next we compute $\sigma^2$ as follows. First consider
\begin{align*}
&\norm{\frac{1}{\alpha} \sum_t \ep\left[ \bphi \lt \zt{}' \zt \lt{}' \bphi \right]} = \norm{\frac{1}{\alpha}\ep[\|\zt\|^2 \bphi \lt \lt{}' \bphi]} \\
&\leq \left( \max_{\zt} \|\zt\|^2 \right) \norm{\frac{1}{\alpha} \sum_t \ep[ \bphi \lt \lt{}'\bphi]} \\
&\leq (8 \zz^2 r \lambda ^+ + 2 q_\rot^2 \lcp) (2\zz^2 \lambda^+ +  \sin^2\theta \lcp) \\
&\leq c_1\qa^2\sin^2\theta\eta (\lcp)^2 + c_2 \zz^2 \eta \rfix \sin^2\theta \lfp \lcp := \sigma_1^2
\end{align*}
we note here that since $b_z^2 = r \lambda_z^+$, the other term in the expression for $\sigma^2$ is the same (modulo constants) as $\sigma_1^2$. Furthermore, notice that the expressions for both $R$ and $\sigma^2$ are the same as the ones obtained in \emph{Item 1}. Thus, we use the same deviation, $\epsilon_0$ here, and hence also obtain the same sample complexity, $\alpha$; i.e., we let $\epsilon = \epsilon_0 \lcm$ where $\epsilon_0 = 0.01 \sin\theta(\qa + \zz)$, and obtain $\alpha \geq \alpha_{(1)}$ derived in item 1. 

\emph{Item 7}: This term follows in a similar fashion as Item 2, 6 and $\alpha \geq \alpha_{(2)}$ suffices.
\end{proof}
\color{black}

\renewcommand{\zt}{\bm{z}_t}
\section{Proof of Projected CS Lemma} \label{proof_CSlem}

\color{black}

\renewcommand{\znkop}{\zeta_{k-1}^+}
\color{black}
\begin{proof}[Proof of Lemma \ref{CSlem}]
The first four claims were already proved below the lemma statement.
Consider the fifth claim (exact support recovery). Recall that for any $t \in \J_k$, $\vt$ satisfies $\norm{\vt} \leq C( 2\zz \sqrt{r \lambda^+} +  \znkop \sqrt{\lcp} ) := b_{v,t}$ (for $t\in \J_1$ and $t \in \J_0$ the bounds are the same) with $C = \sqrt{\eta}$, and thus $\bt := \bpsi (\lt + \vt)$ satisfies
\begin{align*}
\norm{\bt} &= \norm{\bpsi (\lt + \vt)} \leq \norm{\bpsi\lt} + \norm{\bpsi}\norm{\vt} \\
&\leq \left(\zz \sqrt{\rfix \eta \lfp} + \zeta_{k-1}^+ \sqrt{\eta \lcp}\right) \\
&+  \sqrt{\eta} \left( 2\zz \sqrt{r \lambda^+} + \znkop \sqrt{\lcp} \right)  \\
&\leq 2\sqrt{\eta} \left(2\zz \sqrt{\rfix \lfp} + 0.5^{k-1} \cdot 0.06 |\sin\theta| \sqrt{\lcp}\right) \\
&:= b_{b,t} = 2 b_{v,t}
\end{align*}
From the lower bound on $\xmint$ in Theorem \ref{thm1} or that in Corollary \ref{gen_xmin}, $b_{b,t} < \xmint/15$. Also, we set $\xi_t = \xmint/15$. Using these facts, and $\delta_{2s} (\bpsi) \leq 0.12 < 0.15$ (third claim of this lemma), \cite[Theorem 1.2]{candes_rip} implies that
\begin{align*}
\norm{\xhat_{t,cs} - \xt} &\leq  7 \xi_t = 7\xmint/15
\end{align*}
Thus,
\begin{align*}
| (\shatcs - \st)_i | \leq \norm{\shatcs - \st} \leq 7\xmint/15 < \xmint/2
\end{align*}
We have $\omega_{supp,t} = \xmint/2$. Consider an index $i \in \Tt$. Since $|(\st)_i| \geq \xmint$,
\begin{align*}
\xmint - |(\shatcs)_i| &\le  |(\st)_i| - |(\shatcs )_i| \\ 
&\le | (\st - \shatcs )_i | < \frac{\xmint}{2}
\end{align*}
Thus, $|(\shatcs)_i| > \frac{\xmint}{2} = \omega_{supp,t}$ which means $i \in \Thatt$. Hence $\Tt \subseteq \Thatt$. Next, consider any $j \notin \Tt$. Then, $(\st)_j = 0$ and so
\begin{align*}
|(\shatcs)_j| &= |(\shatcs)_j)| - |(\st)_j| \\ 
&\leq |(\shatcs)_j -(\st)_j| \leq b_{b,t} < \frac{\xmint}{2}
\end{align*}
which implies $j \notin \Thatt$ and so $\Thatt \subseteq \Tt$. Thus $\Thatt = \Tt$.

With  $\That_t = \Tt$, the sixth claim follows easily. Since $\Tt$ is the support of $\xt$, $\xt = \I_{\Tt} \I_{\Tt}{}' \xt$, and so
\begin{align*}
\shatt &= \bm{I}_{\Tt}\left(\bpsi_{\Tt}{}'\bpsi_{\Tt}\right)^{-1}\bpsi_{\Tt}{}'(\bpsi \lt + \bpsi \st) \\
 &= \bm{I}_{\Tt}\left(\bpsi_{\Tt}{}'\bpsi_{\Tt}\right)^{-1} \I_{\Tt}{}' \bpsi(\lt + \vt) + \st
\end{align*}
since $\bpsi_{\Tt}{}' \bpsi = \I_{\Tt}' \bpsi' \bpsi = \I_{\Tt}{}' \bpsi$. Thus $\et = \shatt-\st$ satisfies \eqref{etdef0}.
Using \eqref{etdef0} and the earlier claims,
\begin{align*}
\norm{\et} \leq \norm{\left(\bpsi_{\Tt}{}'\bpsi_{\Tt}\right)^{-1}} \norm{\itt{}'\bpsi (\lt + \vt)} \\ 
\le 1.2  \left[\norm{\itt{}'\bpsi \lt} + \norm{\vt}\right]
\end{align*}
When $k=1$, $\bpsi = \I - \Phat_* \Phat_*{}'$. Thus, using \eqref{dense_bnd} and $\|\Phat_*{}' \P_\new\| \le \zz$ (follows from Lemma \ref{hatswitch}),
\begin{align*}
\norm{\itt{}'\bpsi \lt}
& \le \norm{\bpsi \P_{*,\fx}} \norm{\atf}  \\ 
&+(\norm{\bpsi \P_{*,\ch} \cos \theta} + \norm{\itt{}'\bpsi \P_{\new} \sin \theta})\norm{\atr}  \\
& \le \zz \sqrt{\eta r \lambda^+} \\ &+  \zz |\cos \theta| \sqrt{\eta \lambda_\ch} +  (0.1 + \zz) |\sin \theta| \sqrt{\eta \lambda_\ch} \\
&\le 2\zz \sqrt{\eta r \lambda^+} + 0.11 |\sin \theta| \sqrt{\eta \lambda_\ch}
\end{align*}
also, in this interval, $b_{v,t} \leq 2\zz \sqrt{\eta r \lambda^+} + 0.11 |\sin \theta| \sqrt{\eta \lambda_\ch}$  so that
$
\norm{\et} \leq 2.4 \cdot \left( 2\zz \sqrt{\eta r \lambda^+} + 0.11 |\sin \theta| \sqrt{\eta \lambda_\ch}\right)
$
When $k>1$,
$
\norm{\itt{}'\bpsi \lt} \le \norm{\bpsi \lt} \le  \zz \sqrt{\rfix \eta \lfp} + \znkop \sqrt{\eta \lcp}.
$
and the same bound holds on $b_{v, t}$ so that
$
\norm{\et} \leq 2.4 \cdot \left(\zz \sqrt{\rfix \eta \lfp} + \znkop \sqrt{\eta \lcp}\right)
$
\end{proof}

\section{Time complexity of s-ReProCS} \label{time_comp_reprocs}

The time-consuming steps of s-ReProCS are either $l_1$ minimization or the subspace update steps. Support estimation and LS steps are much faster and hence can be ignored for this discussion.
The computational complexity of $l_1$ minimization (if the best solver were used) \cite{l1_best} is the cost of multiplying the CS matrix or its transpose with a vector times $\log(1/\epsilon)$ if $\epsilon$ is the bound on the error w.r.t. the true minimizer of the program. In ReProCS, the CS matrix is of the form $\I - \Phat \Phat'$ where $\Phat$ is of size $n \times r$ or $n \times (r+1)$, thus multiplying a vector with it takes time $O(n r)$. Thus, the $l_1$ minimization complexity per frame is $O(nr \log (1/\epsilon))$, and thus the total cost for $\tmax-t_\train$ frames is $O(nr\log (1/\epsilon) (\tmax-t_\train))$.
%
%
The subspace update step consists of $(\tmax - t_\train - J \alphadel)/\alpha$ rank one SVD's on an $n \times \alpha$ matrix (for either detecting subspace change or for projection-SVD) and $J$ rank $r$ SVD's on an $n \times \alpha$ matrix (for subspace re-estimation). Thus the subspace update complexity is at most $O( n  (\tmax - t_\train) r \log(1/\epsilon))$ and the total ReProCS complexity (without the initialization step) is $O(n (\tmax-t_\train) r \log (1/\epsilon))$.

If we assume that the initialization uses AltProj, AltProj is applied to a matrix of size $n \times t_\train$ with rank $r$. Thus the initialization complexity is $O(n t_\train r^2 \log (1/\epsilon))$. If instead GD \cite{rpca_gd} is used, then the time complexity is reduced to $O(n t_\train r f \log (1/\epsilon))$. 
Treating $f$ as a constant (our discussion treats condition numbers as constants), the final complexity of s-ReProCS is $O(n \tmax r \log (1/\epsilon))$.

If s-ReProCS is used to only solve the RPCA problem (compute column span of the entire matrix $\L$), then the SVD based subspace re-estimation step can be removed. With this change, the complexity of s-ReProCS (without the initialization step) reduces to just $O(n \tmax  \log (1/\epsilon))$ since only 1-SVDs are needed. Of course this would mean a slightly tighter bound on $\outfraccol$ is required -- it will need to be less than $c/(r+J)$.

\section{Preliminaries: Cauchy-Schwarz, matrix Bernstein and Vershynin's sub-Gaussian result}\label{sec:LAandProb}

Cauchy-Schwarz for sums of matrices says the following \cite{rrpcp_perf}.
\begin{theorem} \label{CSmat}
For matrices $\bm{X}$ and $\bm{Y}$ we have
\begin{eqnarray}\label{eq:csmat}
\norm{\frac{1}{\alpha} \sum_t \bm{X}_t \bm{Y}_{t}{}'}^2 \leq \norm{\frac{1}{\alpha} \sum_t \bm{X}_t \bm{X}_t{}'} \norm{\frac{1}{\alpha} \sum_t \bm{Y}_t \bm{Y}_t{}'}
\end{eqnarray}
\end{theorem}

\renewcommand{\zt}{\bm{Z}_t}
Matrix Bernstein \cite{tail_bound}, conditioned on another r.v. $X$, says the following.
\begin{theorem}\label{matrix_bern}
Given an $\alpha$-length sequence of $n_1 \times n_2$ dimensional random matrices and a r.v. $X$
Assume the following. For all $X \in \mathcal{C}$, (i) conditioned on $X$, the matrices $\zt$ are mutually independent, (i) $\mathbb{P}(\norm{\zt} \leq R | X)  = 1$,  and (iii) $\max\left\{\norm{\frac{1}{\alpha}\sum_t \ep{\left[\zt{}'\zt | X\right]}},\ \norm{\frac{1}{\alpha}\sum_t \ep{\left[\zt\zt{}' | X\right]}}\right\} \le \sigma^2$. Then, for an $\epsilon > 0$,
\begin{align*}
&\mathbb{P}\left(\norm{\frac{1}{\alpha} \sum_t \zt} \leq \norm{\frac{1}{\alpha} \sum_t \ep{\left[\zt|X\right]}} + \epsilon\bigg|X\right) \\
&\geq 1 - (n_1 + n_2) \exp\left(\frac{-\alpha\epsilon^2}{2\left(\sigma^2 + R \epsilon\right)} \right) \ \text{for all $X \in \mathcal{C}$}.
\end{align*}

\end{theorem}
Vershynin's result for matrices with independent sub-Gaussian rows \cite[Theorem 5.39]{vershynin}, conditioned on another r.v. $X$, says the following.
\begin{theorem}\label{versh}
Given an $N$-length sequence of sub-Gaussian random vectors $\bm{w}_i$ in $\mathbb{R}^{n_w}$, an r.v $X$, and a set $\mathcal{C}$. Assume that for all $X \in \mathcal{C}$, (i) $\bm{w}_i$ are conditionally independent given $X$; (ii) the sub-Gaussian norm of $\bm{w}_i$ is bounded by $K$  
for all $i$. Let $\bm{W}:=[\bm{w}_1, \bm{w}_2, \dots, \bm{w}_N]{}'$.
Then for an $ 0 < \epsilon < 1$ we have
\begin{align}
&\mathbb{P}\left(\norm{\frac{1}{N}\bm{W}{}'\bm{W} - \frac{1}{N}\ep{\left[\bm{W}{}'\bm{W} | X \right]}} \leq \epsilon \bigg| X\right) \nn \\ 
&\geq 1 - 2\exp\left({n_w} \log 9 - \frac{c \epsilon^2 N}{4K^4}\right)  \ \text{for all $X \in \mathcal{C}$}.
\end{align}
\end{theorem}

\bibliographystyle{IEEEbib} 
\bibliography{tipnewpfmt_kfcsfullpap}

\begin{thebibliography}{10}

\bibitem{rrpcp_isit18}
Praneeth Narayanamurthy and Namrata Vaswani,
\newblock ``Provable dynamic robust pca or robust subspace tracking,''
\newblock in {\em 2018 IEEE International Symposium on Information Theory
  (ISIT)}, 2018, pp. 376--380.

\bibitem{rpca}
E.~J. Cand{\`e}s, X.~Li, Y.~Ma, and J.~Wright,
\newblock ``Robust principal component analysis?,''
\newblock {\em J. ACM}, vol. 58, no. 3, 2011.

\bibitem{rpca2}
V.~Chandrasekaran, S.~Sanghavi, P.~A. Parrilo, and A.~S. Willsky,
\newblock ``Rank-sparsity incoherence for matrix decomposition,''
\newblock {\em SIAM Journal on Optimization}, vol. 21, 2011.

\bibitem{rpca_zhang}
D.~Hsu, S.~M. Kakade, and T.~Zhang,
\newblock ``Robust matrix decomposition with sparse corruptions,''
\newblock {\em IEEE Trans. Info. Th.}, Nov. 2011.

\bibitem{rrpcp_perf}
C.~Qiu, N.~Vaswani, B.~Lois, and L.~Hogben,
\newblock ``Recursive robust pca or recursive sparse recovery in large but
  structured noise,''
\newblock {\em IEEE Trans. Info. Th.}, pp. 5007--5039, August 2014.

\bibitem{robpca_nonconvex}
P.~Netrapalli, U~N Niranjan, S.~Sanghavi, A.~Anandkumar, and P.~Jain,
\newblock ``Non-convex robust pca,''
\newblock in {\em NIPS}, 2014.

\bibitem{rrpcp_aistats}
J.~Zhan, B.~Lois, H.~Guo, and N.~Vaswani,
\newblock ``{Online (and Offline) Robust PCA: Novel Algorithms and Performance
  Guarantees},''
\newblock in {\em Intnl. Conf. Artif. Intell. Stat. (AISTATS)}, 2016.

\bibitem{rpca_gd}
X.~Yi, D.~Park, Y.~Chen, and C.~Caramanis,
\newblock ``Fast algorithms for robust pca via gradient descent,''
\newblock in {\em NIPS}, 2016.

\bibitem{rrpcp_tsp}
H.~Guo, C.~Qiu, and N.~Vaswani,
\newblock ``An online algorithm for separating sparse and low-dimensional
  signal sequences from their sum,''
\newblock {\em IEEE Trans. Sig. Proc.}, vol. 62, no. 16, pp. 4284--4297, 2014.

\bibitem{candes_rip}
E.~Candes,
\newblock ``The restricted isometry property and its implications for
  compressed sensing,''
\newblock {\em C. R. Math. Acad. Sci. Paris Serie I}, 2008.

\bibitem{rrpcp_icml_trans_it}
P.~Narayanamurthy and N.~Vaswani,
\newblock ``Nearly optimal robust subspace tracking,''
\newblock {\em arxiv:1712.06061 under review for IEEE Trans. Info Theory},
  2017.

\bibitem{rrpcp_isit15}
B.~Lois and N.~Vaswani,
\newblock ``Online matrix completion and online robust pca,''
\newblock in {\em IEEE Intl. Symp. Info. Th. (ISIT)}, 2015.

\bibitem{zhan_pcp_jp}
J.~Zhan and N.~Vaswani,
\newblock ``Robust pca with partial subspace knowledge,''
\newblock {\em IEEE Trans. Sig. Proc.}, July 2015.

\bibitem{xu_nips2013_1}
J.~Feng, H.~Xu, and S.~Yan,
\newblock ``Online robust pca via stochastic optimization,''
\newblock in {\em NIPS}, 2013.

\bibitem{grass_undersampled}
J.~He, L.~Balzano, and A.~Szlam,
\newblock ``Incremental gradient on the grassmannian for online foreground and
  background separation in subsampled video,''
\newblock in {\em IEEE Conf. on Comp. Vis. Pat. Rec. (CVPR)}, 2012.

\bibitem{Li03anintegrated}
Y.~Li, L.~Xu, J.~Morphett, and R.~Jacobs,
\newblock ``An integrated algorithm of incremental and robust pca,''
\newblock in {\em IEEE Intl. Conf. Image Proc. (ICIP)}, 2003, pp. 245--248.

\bibitem{ipca_weightedand}
D.~Skocaj and A.~Leonardis,
\newblock ``Weighted and robust incremental method for subspace learning,''
\newblock in {\em IEEE Intl. Conf. Comp. Vis. (ICCV)}, Oct 2003, vol.~2, pp.
  1494 --1501.

\bibitem{rmc_gd}
Y.~Cherapanamjeri, K.~Gupta, and P.~Jain,
\newblock ``Nearly-optimal robust matrix completion,''
\newblock {\em ICML}, 2016.

\bibitem{rrpcp_icml}
Praneeth Narayanamurthy and Namrata Vaswani,
\newblock ``Nearly optimal robust subspace tracking,''
\newblock in {\em Intl. Conf. Machine Learning (ICML)}, 2018, pp. 3701--3709.

\bibitem{rrpcp_merop}
P.~Narayanamurthy and N.~Vaswani,
\newblock ``{A Fast and Memory-Efficient Algorithm for Robust PCA (MERoP)},''
\newblock in {\em IEEE Intl. Conf. Acoustics, Speech, Sig. Proc. (ICASSP)},
  2018.

\bibitem{rrpcp_review}
N.~Vaswani, T.~Bouwmans, S.~Javed, and P.~Narayanamurthy,
\newblock ``Robust subspace learning: Robust pca, robust subspace tracking and
  robust subspace recovery,''
\newblock {\em IEEE Signal Proc. Magazine}, July 2018.

\bibitem{corpca_nips}
N.~Vaswani and H.~Guo,
\newblock ``Correlated-pca: Principal components' analysis when data and noise
  are correlated,''
\newblock in {\em Adv. Neural Info. Proc. Sys. (NIPS)}, 2016.

\bibitem{pca_dd}
N.~Vaswani and P.~Narayanamurthy,
\newblock ``Finite sample guarantees for pca in non-isotropic and
  data-dependent noise,''
\newblock in {\em Allerton 2017, long version at arXiv:1709.06255}, 2017.

\bibitem{past}
B.~Yang,
\newblock ``Projection approximation subspace tracking,''
\newblock {\em IEEE Trans. Sig. Proc.}, pp. 95--107, 1995.

\bibitem{past_conv}
B.~Yang,
\newblock ``Asymptotic convergence analysis of the projection approximation
  subspace tracking algorithms,''
\newblock {\em Signal Processing}, vol. 50, pp. 123--136, 1996.

\bibitem{adaptivesigproc_book}
T.~Adali and S.~Haykin, Eds.,
\newblock {\em Adaptive Signal Processing: Next Generation Solutions},
\newblock Wiley \& Sons, 2010.

\bibitem{grouse}
L.~Balzano, B.~Recht, and R.~Nowak,
\newblock ``{Online Identification and Tracking of Subspaces from Highly
  Incomplete Information},''
\newblock in {\em Allerton Conf. Comm., Control, Comput.}, 2010.

\bibitem{petrels}
Y.~Chi, Y.~C. Eldar, and R.~Calderbank,
\newblock ``Petrels: Parallel subspace estimation and tracking by recursive
  least squares from partial observations,''
\newblock {\em IEEE Trans. Sig. Proc.}, December 2013.

\bibitem{local_conv_grouse}
L.~Balzano and S.~Wright,
\newblock ``Local convergence of an algorithm for subspace identification from
  partial data,''
\newblock {\em Found. Comput. Math.}, vol. 15, no. 5, 2015.

\bibitem{rrpcp_proc}
N.~Vaswani and P.~Narayanamurthy,
\newblock ``Static and dynamic robust pca and matrix completion: A review,''
\newblock {\em Proceedings of the IEEE (Special Issue on Rethinking PCA for
  Modern Datasets)}, August 2018.

\bibitem{tail_bound}
J.~A. Tropp,
\newblock ``User-friendly tail bounds for sums of random matrices,''
\newblock {\em Found. Comput. Math.}, vol. 12, no. 4, 2012.

\bibitem{vershynin}
R.~Vershynin,
\newblock ``Introduction to the non-asymptotic analysis of random matrices,''
\newblock {\em Compressed sensing}, pp. 210--268, 2012.

\bibitem{davis_kahan}
C.~Davis and W.~M. Kahan,
\newblock ``The rotation of eigenvectors by a perturbation. iii,''
\newblock {\em SIAM J. Numer. Anal.}, vol. 7, pp. 1--46, Mar. 1970.

\bibitem{streaming_rpca}
U.~N. Niranjan and Y.~Shi,
\newblock ``Streaming robust pca,''
\newblock 2016.

\bibitem{hornjohnson}
R.~Horn and C.~Johnson,
\newblock {\em Matrix Analysis},
\newblock Cambridge Univ. Press, 1985.

\bibitem{l1_best}
Lin Xiao and Tong Zhang,
\newblock ``A proximal-gradient homotopy method for the l1-regularized
  least-squares problem,''
\newblock in {\em ICML}, 2012.

\end{thebibliography}

\end{document}